\documentclass[reprint,twocolumn,amsmath,amssymb,amsthm, dsfont,prx,nobibnotes,nofootinbib]{revtex4-2}
\usepackage{times}
\usepackage[colorlinks=true,citecolor=blue,linkcolor=blue]{hyperref}

\usepackage{graphicx}
\usepackage{blkarray}
\usepackage{mathrsfs}
\usepackage{dcolumn}
\usepackage{bm}
\usepackage{hyperref}
\usepackage{tikz}
\usepackage{hhline}
\usepackage{float}
\usepackage{scalerel}
\usepackage{amsfonts}
\usepackage{bbm}
\usepackage{color}
\usepackage{tabularx}
\usepackage{multirow}
\usepackage{enumitem}
\usepackage{amsthm}
\usepackage{lipsum}
\usepackage{tikz-cd} 
\usepackage{subfig}
\usepackage{mathtools}
\usepackage{CJKutf8}
\usepackage{caption}
\usepackage{adjustbox}

\usepackage{dashrule}

\usepackage{mathabx}

\usepackage{import}

\usetikzlibrary{decorations.markings}
\usetikzlibrary{3d} 
\usetikzlibrary{calc}

\theoremstyle{definition}

\newtheorem{thm}{Theorem}
\newtheorem{lemma}[thm]{Lemma}

\newcommand{\ee}{\mathsf{e}}
\newcommand{\mm}{\mathsf{m}}
\newcommand{\ff}{\mathsf{f}}
\newcommand{\Ker}{\text{Ker\,}}
\newcommand{\Imaa}{\text{Im\,}}
\newcommand{\ra}{\rightarrow}
\newcommand{\supp}[1]{\text{Supp}(#1)}
\newcommand{\sgn}[1]{\text{sgn}(#1)}

\newcommand{\NT}[1]{ { \color{purple} (NT: {}) }}
\newcommand{\SZ}[1]{ { \color{orange} (SZ: {}) }}
\newcommand{\TW}[1]{ { \color{brown} (TW: {}) }}
\makeatletter
\def\l@subsubsection#1#2{}
\makeatother

\begin{document}


\setcounter{secnumdepth}{3}

\title{Coupled-Layer Codes: Beyond Quantum Product Constructions}

\author{Shuyu Zhang
(\begin{CJK*}{UTF8}{gbsn}张舒予\end{CJK*})}
\affiliation{C. N. Yang Institute for Theoretical Physics and Department of Physics and Astronomy,  Stony Brook University, Stony Brook, NY 11794, USA}

\author{Tzu-Chieh Wei (\begin{CJK*}{UTF8}{bsmi}魏子傑\end{CJK*})}
\affiliation{C. N. Yang Institute for Theoretical Physics and Department of Physics and Astronomy,  Stony Brook University, Stony Brook, NY 11794, USA}

\author{Nathanan Tantivasadakarn}
\affiliation{C. N. Yang Institute for Theoretical Physics and Department of Physics and Astronomy,  Stony Brook University, Stony Brook, NY 11794, USA}

\date{\today}
\begin{abstract}
Product codes are an important class of quantum error-correcting codes constructed from multiple input codes, which can give rise to asymptotically good quantum low-density parity check codes. In previous work, we showed how the product between two codes can be physically implemented by coupling layers of the first code using checks of the second code. In this work, we further unify product code constructions with coupled-layer constructions of phases of matter by introducing \emph{coupled-layer codes}. The essential strategy is to condense general excitations created by Pauli operators among multiple decoupled layers of the first code. The condensation is specified by an excitation algebra, which encodes the excitations, along with an algebra-preserving map. This coupling between layers generalizes the notion of gauging in physics as well as the mapping cone in homological algebra, and can be used to produce non-CSS codes. As examples, we show how coupled-layer codes reproduce the X-cube and Chamon models. We further generalize the balanced product code by allowing a unitary transformation in addition to a group action by free permutation and describe its corresponding coupled-layer construction. In particular, we show how balancing by a ZX-duality can reproduce non-CSS codes such as the fermionic toric code and the 3-fermion Walker-Wang model in 3D.

\end{abstract}

\maketitle

\tableofcontents

\section{Introduction}

Quantum product constructions provide a powerful route for building quantum error-correcting codes from simpler constituent codes. They include the tensor and balanced product codes, and have played an important role in the search for quantum low-density parity-check (qLDPC) codes with favorable asymptotic parameters \cite{Tillich:2013esj,Bravyi14,hastings2021fiber,panteleev2021quantum,breuckmann2021balanced,Panteleev_ACM2022}.

 At the same time, product constructions have close connections with topological phases of matter. In particular, familiar topological and fracton phases of matter can often be understood as arising from lower-dimensional systems via the coupled-layer construction, in which decoupled copies of a phase are coupled between layers to produce a new phase, often in a higher dimension \cite{Sondhi01,Kane02,TeoKane14,wang2013boson,fidkowski2013non,JianQi14,Vijay17,MaLakeChenHermele2017,Halasz17,PremHuangSongHermele2019,schmitz2019distilling,designer,Wen20,Williamson21,SullivanPlanarpstring,Wang22,Liu23,luo2023gapped,Garre-Rubio24,Cuiper2025systematic,gorantla2025string,gorantla2026gauging,Liu26}. In our previous work \cite{zhang2026coupled}, we showed that the tensor and balanced products of CSS codes admit such a realization. One begins with layers of a first code, indexed by qubits of a second code, and couples the layers according to the latter's stabilizers. 

This coupled-layer viewpoint of product codes naturally raises a broader question: which types of error-correcting codes can be obtained by coupling layers when one is not restricted to the product constructions or to condensates generated by single Pauli operators? Indeed, there are known stabilizer codes in the literature which do not admit such a product construction, but are known to have a coupled-layer construction. Examples include the X-cube model \cite{vijay2016fracton,Vijay17,MaLakeChenHermele2017,Liu23} and the fermionic toric code \cite{2023ScPP...14...65B}. The latter example is also not a product code for the simple reason that it is not a CSS code. Therefore, addressing this question is twofold. First, it can be used to extend product-code techniques beyond CSS settings, giving rise to potentially new types of good qLDPC codes. Second, it gives an organizing principle and provides new insights into known and new topological and fracton phases of matter realizable by Pauli stabilizer codes. 

In this work, we introduce \emph{coupled-layer codes}, a family of constructions that generalizes the coupled-layer realization of quantum product codes. Coupled-layer codes are Pauli stabilizer codes formed from two classical or quantum codes. The central spirit of the coupled-layer construction is still prevalent throughout our new constructions, in the sense that the second code is used as a recipe for gluing copies of the first code together. The fact that the stabilizers of the second code commute guarantees these couplings are compatible and can be performed simultaneously. However, the way in which they are glued together is also determined by certain additional input data, and thus can result in a different output code than the usual tensor or balanced product codes. 

We present a family of constructions that generalizes the coupled-layer realization of quantum product codes in two directions. First, for CSS codes, we formulate couplings that condense general excitations and relate them to mapping cones, partial gauging, and lattice-surgery-type operations \cite{hastings2021quantum, cowtan2024ssip,ide2025fault,cowtan2026parallel,yuan2026unified,cowtan2025fast,chang2026constant,zheng2026logical}. We show how the X-cube model and a 3-fold product generalizing X-cube~\cite{tan2025fracton,RakovszkyLDPC2} is captured by this construction. We then formulate a more general coupled-layer construction in terms of an excitation algebra and an algebra-preserving map. Because this formulation does not require a CSS chain-complex description, it can also produce non-CSS stabilizer codes. In particular, we show how the Chamon model~\cite{chamon2005quantum} and more generally quantum XYZ product codes~\cite{2022Quant...6..766L} fit into this framework.

Second, we generalize the notions of quotient and balancing in the context of balanced product codes by allowing the relevant group action to combine permutations with onsite unitaries. In particular, when the onsite unitary is given by a Hadamard, the resulting balanced-product can generate non-CSS codes. We show how the three-dimensional fermionic toric code~\cite{2023ScPP...14...65B} and the 3-fermion Walker-Wang model~\cite{walker20123+, Burnell14, 2023CMaPh.398..469H} can be obtained in this way. Finally, we combine these ingredients in a balanced coupled-layer construction that accommodates both general excitation condensates and group actions that involve onsite unitaries.

The remainder of the paper is organized as follows. In Sec.~\ref{sec:review} we review the coupled-layer construction of the tensor product between two CSS codes, and its relation to gauging. To prepare for the next section, we also review the mapping cone, which has been recently used in the context of code surgery in qLDPC codes \cite{hastings2021quantum, cowtan2024ssip,ide2025fault,cowtan2026parallel,yuan2026unified,cowtan2025fast,chang2026constant,zheng2026logical}. We show how the mapping cone is related to partial gauging, as well as an equivalent formulation purely in terms of operator algebras. In Sec.~\ref{sec:cssclc} we define the CSS coupled-layer code using the mapping cone. We present its coupled-layer construction, as well as a related subsystem code and concatenated code based on this construction. We illustrate this framework using the X-cube model and a 4D code as examples. In Sec.~\ref{sec:clc} we define the coupled-layer code, which is a generalization of the CSS coupled-layer code to the non-CSS scenario. We illustrate this framework using the quantum XYZ product, which includes the Chamon model. In Sec.~\ref{sec:gqc} we define a generalization of code quotient by a group action. Here the action is not only a permutation of qubits and checks, but also contains a transversal onsite unitary. We study the case where the unitary is the Hadamard, and we refer to the quotient as $\ee-\mm$ quotient. In Sec.~\ref{sec:embp} we present a variant of the balanced product construction using the $\ee-\mm$ quotient,  which we refer to as the $\ee-\mm$ balanced product. The coupled-layer construction of $\ee-\mm$ balanced product is presented. We demonstrate the framework using the 3D fermionic TC, and a family of stabilizer codes containing the 3-fermion Walker-Wang model. In Sec.~\ref{sec:bclc} we define the balanced coupled-layer code. This code is an attempt to unify all the ideas from previous sections. Sec.~\ref{sec:conclusion}, we make concluding remarks, as well as pointing out further directions. For completeness of the work, we provide some technical details in the appendices. In Appendix~\ref{sec:CSSlogical} we enumerate the logical operators of the CSS coupled-layer code under moderate assumptions. In Appendix~\ref{sec:sd} we show the relation between $\ee-\mm$ quotient/balanced product with an operation called symplectic doubling. In Appendix~\ref{app:balancedcommute} we prove the code switching in the balanced coupled-layer code is well-defined.

\section{Review}\label{sec:review}

In this section, we review the coupled-layer construction in \cite{zhang2026coupled} and discuss its relation to gauging. We then introduce the mapping cone formalism which we will use to generalize the tensor product code to the coupled-layer code.

\subsection{Coupled-Layer Construction of the Tensor Product code}\label{sec:review_cltp}

A CSS code can be formulated via a three-term chain complex $A\ra Q \ra B$, where $A$, $B$ and $Q$ are $\mathbb{Z}_2$ vector spaces spanned by $X$-stabilizers, $Z$-stabilizers and qubits respectively. For each basis element, or vector $a \in A$, we denote its corresponding $X$-stabilizer $\bm{a} = \prod_{q\in a} X_q$ in bold. Similarly, for each $b \in B$ the corresponding $Z$-stabilizer is $\bm{b}=\prod_{q\in b} Z_q$. These operators act on the Hilbert space $\bm{Q}$ which is a tensor product of $n$ qubits.  More generally, we will also denote the stabilizers of the corresponding vectors in bold.

Given two cochain complexes over $\mathbb{Z}_2$ ($C^n$, $\delta_C^i : C^i \ra C^{i+1}$) and ($D^n$, $\delta_D^j : D^j \ra D^{j+1}$), we can form their tensor product $(C\otimes D)^n := \oplus_{i+j = n} C^i\otimes D^j$ with boundary maps $\delta_{C\otimes D}^n = \sum_{i+j = n} \delta_C^i\otimes id +  id \otimes \delta_D^j$. In the case of two CSS codes CSS$_i: A_i\ra Q_i \ra B_i$, $i = 1,2$, their tensor product is a length-five cochain complex. We always place the qubits of the code at degree zero in the ccohain complex, so the product code corresponds to the middle three terms in this cochain complex, as shown in Figure~\ref{fig:productcomplex}.

\begin{figure}
\begin{tikzpicture}[scale=0.65, every node/.style={scale=0.8}]
    \def\LL{-6}
    \def\L{-3}
    \def\C{0}
    \def\R{3}
    \def\RR{6}
    \def\shadedwidth{60}
    \def\shadedheight{160}
    \def\H{1}
    \def\HH{2}

    \begin{scope}
        \node[rounded corners=9pt, fill=red!20,  minimum width=\shadedwidth, minimum height=\shadedheight] at (\L,-1) {};
        \node[rounded corners=9pt, fill=black!10, minimum width=\shadedwidth, minimum height=\shadedheight] at (\C,-1) {};
        \node[rounded corners=9pt, fill=blue!20, minimum width=\shadedwidth, minimum height=\shadedheight] at (\R,-1) {};
    \end{scope}

    \node[font=\bfseries\Huge, text=red!80!black]  (A) at (\L, -3) {$\mathit{A}$};
    \node[font=\bfseries\Huge, text=black!75]      (Q) at (\C,-3) {$\mathit{Q}$};
    \node[font=\bfseries\Huge, text=blue!80!black] (B) at (\R, -3) {$\mathit{B}$};
    \draw[very thick,->] (A) -- (Q);
    \draw[very thick,->] (Q) -- (B);

    \node[font=\Large, text=red!80!black,  anchor=north] at (\L, -3.5) {$X$-checks};
    \node[font=\Large, text=black!80,      anchor=north] at (\C, -3.5) {Qubits};
    \node[font=\Large, text=blue!80!black, anchor=north] at (\R, -3.5) {$Z$-checks};

    \node[font=\Large, text=red!80!black]  at (\L, 0) {$\oplus$};
    \node[font=\Large, text=black]      at (\C, \H) {$\oplus$};
    \node[font=\Large, text=black]      at (\C, -\H) {$\oplus$};
    \node[font=\Large, text=blue!80!black] at (\R, 0) {$\oplus$};

    \node (AA) at (\LL, 0) {$A_{1}\!\otimes\! A_{2}$};
    \node (BB) at (\RR, 0) {$B_{1}\!\otimes\! B_{2}$};

    \node[text=red!80!black]  (A1Q2) at (\L, \H) {$A_{1}\otimes Q_{2}$};
    \node[text=red!80!black]  (Q1A2) at (\L, -\H) {$Q_{1}\otimes A_{2}$};

    \node[text=black!80]      (A1B2) at (\C, \HH) {$A_{1}\otimes B_{2}$};
    \node[text=black!80]      (Q1Q2) at (\C, 0) {$Q_{1} \otimes Q_{2}$};
    \node[text=black!80]      (B1A2) at (\C, -\HH) {$B_{1}\otimes A_{2}$};

    \node[text=blue!80!black] (Q1B2) at (\R, \H) {$Q_{1}\otimes B_{2}$};
    \node[text=blue!80!black] (B1Q2) at (\R, -\H) {$B_{1}\otimes Q_{2}$};

    \draw[->] (AA) -- (A1Q2);
    \draw[->] (AA) -- (Q1A2);

    \draw[->] (A1Q2) -- (A1B2);
    \draw[->] (A1Q2) -- (Q1Q2);
    \draw[->] (Q1A2) -- (Q1Q2);
    \draw[->] (Q1A2) -- (B1A2);
    \draw[->] (Q1Q2) -- (Q1B2);

    \draw[black, ->] (A1B2) -- (Q1B2);
    \draw[black, ->] (Q1Q2) -- (B1Q2);
    \draw[black, ->] (B1A2) -- (B1Q2);

    \draw[black, ->] (Q1B2) -- (BB);
    \draw[black, ->] (B1Q2) -- (BB);
\end{tikzpicture}
\caption{The tensor product complex}
\label{fig:productcomplex}
\end{figure}
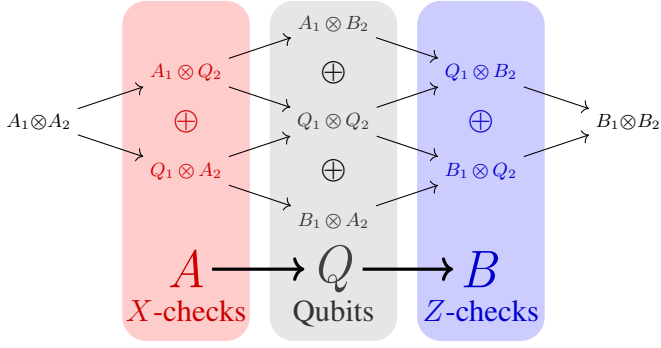

The product code can be constructed in the following way. For each qubit $q_2 \in Q_2$ in CSS$_2$, introduce a copy of CSS$_1$. For each $X$-check $a_2\in A_2$ ($Z$-check $b_2\in B_2$), introduce a layer of $Z$-ancillas ($X$-ancillas) with qubits associated to checks $b_1 \in B_1$ ($a_1\in A_1$). The stabilizer group is
\begin{align*}
    \mathcal{S}_0 = \langle\,
    \bm{a_1}^{(q_2)},\,
    \bm{b_1}^{(q_2)},\,
    X_{a_1,b_2},\,
    Z_{b_1,a_2}
    \,\rangle,
\end{align*}
where $\bm{a_1}^{(q_2)}$ and $\bm{b_1}^{(q_2)}$ are CSS$_1$ stabilizers in layer $q_2$.

We then perform a code switching by enforcing the following stabilizers
\vspace{-0.5em}
\begin{itemize}
    \item For each $q_1 \in Q_1$ and $a_2 \in A_2$, add $X$-stabilizer
    \begin{align*}
        \bm{\alpha}(q_1,a_2) := \prod_{b_1\ni q_1} X_{b_1,a_2}\prod_{q_2\in a_2}X_{q_1,q_2}.
    \end{align*}
    \vspace{-1.5em}
    \item For each $q_1 \in Q_1$ and $b_2 \in B_2$, add $Z$-stabilizer
    \begin{align*}
        \bm{\beta}(q_1,b_2) := \prod_{a_1\ni q_1} Z_{a_1,b_2} \prod_{q_2\in b_2} Z_{q_1,q_2}.
    \end{align*}
\end{itemize}
\vspace{-0.5em}
After the code switching, the commuting terms in $\mathcal{S}_0$ remain, which are
\begin{align*}
    &\bm{\xi}(a_1, q_2) = \bm{a_1}^{(q_2)} \prod_{b_2 \ni q_2}X_{a_1, b_2},\\
    &\bm{\zeta}(b_1, q_2) = \bm{b_1}^{(q_2)} \prod_{a_2\ni q_2} Z_{b_1, a_2}.
\end{align*}
These four sets of terms form the tensor product code CSS$_1\otimes$ CSS$_2$. The schematics of this construction is given by Figure~\ref{fig:coupledlayerchaincomplex}.

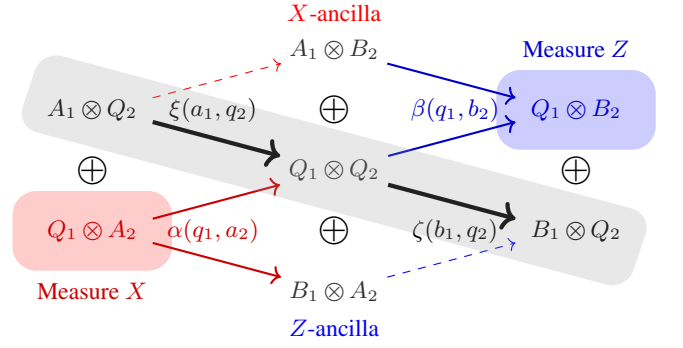
\begin{figure}
\begin{tikzpicture}[scale=0.8]
    \def\L{-4}
    \def\C{0}
    \def\R{4}
    \def\shadedwidth{60}
    \def\shadedheight{160}
    \def\H{1}
    \def\HH{2}
    \def\offs{2}
   
    \begin{scope}
    \node[rounded corners=9pt, fill=red!20,  minimum width=\shadedwidth, minimum height=30] at (\L,-1) {};
    \node[rounded corners=9pt, fill=blue!20, minimum width=\shadedwidth, minimum height=30] at (\R,1) {};
    \end{scope}

    \node[font=\Large, text=black] at (\L, 0) {$\oplus$};
    \node[font=\Large, text=black] at (\C, \H) {$\oplus$};
    \node[font=\Large, text=black] at (\C, -\H) {$\oplus$};
    \node[font=\Large, text=black] at (\R, 0) {$\oplus$};

    \node  (A1Q2) at (\L,  \H) {};
    \node at ($(\L,  \H) + (0, 0)$) {$A_{1}\otimes Q_{2}$};
    \node at ($(\L,  \H) + (\offs, 0)$) {$\xi(a_1, q_2)$};
  
    \node (Q1A2) at (-2, -\H) {};
    \node[text=red!80!black] at ($(\L, -\H) +(0, 0)$) {$Q_{1}\otimes A_{2}$};
    \node[text=red!80!black] at ($(\L, -\H) + ( \offs, 0)$) {$\alpha(q_1, a_2)$};
  
    \node[text=red!80!black] at (\L, -2*\H) {Measure $X$};

    \node[text=black!80]      (A1B2) at (\C,  \HH) {$A_{1}\otimes B_{2}$};
    \node[text=red]      (A1B2caption) at (\C, +1.3*\HH) {$X$-ancilla};
    \node[text=black!80]      (Q1Q2) at (\C,       0.0) {$Q_{1} \otimes Q_{2}$};
    \node[text=black!80]      (B1A2) at (\C, -\HH) {$B_{1}\otimes A_{2}$};
    \node[text=blue]      (B1A2caption) at (\C, -1.3*\HH) {$Z$-ancilla};

    \node[text=blue!80!black] at ($(\R,  \H) + (0, 0)$) {$Q_{1}\otimes B_{2}$};
    \node[text=blue!80!black] at ($(\R,  \H) - (\offs,0)$) {$\beta(q_1, b_2)$};
  
    \node[text=black] (B1Q2) at (\R, -\H) {};
    \node[text=black] at ($(\R, -\H) + (0, 0)$) {$B_{1}\otimes Q_{2}$};
    \node[text=black] at ($(\R, -\H) - (\offs, 0)$) {$\zeta(b_1, q_2)$};
  
    \node[text=blue!80!black] at (\R, 2*\H) {Measure $Z$};

    \draw[->,dashed,red] (-3, 1.2) -- (A1B2);
    \draw[->, ultra thick] (-3, 0.8) -- (Q1Q2);
    \draw[->,red!80!black, thick] (-3, -0.8) -- (Q1Q2);
    \draw[->,red!80!black, thick] (-3, -1.2) -- (B1A2);

    \draw[ ->,blue!80!black, thick] (A1B2) -- (3, 1.2);
    \draw[->,blue!80!black, thick] (Q1Q2) -- (3, 0.8);
    \draw[ ->,ultra thick] (Q1Q2) -- (3, -0.8);
    \draw[ dashed,->,blue] (B1A2) -- (3, -1.2);

    \def\pad{0.35}   
    \def\th{0.21}    
  
    \fill[black!40, fill opacity=0.22, draw opacity=0, rounded corners=8pt]
    ($ (A1Q2) + ({-3*\pad},{\pad}) + ({\th},{3*\th}) $) --
    ($ (B1Q2) + ({ 3*\pad},{-\pad}) + ({\th},{3*\th}) $) --
    ($ (B1Q2) + ({ 3*\pad},{-\pad}) + ({-\th},{-3*\th}) $) --
    ($ (A1Q2) + ({-3*\pad},{\pad}) + ({-\th},{-3*\th}) $) -- cycle;
\end{tikzpicture}
\caption{The coupled layer construction of the tensor product code. Starting with stacks of CSS$_1$ given by the complex in the gray band along with $X/Z$ ancillas, Measuring $X/Z$-stabilizers given by the red/blue arrow ($\alpha(q_1,a_2)/ \beta(q_1,b_2)$) induces the blue/red dashed arrow, completing $\zeta(b_1,q_2)/\xi(a_1,q_2)$.}
\label{fig:coupledlayerchaincomplex}
\end{figure}

The coupled-layer construction is a natural generalization of anyon condensation in topological phases. In the case of a CSS code, a single $X$ on qubit $q$ creates a number of excitations for each $Z$-stabilizer containing $q$. Similarly, $Z$ excites all $X$-stabilizers containing $q$. The code switching $\alpha(q_1, a_2)$ ($\beta(q_1, b_2)$) can be interpreted as simultaneously condensing the excitations created by $X_{q_1}$ ($Z_{q_1}$) in layers of CSS$_1$ labeled by $q_2\in a_2$ ($q_2\in b_2$). That is, CSS$_2$ is treated as a recipe for performing condensation on layers of CSS$_1$. For example, taking CSS$_1$ as the 2D toric code, and CSS$_2$ as the repetition code, then the tensor product corresponds to the 3D toric code. Our coupled layer construction then  reproduces the coupled-layer construction by stacking $2$D toric code and condensing pairs of $\ee$ anyons in nearest neighbor layers.

An alternative view we would like to take is gauging. That is, the condensation of excitations corresponds to gauging the symmetries in each layer. The coupled-layer construction gauges symmetries simultaneously in different layers of CSS$_1$ given by the pattern of checks in CSS$_2$. We make this statement precise in the next section after reviewing the gauging of a CSS code.

\subsection{Gauging}\label{sec:review_g}
In this section, we review gauging of a CSS code. Given a chain complex $A\xrightarrow{\delta} Q\xrightarrow{d^T} B$, we may treat the union of all the $X$-stabilizers and $X$-logicals as an \emph{$X$-type symmetry} of the corresponding code, that is $\Ker d^T$,  which we can gauge. To do this, consider the algebra of symmetric operators under $\Ker d^T$,
\begin{align*}
    \mathcal{A} = \langle \bm b, \,
    X_q
    \,|\,
    b\in B ,\,
    q\in Q
    \rangle.
\end{align*}
First, we introduce ancilla qubits which will play the role of gauge fields in order to promote the symmetry operators into local symmetry operators. To do this, for each $b\in B$, we introduce a qubit with stabilizer $Z_b$. We then impose the following \emph{Gauss law}
\begin{align*}
    \bm{\alpha}(q) = X_q \prod_{b\ni q}X_b.
\end{align*}
which performs a local symmetry action. Indeed, products of these local symmetry actions can recover the original symmetry. 
To see this makes the original symmetry local, take $\mu \in \Ker d^T$ with the corresponding  symmetry operator $\bm\mu = \prod_{q\in \mu} X_q$, we have
\begin{align}
\label{eq:productofGausslaws}
    \prod_{q\in \mu}\bm\alpha(q) = \bm \mu \prod_{b\in d^T(\mu)} X_b = \bm \mu,
\end{align}
where the second equality is because of $d^T(\mu) = 0$.  Not all terms in $\mathcal A$ commute with the Gauss law. In particular, the $Z$-stabilizers $\bm b$ need to be multiplied by $Z_b$ in order to commute with the Gauss law\footnote{In gauge theory language, they need to minimally couple to gauge field in order to be gauge invariant}, which gives
\begin{align*}
    \bm\zeta(b) = \bm b \cdot Z_b.
\end{align*}
Thus the algebra of operators that commute with the Gauss law $\bm\alpha(q)$ is
\begin{align*}
    \mathcal{A}_{\text{gauge}} = 
    \langle \,\bm\zeta(b)
    ,\,
    X_q, 
    \,|\,
    b\in B
    ,\,
    q\in Q 
    \,\rangle
\end{align*}
defined on the Hilbert space constrained by $\bm\alpha(q) = 1$, for all $q\in Q$.

The gauging procedure presented above is not to be confused with the gauging map~\cite{vijay2016fracton,Williamson16,Yoshida16,Yoshida17,Kubica18,Tantivasadakarn20,RakovszkyLDPC1,RakovszkyLDPC2} which acts on the entire symmetric algebra $\mathcal{A}$, rather than deforming a particular code. To recover the gauging map, we note the Gauss law $\bm\alpha(q)$ can actually be imposed exactly after applying the following unitary transformation
\begin{align}\label{eq:clusterdisentangler}
    U = \prod_{q\in Q} \prod_{b\ni q} \text{CNOT}_{q,b}.
\end{align}
Upon conjugating by the unitary
\begin{align*}
    U\bm \zeta(b)U^\dagger = Z_b
    \hspace{15pt}
    U\bm \alpha(q)U^\dagger = X_q,
\end{align*}
and thus we have
\begin{align*}
    U\mathcal{A}_{\text{gauge}}U^\dagger 
    =
    \langle \,
    Z_b,\,\bm \alpha(q)
    \,|\,
    b\in B, \,
    q\in Q
    \,\rangle,
\end{align*}
and the constraint becomes $X_q = 1$ for each $q \in Q$, which projects out all the matter qubits, and can now be directly imposed, giving $\bm \alpha(q) \mapsto \prod_{b\ni q}X_b$. The gauge and matter degree-of-freedom are decoupled in this case. We obtain the gauging map
\begin{align*}
    \bm b\mapsto Z_b,
    \hspace{15pt}
    X_q \mapsto \prod_{b\ni q}X_b.
\end{align*}

To connect with the coupled-layer construction, note the similarity between the code switching $\bm \alpha(q_1, a_2)$ and the Gauss law $\bm\alpha(q)$ (we intentionally used the same notation). One can see that $\bm\alpha(q_1, a_2)$ is the Gauss law upon gauging all $X$-stabilizers and logicals that is diagonal in layers of CSS$_1$ labeled by $q_2\in a_2$, and $\bm\zeta(b_1, q_2)$ is the minimal coupling between $\bm{b_1}^{(q_2)}$ and the gauge qubits. A similar discussion holds when treating all $Z$-stabilizers and logicals as a symmetry to be gauged, in which case $\bm\beta(q_1, b_2)$ is the Gauss law upon gauging the $Z$-symmetry diagonal in layers of CSS$_1$ labeled by $q_2\in b_2$, and $\bm\xi(a_1, q_2)$ is the minimal coupling with gauge qubits. The structure of the coupled-layer construction is then apparent. CSS$_2$ is used as a bookkeeping device to perform simultaneous gauging on multiple layers of CSS$_1$.

To formulate gauging on the chain complex level, we drop the transverse field, and add the $X$-stabilizers explicitly. We start with the original stabilizer group
\begin{align*}
    \mathcal{S} = \langle\,
    \bm b ,\, \bm a
    \,|\,
    b\in B, \, a\in A
    \,\rangle
\end{align*}
After gauging, we may add the Gauss law $\alpha(q)$ explicitly, and obtain a deformed stabilizer code
\begin{align*}
    \mathcal{S}_{\text{gauged}} 
    =
    \langle\,
    \bm\zeta(b),\,
    \bm a,\,
    \bm  \alpha(q),
    \,|\,
    b\in B,\,
    a\in A,\,
    q\in Q
    \,\rangle
\end{align*}
This can be treated as a deformed CSS code, whose associated chain complex is shown in Figure~\ref{fig:gaugecomplex}. It is interesting to note that $\bm a$ are now redundant stabilizers since they live in $Ker d^T$ and therefore can be written as products of Gauss laws according to Eq.~\eqref{eq:productofGausslaws}. The gauged stabilizer group can thus be reduced to 
\begin{align*}
    \mathcal{S}_{\text{gauged}} 
    =
    \langle\,
    \bm\zeta(b),\,
    \bm  \alpha(q),
    \,|\,
    b\in B,\,
    a\in A,\,
    q\in Q
    \,\rangle.
\end{align*}
This is nothing but the stabilizer group of a CSS cluster state, which can be disentangled using the unitary in Eq.~(\ref{eq:clusterdisentangler}).

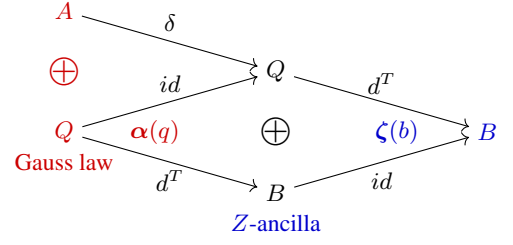
\begin{figure}
\begin{tikzpicture}[scale=0.8, every node/.style={scale=1}]
    \def\L{-3.5}
    \def\C{0}
    \def\R{3.5}
    \def\shadedwidth{60}
    \def\shadedheight{160}
    \def\H{1}
    \def\HH{2}

    \def\shift{0.3}

    \node[text=blue!80!black] at (0, -2.5) {$Z$-ancilla};
    \node[text=red!80!black] at (\L, -1.5) {Gauss law};

    \node[text=red!80!black] (A) at (\L, \H) {$A$};
    \node (Q) at (0 , 0) {$Q$};
    \node[text=blue!80!black] (B) at (\R, -\H) {$B$};

    \node[text=red!80!black] (aux0) at (\L, -\H) {$Q$};
    \node (aux1) at (0, -\HH) {$B$};

    \node[font=\Large, text=red!80!black]  at ($(A)!0.5!(aux0)$) {$\oplus$};
    \node[font=\Large] at ($(Q)!0.5!(aux1)$) {$\oplus$};

    \draw[->] (A) -- (Q);
    \draw[->] (Q) -- (B); 
    \draw[->] (aux0) -- (aux1);

    \draw[->] (aux0) -- (Q); 
    \draw[->] (aux1) -- (B);

    \node at ($(A)!0.5!(Q) + (0,\shift)$) {$\delta$};
    \node at ($(B)!0.5!(Q) + (0,\shift)$) {$d^T$};
    \node at ($(aux0)!0.5!(aux1) - (0,\shift)$) {$d^T$};

    \node at ($(aux0)!0.5!(Q) + (0,\shift)$) {$id$};
    \node at ($(aux1)!0.5!(B) - (0,\shift)$) {$id$};

    \node[text=red!80!black] at ($(aux0) + (1.5,-0)$) {$\bm \alpha(q)$};
    \node[text=blue!80!black] at ($(B) - (1.5,-0)$) {$\bm \zeta(b)$};
\end{tikzpicture}

\caption{The chain complex after gauging all $X$-stabilizers and $X$-logicals in CSS. 
}
\label{fig:gaugecomplex}
\end{figure}

The gauged chain complex is a special case of a more general construction in homological algebra called the mapping cone. We will review this in the next section, and discuss how to generalize the coupled-layer construction beyond tensor product codes from this perspective.

\subsection{Mapping Cone}\label{sec:review_mc}

In this section, we review the mapping cone construction, and its relation to gauging. This relation motivates our constructions in the remainder of the paper. The mapping cone has its roots in algebraic topology~\cite{Hatcher:478079}, where a non-trivial cycle can be trivialized by attaching a cone whose boundary matches that given cycle in order to trivialize the corresponding homology class. This idea can be abstracted to more general (co)chain complexes in homological algebra~\cite{rotman1979introduction}. This construction has found recent use in quantum codes for generalized lattice surgery \cite{hastings2021quantum, cowtan2024ssip,ide2025fault,cowtan2026parallel,yuan2026unified,cowtan2025fast,chang2026constant,zheng2026logical}.

Given a pair of $\mathbb{Z}_2$ chain complexes ($C^n$, $\delta_C^i : C^i \ra C^{i+1}$) and ($W^n$, $\delta_W^i : W^i \ra W^{i+1}$). Choose a degree-$1$ chain map $\Gamma: W\ra C$. That is, a sequence of maps $\Gamma^n: W^n \ra C^{n+1}$ satisfying $\delta_C^{n} \Gamma^{n-1} = \Gamma^{n} \delta_W^{n-1}$. We can construct the mapping cone of $\Gamma$, denoted by cone$(\Gamma)$, which is a chain complex with vector spaces
\begin{align*}
    \text{cone}(\Gamma)^n = W^n\oplus C^n
\end{align*}
and chain maps $\delta^n: \text{cone}(\Gamma)^n \ra \text{cone}(\Gamma)^{n+1}$ given by
\begin{align*}
    (w^n,c^n) 
    \mapsto (\delta^n_W(w^n),\,
    \Gamma^n(w^n) + \delta^{n}_C(c^{n})).
\end{align*}
This chain complex is shown below
\begin{equation*}
\begin{tikzpicture}[scale=1, every node/.style={scale=1}]
    \def\L{-2}
    \def\LL{-4}
    \def\C{0}
    \def\R{2}
    \def\RR{4}
    \def\shadedwidth{27}
    \def\shadedheight{90}
    \def\H{2.5}

    \def\shift{0.3}

    \begin{scope}
        \node[rounded corners=9pt, draw,  minimum width=\shadedwidth, minimum height=\shadedheight] at ($(0,0)!0.5!(0,\H)$) {};
        
        \node[rounded corners=9pt, draw,  minimum width=\shadedwidth, minimum height=\shadedheight] at ($(\L,0)!0.5!(\L,\H)$) {};
        
        \node[rounded corners=9pt, draw,  minimum width=\shadedwidth, minimum height=\shadedheight] at ($(\R,0)!0.5!(\R,\H)$) {};
    \end{scope}

     \node at (-2, -0.7) {\rotatebox{90}{$\,=$}};
    \node at (0, -0.7) {\rotatebox{90}{$\,=$}};
    \node at (2, -0.7) {\rotatebox{90}{$\,=$}};

    \node at (-2, -1.1) {cone$(\Gamma)^{n-1}$};
    \node at (0.2, -1.1) {cone$(\Gamma)^{n}$};
    \node at (2.4, -1.1) {cone$(\Gamma)^{n+1}$};

    \node (Cn-1) at (\L, 0) {$C^{n-1}$};
    \node (Cn) at (0 , 0) {$C^n$};
    \node (Cn+1) at (\R, 0) {$C^{n+1}$};
    
    \node at (\LL, 0) {$\cdots$};
    \node at (\RR, 0) {$\cdots$};

    \node (Zn-1) at (\L, \H) {$W^{n-1}$};
    \node (Zn) at (0 , \H) {$W^n$};
    \node (Zn+1) at (\R, \H) {$W^{n+1}$};
    
    \node at (\LL, \H) {$\cdots$};
    \node at (\RR, \H) {$\cdots$};

    \node [font = \Large] at ($(Cn-1)!0.5!(Zn-1)$) {$\oplus$};
    \node [font = \Large] at ($(Cn)!0.5!(Zn)$) {$\oplus$};
    \node [font = \Large] at ($(Cn+1)!0.5!(Zn+1)$) {$\oplus$};

    \draw[->] (Cn-1) -- (Cn);
    \draw[->] (Cn) -- (Cn+1);
    \draw[->] (-3.6,0) -- (Cn-1);
    \draw[->] (Cn+1) -- (3.6,0);

    \draw[->] (Zn-1) -- (Zn);
    \draw[->] (Zn) -- (Zn+1);
    \draw[->] (-3.6,\H) -- (Zn-1);
    \draw[->] (Zn+1) -- (3.6,\H);

    \draw[->] (Zn-1) -- (Cn); 
    \draw[->] (Zn) -- (Cn+1);
    \draw[->] (-3.7,2.3) -- (Cn-1);
    \draw[->] (Zn+1) -- (3.7,0.2);

    \draw[->] (-1.1,-1.1) -- (-0.5,-1.1);
    \draw[->] (0.9,-1.1) -- (1.5,-1.1);

    \node at ($(Cn-1)!0.5!(Cn) + (0,\shift)$) {$\delta_C^{n-1}$};
    \node at ($(Cn)!0.5!(Cn+1) + (0,\shift)$) {$\delta_C^n$};
    
    \node at ($(Zn-1)!0.5!(Zn) + (0,\shift)$) {$\delta_W^{n-1}$};
    \node at ($(Zn)!0.5!(Zn+1) + (0,\shift)$) {$\delta_W^n$};

    \node at ($(Zn-1)!0.5!(Cn) + (-0.,0)$) {$\Gamma^{n-1}$};
    \node at ($(Zn)!0.5!(Cn+1)+ (0.1,0)$) {$\Gamma^n$};

\end{tikzpicture}
\end{equation*}

If we take the chain complex $C$ as a three-term complex, that is a CSS code $A\xrightarrow{\delta} Q\xrightarrow{d^T} B$, taking $Q$ to be on degree $0$, a degree-$1$ chain map $f: W \ra$ CSS has only one commuting square, hence is equivalent to the $f^0 \delta_W^{-1} = d^T f^{-1}$. In particular, gauging in Figure~\ref{fig:gaugecomplex} corresponds to taking $W^{n} = C^{n+1}$. Thus, the diagrams are trivially commuting, and the gauging chain complex is a special case of a mapping cone.

Yet another example is the tensor product chain complex CSS$_1 \otimes$ CSS$_2$ as shown in Figure~\ref{fig:tensormappingcone}. In this case, the complex $C =  \text{CSS}_1 \otimes Q_2$, the initially decoupled layers. A mapping cone is constructed via a chain map to $C$, and another cone is constructed on the dual complex of $C$.

\begin{figure}[h]
\begin{tikzpicture}[scale=0.8]
    \def\L{-4}
    \def\C{0}
    \def\R{4}
    \def\shadedwidth{60}
    \def\shadedheight{160}
    \def\H{1}
    \def\HH{2}
    \def\offs{2}

    \node[font=\Large, text=black] at (\L, 0) {$\oplus$};
    \node[font=\Large, text=black] at (\C, \H) {$\oplus$};
    \node[font=\Large, text=black] at (\C, -\H) {$\oplus$};
    \node[font=\Large, text=black] at (\R, 0) {$\oplus$};

    \node  (A1Q2) at (\L,  \H) {};
    \node at ($(\L,  \H) + (0, 0)$) {$A_{1}\otimes Q_{2}$};

    \node (Q1A2) at (\L, -\H) {};
    \node[text=red!80!black] at ($(\L, -\H) +(0, 0)$) {$Q_{1}\otimes A_{2}$};

    \node[text=black!80]      (A1B2) at (\C,  \HH) {$A_{1}\otimes B_{2}$};
    \node[text=black!80]      (Q1Q2) at (\C,       0.0) {$Q_{1} \otimes Q_{2}$};
    \node[text=black!80]      (B1A2) at (\C, -\HH) {$B_{1}\otimes A_{2}$};

    \node[text=blue!80!black] at ($(\R,  \H) + (0, 0)$) {$Q_{1}\otimes B_{2}$};

    \node[text=black] (B1Q2) at (\R, -\H) {};
    \node[text=black] at ($(\R, -\H) + (0, 0)$) {$B_{1}\otimes Q_{2}$};
    
    \node (Q1B2) at (\R, \H) {};

    \draw[->] (-3, 1.2) -- (A1B2);
    \draw[->, ultra thick] (-3, 0.8) -- (Q1Q2);
    \draw[->] (-3, -0.8) -- (Q1Q2);
    \draw[->,red!80!black, ultra thick] (-3, -1.2) -- (B1A2);

    \draw[->,blue!80!black, ultra thick] (A1B2) -- (3, 1.2);
    \draw[->] (Q1Q2) -- (3, 0.8);
    \draw[->,ultra thick] (Q1Q2) -- (3, -0.8);
    \draw[->] (B1A2) -- (3, -1.2);

    \def\pad{0.35}   
    \def\th{0.21}    
  
    \fill[black!40, fill opacity=0.22, draw opacity=0, rounded corners=8pt]
    ($ (A1Q2) + ({-3*\pad},{\pad}) + ({\th},{3*\th}) $) --
    ($ (B1Q2) + ({ 3*\pad},{-\pad}) + ({\th},{3*\th}) $) --
    ($ (B1Q2) + ({ 3*\pad},{-\pad}) + ({-\th},{-3*\th}) $) --
    ($ (A1Q2) + ({-3*\pad},{\pad}) + ({-\th},{-3*\th}) $) -- cycle;

    \fill[red!80, fill opacity=0.22, draw opacity=0, rounded corners=8pt]
    ($ (Q1A2) + ({-3*\pad},{\pad}) + ({\th},{3*\th}) $) --
    ($ (B1A2) + ({ 3*\pad},{-\pad}) + ({\th},{3*\th}) $) --
    ($ (B1A2) + ({ 3*\pad},{-\pad}) + ({-\th},{-3*\th}) $) --
    ($ (Q1A2) + ({-3*\pad},{\pad}) + ({-\th},{-3*\th}) $) -- cycle;

    \fill[blue!80, fill opacity=0.22, draw opacity=0, rounded corners=8pt]
    ($ (A1B2) + ({-3*\pad},{\pad}) + ({\th},{3*\th}) $) --
    ($ (Q1B2) + ({ 3*\pad},{-\pad}) + ({\th},{3*\th}) $) --
    ($ (Q1B2) + ({ 3*\pad},{-\pad}) + ({-\th},{-3*\th}) $) --
    ($ (A1B2) + ({-3*\pad},{\pad}) + ({-\th},{-3*\th}) $) -- cycle;
\end{tikzpicture}
\caption{The tensor product CSS$_1 \otimes$ CSS$_2$ as a mapping cone. The gray band is layers of CSS$_1$ indexed by $Q_2$, giving CSS$_1 \otimes Q_2$. The red band is an auxiliary complex with a chain map to CSS$_1 \otimes Q_2$. Together with the gray band, they form the mapping cone. The blue band is another auxiliary complex with a chain map to the dual of CSS$_1 \otimes Q_2$ upon transposing, which also forms a mapping cone. The two mapping cones together form the tensor product chain complex.}
\label{fig:tensormappingcone}
\end{figure}

The mapping cone is closely related to code surgery and gauging logicals. Start with a CSS code, and attach the complex $\mathcal{E}_x \xrightarrow{\epsilon'_x} Q_x$ with maps $\Gamma^{-1} = \epsilon_x$ and $\Gamma^0 = {d'}^T$ as shown in Figure~\ref{fig:MappingConeX}. Due to the commutativity of the diagram, it defines a deformed CSS code with new qubits $Q_x$ and new $X$-stabilizers $\mathcal{E}_x$. Given $e_x \in \mathcal{E}_x$, the corresponding stabilizer takes the form
\begin{align*}
    \bm{e_x} = \prod_{q\in \epsilon_x(e_x)} X_{q} \prod_{q_x\in \epsilon_x'(e_x)}X_{q_x}.
\end{align*}
In addition, we have $\epsilon_x(\Ker \epsilon_x') \subseteq \Ker d^T$. To see this, given $\mu \in \Ker \epsilon_x'$, we have $d^T \epsilon_x(\mu) = d'^T \epsilon_x'(\mu) = 0$. It follows that
{\begin{align*}
    \prod_{e_x\in \mu} \bm{e_x} = \prod_{q\in \epsilon_x(\mu)} X_{q} \prod_{q_x\in \epsilon_x'(\mu)}X_{q_x}
    =
    \prod_{q\in \epsilon_x(\mu)} X_{q}
\end{align*}}
We see the $X$ on ancilla qubits $Q_x$ cancel, and we are left with the $X$-logical in the CSS code. Therefore, we can think of $Q_x$ as the gauge qubits introduced in order to gauge a \emph{subset} of the $X$-symmetry given by $\epsilon_x(\Ker \epsilon_x')$, and the $X$-stabilizers in $\mathcal{E}_x$ are the Gauss laws.  In the special case of Figure~\ref{fig:gaugecomplex}, $\epsilon_x = id$, so $\epsilon_x(\Ker d^T) = \Ker d^T$, which indeed gauges all $X$-stabilizers as well as $X$-logicals in CSS. In the case of Figure~\ref{fig:tensormappingcone}, looking at the red mapping cone, we have $id \otimes \delta_2 (\Ker d_1^T \otimes A_2) = \Ker d_1^T \otimes \Imaa \delta_2$, which gauges the $X$-symmetry generated by all $X$-stabilizers and $X$-logicals that are diagonal in layers of CSS$_1$ labeled by elements in $\Imaa \delta_2$.

\begin{figure}[h]
\begin{tikzpicture}[scale=0.8, every node/.style={scale=1}]
    \def\L{-3.5}
    \def\C{0}
    \def\R{3.5}
    \def\shadedwidth{60}
    \def\shadedheight{160}
    \def\H{1}
    \def\HH{2}

    \def\shift{0.3} 

    \node[text=red!80!black] (A) at (\L, \H) {$A$};
    \node (Q) at (0 , 0) {$Q$};
    \node[text=blue!80!black] (B) at (\R, -\H) {$B$};

    \node[text=red!80!black] (aux0) at (\L, -\H) {$ \mathcal{E}_x$};
    \node (aux1) at (0, -\HH) {$Q_x$};

    \node[font=\Large, text=red!80!black]  at ($(A)!0.5!(aux0)$) {$\oplus$};
    \node[font=\Large] at ($(Q)!0.5!(aux1)$) {$\oplus$};

    \draw[->] (A) -- (Q);
    \draw[->] (Q) -- (B); 
    \draw[->] (aux0) -- (aux1);

    \draw[->] (aux0) -- (Q); 
    \draw[->] (aux1) -- (B);

    \node at ($(A)!0.5!(Q) + (0,\shift)$) {$\delta$};
    \node at ($(B)!0.5!(Q) + (0,\shift)$) {$d^T$};
    \node at ($(aux0)!0.5!(aux1) - (0,\shift)$) {$\epsilon_x'$};

    \node at ($(aux0)!0.5!(Q) + (0.2,\shift)$) {$\epsilon_x$};
    \node at ($(aux1)!0.5!(B) - (0.3,\shift)$) {$d'^T$};

    \node[text=blue!80!black] at (0, -2.5) {$Z$-ancilla};
    \node[text=red!80!black] at (\L, -1.5) {Gauss law};
\end{tikzpicture}
\caption{The mapping cone for code surgery on $X$-logicals. The diagram commutes $d^T\epsilon_x = d'^T\epsilon_x'$.}
\label{fig:MappingConeX}
\end{figure}
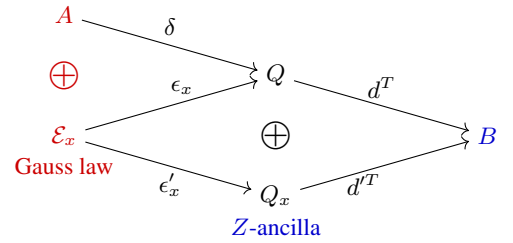

\begin{figure}[h]
    \begin{tikzpicture}[scale=0.8, every node/.style={scale=1}]
    \def\L{-3.5}
    \def\C{0}
    \def\R{3.5}
    \def\shadedwidth{60}
    \def\shadedheight{160}
    \def\H{1}
    \def\HH{2}

    \def\shift{0.3}

    \node[text=red!80!black] at (0, 2.6) {$X$-ancilla};
    \node[text=blue!80!black] at (\R, 1.6) {Gauss law};
    
    \node [text=red!80!black] (A) at (\L, \H) {$A$};
    \node (Q) at (0 , 0) {$Q$};
    \node [text=blue!80!black] (B) at (\R, -\H) {$B$};

    \node (aux0) at (0, \HH) {$Q_z$};
    \node [text=blue!80!black] (aux1) at (\R, \H) {$\mathcal{E}_z$};

    \node[font=\Large, text=blue!80!black]  at ($(B)!0.5!(aux1)$) {$\oplus$};
    \node[font=\Large] at ($(Q)!0.5!(aux0)$) {$\oplus$};

    \draw[<-] (A) -- (Q);
    \draw[<-] (Q) -- (B); 
    \draw[<-] (aux0) -- (aux1);

    \draw[->] (aux0) -- (A); 
    \draw[->] (aux1) -- (Q);

    \node at ($(A)!0.5!(Q) - (0,\shift)$) {$\delta^T$};
    \node at ($(B)!0.5!(Q) - (0,\shift)$) {$d$};
    \node at ($(aux0)!0.5!(aux1) + (0,\shift)$) {$\epsilon_z'$};

    \node at ($(aux0)!0.5!(A) + (0,\shift)$) {$\delta'^T$};
    \node at ($(aux1)!0.5!(Q) + (0,\shift)$) {$\epsilon_z$};
\end{tikzpicture}
\caption{The mapping cone for code surgery on $Z$-logicals. The diagram commutes $\delta^T\epsilon_z = \delta'^T\epsilon_z'$. 
}
\label{fig:MappingConeZ}
\end{figure}
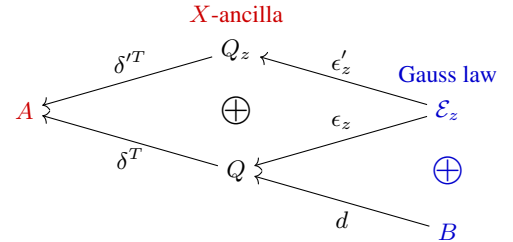

Let us also provide a physical interpretation for the spaces $\mathcal{E}_x$ and $Q_x$ beyond its role as the Gauss law and the $Z$-ancilla. In particular, 
the $\mathbb{Z}_2$ vector space $\mathcal{E}_x$ is spanned by a basis set $\{e_x\}$, each of which corresponds to some $X$-excitations in the CSS code $A\xrightarrow{\delta} Q\xrightarrow{d^T} B$. These $X$-excitations are products of $X$, which anticommute with a number of $Z$-stabilizers, and are expressed through the ``excitation" map $\epsilon_x: \mathcal{E}_x \ra Q$. On the other hand, the map $d': B \rightarrow Q_x$ plays the role of an auxiliary ``dual" system, defined on qubits spanning the $\mathbb{Z}_2$ vector space $Q_x$, which is a classical code in the $Z$-basis. In the case of gauging in Figure~\ref{fig:gaugecomplex}, the $X$-excitations are single $X_q$ for $q\in Q$, and the dual system is defined on qubits labeled by $B$ with single-Pauli stabilizers $Z_b$ for $b\in B$.

With the dual system and excitations in the CSS code specified, the linear map $\epsilon_x': \mathcal{E}_x\ra Q_x$ identifies a set of $X$-excitations in the dual system. Importantly, the commutativity of the diagram means that given $e_x\in \mathcal{E}_x$ and $b\in B$, $\epsilon_x(e_x)$ and $d(b)$ have even/odd overlap in $Q$ if and only if $\epsilon_x'(e_x)$ and $d'(b)$ have even/odd overlap in $Q_x$.
Thus, the corresponding operators $ \bm{\epsilon_x}(e_x) =\prod_{q\in \epsilon_x(e_x)} X_q$ and $\bm d(b) = \prod_{q \in d(b)} Z_q$ on the Hilbert space $\bm Q$ have the same commutation relations as the operators $\bm{\epsilon_x}'(e_x)$ and $\bm{d'}(b)$ on the Hilbert space $\bm{Q}_x$. That is, the algebras generated by $\{\bm{\epsilon_x}(e_x) ,\, \bm{d}(b)\}$ and $\{\bm{\epsilon_x}'(e_x) ,\, \bm{d}'(b)\}$ correspond to the same \emph{bond algebra} ~\cite{Nussinov09,Cobanera10,Cobanera11} (also sometimes called patch operators~\cite{JiWen20}).

This fact can be made more precise as follows. Define an algebra
\begin{align*}
    \bm {\mathcal A}_x &=  \langle\bm{\epsilon_x}(e_x) ,\, \bm{d}(b), \,\bm \delta(a) \,|\, e_x \in \mathcal{E}_x, \, b\in B, \, a \in A \,\rangle,
\end{align*}
which we refer to $\bm {\mathcal A}_x$ as the ($X$-type) \textit{excitation algebra} in the CSS code, it is generated by the stabilizers as well as the Pauli $X$ excitations $\bm{\epsilon_x}(e_x)$. The preservation of commutation relation is captured by the following map $\bm{\Gamma}_x: \bm {\mathcal A}_x \rightarrow \bm \langle\bm{\epsilon_x'}(e_x) ,\, \bm{d}'(b) \, \rangle$
\begin{align*}
    \bm{\Gamma}_x: \bm{\epsilon_x}(e_x) \mapsto \bm{\epsilon_x'}(e_x),
    \hspace{15pt}
    \bm d(b)\mapsto \bm{d'}(b),
    \hspace{15pt}
    \bm \delta(a) \mapsto 1
\end{align*}
Therefore, given a mapping cone as in Figure~\ref{fig:MappingConeX}, we have defined a operator map $\bm \Gamma_x$ on the excitation algebra, which preserves the commutation relations.

The converse is also true. Starting with the CSS code, and a set of $X$-excitations given by a map $\epsilon_x: \mathcal{E}_x \ra Q$. Assume we are given an algebra-preserving map $\bm \Gamma_x$ defined on the excitation algebra $\bm{\mathcal{A}}_x$, whose image is defined on qubits $\bm{Q}_x$, with $\bm \Gamma_x(\bm a) = 1$. Then we can use $\bm \Gamma_x$ to define a mapping cone. Indeed, since $\bm\Gamma_x(\bm \epsilon_x(e_x))$ are Pauli-$X$ operators, it is equivalent to a linear map $\epsilon_x': \mathcal{E}_x \ra Q_x$. Similarly, the Pauli-$Z$ operators $\bm \Gamma_x(\bm d(b))$ is equivalent to a linear map $d': B \ra Q_x$. The assumption that $\bm \Gamma_x$ preserves commutation relation implies $d^T \epsilon_x = d'^T\epsilon_x'$, hence defines a chain map.

As a result, we see there is a bijection between the mapping cones as in Figure~\ref{fig:MappingConeX}, and pairs of $\bm \Gamma_x$ and $\mathcal{E}_x$ with $\bm \Gamma_x(\bm a) = 1$. We have thus intentionally used the symbol $\Gamma$ for both the chain map of mapping cone and the algebra-preserving map. One may think of $\bm \Gamma_x$ as the operator version of the mapping cone.

We note, similar to gauging, in order for this map to be an algebra homomorphism, we need to restrict the domain and codomain to act on appropriate subspaces. On the domain, we need to restrict to the subspace symmetric under the kernel of $\bm \Gamma_x$. Since $\bm \Gamma_x$ preserves commutation relation, $\bm{\mathcal{A}}_x$ is well-defined on this subspace. On the codomain, we need to restrict to the subspace symmetric under those operators on $Q_x$ whose preimage is $1$. The operators in the image of $\bm \Gamma_x$ are well-defined on this subspace also because $\bm \Gamma_x$ preserves commutation relation. Since the constructions to be presented only require the preservation of commutation relation, we do not develop further in this direction.

We illustrate this using the simple mapping cone corresponding to full gauging in Fig.~\ref{fig:gaugecomplex}. Given a CSS code $A\xrightarrow{\delta} Q\xrightarrow{d^T} B$, the excitation map is $\epsilon_x = id: Q \ra Q $. This says the excitations are single-Pauli operators, and the excitation algebra is
\begin{align*}
    \bm {\mathcal A}_x &=  \langle\,
    X_q ,\, \bm{d}(b), \,\bm \delta(a) \,|\, q \in Q, \, b\in B, \, a \in A \,\rangle,
\end{align*}
The auxiliary dual system is defined $d' = id: B\ra B$, which lives on a Hilbert space spanned by $B$, with the trivial stabilizers $Z_b$. As a result, $\bm \Gamma_x$ acts by
\begin{align*}
    \bm\Gamma_x(\bm d(b)) = Z_b
\end{align*}
The excitation in the dual system is $\epsilon_x' = d^T : Q \ra B$, hence $\bm \Gamma_x$ acts by
\begin{align*}
    \bm\Gamma_x(X_q) = \prod_{b \ni q} X_b
\end{align*}
In addition to the trivial action on $X$-stabilizers
\begin{align*}
    \bm \Gamma_x(\bm\delta(a)) = 1
\end{align*}
These define the full action of $\bm \Gamma_x$ on the excitation algebra $\mathcal{A}_x$. This is precisely the gauging map which gauges all $X$-operators in $\Ker d^T$. We see Fig.~\ref{fig:gaugecomplex} indeed recovers full gauging.

Although we will not focus on such cases, one may also choose $\bm \Gamma_x$ such that the $X$-stabilizers $\bm \delta(a)$ are mapped to non-trivial operators. In this case, the $X$-operators $\bm \Gamma_x(\bm a)$ define an additional linear map $\delta': A \ra Q_x$, and we need to modify the mapping cone to the following diagram
\begin{equation*}
\begin{tikzpicture}[scale=0.8, every node/.style={scale=1}]
    \def\L{-3.5}
    \def\C{0}
    \def\R{3.5}
    \def\shadedwidth{60}
    \def\shadedheight{160}
    \def\H{1}
    \def\HH{2}
    
    \def\shift{0.3} 

    \node[text=red!80!black] (A) at (\L, \H) {$A$};
    \node (Q) at (0 , 0) {$Q$};
    \node[text=blue!80!black] (B) at (\R, -\H) {$B$};

    \node[text=red!80!black] (aux0) at (\L, -\H) {$\mathcal{E}_x$};
    \node (aux1) at (0, -\HH) {$Q_x$};

    \draw[->] (A) -- (Q);
    \draw[->] (Q) -- (B); 
    \draw[->] (aux0) -- (aux1);

    \draw[->] (aux0) -- (Q); 
    \draw[->] (aux1) -- (B); 
    \draw[->] (A) -- (aux1);

    \node at ($(A)!0.5!(Q) + (0,\shift)$) {$\delta$};
    \node at ($(B)!0.5!(Q) + (0,\shift)$) {$d^T$};
    \node at ($(aux0)!0.5!(aux1) - (0,\shift)$) {$\epsilon_x'$};

    \node at ($(A)!0.5!(aux1) - (0.3,-0.6)$) {$\delta'$};
    \node at ($(aux0)!0.5!(Q) - (0.5,0.4)$) {$\epsilon_x$};
    \node at ($(aux1)!0.5!(B) - (0.3,\shift)$) {$d'^T$};
\end{tikzpicture}
\end{equation*}
such that each square commutes, that is,
\begin{align*}
    d^T\epsilon_x = d'^T\epsilon_x'
    \hspace{15pt}
    \text{and}
    \hspace{15pt}
    d^T\delta = d'^T \delta'.
\end{align*}
Conversely, given such a diagram, we can define an algebra-preserving map on the excitation algebra by
\begin{align*}
    \bm{\Gamma}_x: \bm{\epsilon_x}(e_x) \mapsto \bm{\epsilon_x'}(e_x),
    \hspace{13pt}
    \bm d(b)\mapsto \bm{d'}(b),
    \hspace{13pt}
    \bm \delta(a) \mapsto \bm \delta'(a)
\end{align*}

Until now, we have discussed the partial gauging of $X$ logicals. The partial gauging of $Z$ logicals is almost identical. Taking the reflection of Figure~\ref{fig:MappingConeX} gives Figure~\ref{fig:MappingConeZ}. Similarly, this mapping cone gauges the $Z$-operators in $\epsilon_z(\Ker \epsilon_z')$. For instance, the blue mapping cone in Figure~\ref{fig:tensormappingcone} gauges the $Z$-symmetry generated by $\Ker \delta_1^T \otimes \Imaa d_2$. Similarly, we can define a $Z$-type excitation algebra
\begin{align*}
    \bm {\mathcal{A}}_z = \langle\bm{\epsilon_z}(e_z) ,\, \bm{\delta}(a), \,\bm b \,|\, e_z \in \mathcal{E}_z, \, b\in B, \, a \in A \,\rangle
\end{align*}
and an algebra-preserving map
\begin{align*}
    \bm{\Gamma}_z: \bm{\epsilon_z}(e_z) \mapsto \bm{\epsilon_z'}(e_z),
    \hspace{15pt}
    \bm \delta(a) \mapsto \bm{\delta'}(a),
    \hspace{15pt}
    \bm b \mapsto 1
\end{align*}
which encodes the same data as the mapping cone.

Henceforth we use $\bm{a}$, $\bm{b}$, $\bm{e_x}$ and $\bm{e_z}$ to denote the operators $\bm{\delta}(a)$, $\bm d(b)$, $\bm{\epsilon_x}(e_x)$ and $\bm{\epsilon_z}(e_z)$ in $Q$. Their images under the algebra map $\bm{\Gamma}_z(\bm{a})$, $\bm{\Gamma}_x(\bm{b})$, $\bm{\Gamma}_x(\bm{e_x})$ and $\bm{\Gamma}_z(\bm{e_z})$ are the operators $\bm{\delta'}(a)$, $\bm{d'}(b)$, $\bm{\epsilon_x'}(e_x)$ and $\bm{\epsilon_z'}(e_z)$ in $Q_x$.

\subsubsection{Example: Partial Gauging in a Classical Code}

Let us give an example of how partial gauging arises from the mapping cone. In anticipation of the next section, we will use this specific mapping cone to perform partial gauging in multiple layers in order to obtain the X-cube model via the coupled-layer construction in Sec.~\ref{sec:cssclc_xc}. 

Consider two repetition codes denoted $E^x \xrightarrow{\partial^x} V^x$ and $E^y \xrightarrow{\partial^y} V^y$. We note that although they are isomorphic chain complexes, we denote them differently for clarity. The classical code consists of taking a 2D square lattice and stacking the first code in the $x$ direction and the second code in the $y$ direction. Visually, the $Z$ stabilizers are given by
\begin{align}
\label{eq:XcubepregaugeIsingterms}
  \bm b_{v,x} &=    \begin{tikzpicture}
    [baseline=0cm, scale = 0.8]
        \node at (0,0) {};
        \node at (-0.7, 0) {$Z$};
        \node at (0.7,0) {$Z$};
        \node at (0, 0.7) {};
        \node at (0,-0.7) {};
        \draw[step = 2] (-1.2,-1.2) grid (1.2,1.2);
    \end{tikzpicture} & \bm b_{v,y}&=\begin{tikzpicture}
    [baseline=0cm, scale = 0.8]
        \node at (0,0) {};
        \node at (-0.7, 0) {};
        \node at (0.7,0) {};
        \node at (0, 0.7) {$Z$};
        \node at (0,-0.7) {$Z$};
        \draw[step = 2] (-1.2,-1.2) grid (1.2,1.2);
    \end{tikzpicture}
\end{align}

To write the corresponding chain complex, we introduce a shorthand
\begin{align*}
    V &= V^x \otimes V^y\\
    E &= (E^x \otimes V^y) \oplus (V^x \otimes E^y)\\
    P &= E^x \otimes E^y
\end{align*}
which are just the vertex, edges and plaquettes of the square lattice. The corresponding chain complex is
\begin{equation*}
\begin{tikzpicture}[scale=0.8, every node/.style={scale=1}]
    \def\L{-3.5}
    \def\R{3.5}
    \node at (-0.75,0.05) {$d^T:$};
    \node (Q1) at (\R, 0) {$V\oplus V $};
    \node [text=black] (Ex) at (0 , 0) {$E$};
    
    \draw[<-] (Q1) -- (Ex);
\end{tikzpicture}
\end{equation*}
where $d^T = (\partial^x \otimes \mathbbm{1}) \oplus (\mathbbm{1} \otimes \partial^y)$.

Next, we specify the algebra-preserving map $\bm \Gamma_x$, or equivalently, the mapping cone for partial gauging. The excitations in CSS$_1$ are chosen to be the plaquette terms
\begin{align*}
    \bm e_{x,p} = \begin{tikzpicture}
    [baseline = 0.8cm, scale = 0.8]
        \node at (0, 1) {$X$};
        \node at (2,1) {$X$};
        \node at (1, 0) {$X$};
        \node at (1,2) {$X$};
        \draw[step = 2] (0,0) grid (2,2);
    \end{tikzpicture}
\end{align*}
This can be written in terms of linear map
\begin{equation*}
\begin{tikzpicture}[scale=0.8, every node/.style={scale=1}]
    \def\L{-3.5}
    \def\R{3.5}
    \node at (-0.75,-0.05) {$\epsilon_x:$};
    \node (Q1) at (\R, 0) {$E$};
    \node [text=black] (Ex) at (0 , 0) {$P$};
    
    \draw[<-] (Q1) -- (Ex);
\end{tikzpicture}
\end{equation*}
where $\epsilon_x = (\mathbbm{1} \otimes  \partial^y \ \   \partial^x \otimes \mathbbm{1}   )^T$.
Each $X$ in $\bm e_{x,p}$ creates a pair of domain walls, so $\bm e_{x,p}$ creates a pair of domain walls at each corner of $p$. Thus the excitation algebra is
\begin{align*}
    \bm{\mathcal{A}}_x = 
    \langle\,
    \bm e_{x,p}, \bm b_{v,x}, \bm b_{v,y} 
    \,|\,
    p\in P, \, v\in V
    \,\rangle
\end{align*}

The qubits of the auxiliary system are placed on vertices, and the algebra-preserving map $\bm{\Gamma}_x$ defined by
\begin{equation*}
    \centering
    \begin{matrix}
        \bm b_{v,x},\\
        \bm b_{v,y}
    \end{matrix}
    \hspace{5pt}
    \xrightarrow{\bm{\Gamma}_x}
    \hspace{5pt}
    \begin{tikzpicture}
    [baseline=0cm, scale = 0.8]
        \node at (0,0) {$Z$};
        
        \draw[step = 2] (-1.2,-1.2) grid (1.2,1.2);
    \end{tikzpicture}
    \hspace{25pt}
    \bm e_{x,p}
    \hspace{5pt}
    \xrightarrow{\bm{\Gamma}_x}
    \hspace{5pt}
    \begin{tikzpicture}
    [baseline=0.8cm, scale = 0.8]
        \node at (0, 0) {$X$};
        \node at (2,0) {$X$};
        \node at (0,2) {$X$};
        \node at (2,2) {$X$};
        
        \draw[ step = 2] (0,0) grid (2,2);
    \end{tikzpicture}
\end{equation*}
One can check explicitly that $\bm{\Gamma}_x$ indeed preserves the commutation relations on both sides, and in the auxiliary system, we see $\bm{\Gamma}_x(\bm e_{x,p})$ creates four $\mathbb{Z}_2$ charges on the corners of $p$. This algebra-preserving map is equivalent to the mapping cone 
\begin{equation}\label{eqn:Xcube_MC}
\begin{adjustbox}{valign=c}
\begin{tikzpicture}[scale=0.8, every node/.style={scale=1}]
    \def\L{-3.5}
    \def\R{3.5}
    \def\H{1}
    \def\HH{2}
    \def\shift{0.3}
    
    \node (P) [text=red!80!black] at (\L, -\H) {$P$};
    \node (V) at (0, -\HH) {$V$};
    \node (E) at (0, 0) {$E$};
    \node [text=blue!80!black] (VxVy) at (\R , -\H) {$V \oplus V$};
    
    \draw[->] (P) -- (V);
    \draw[->] (P) -- (E);
    \draw[->] (V) -- (VxVy);
    \draw[->] (E) -- (VxVy);
    \node at ($(E)!0.5!(VxVy) + (0,\shift)$) {$d^T$};
    \node at ($(P)!0.5!(V) - (0,\shift)$) {$\epsilon_x'$};

    \node at ($(P)!0.5!(E) + (0.2,\shift)$) {$\epsilon_x$};
    \node at ($(V)!0.5!(VxVy) - (0.3,\shift)$) {$d'^T$};
\end{tikzpicture}
\end{adjustbox}
\end{equation}
where
\begin{align*}
    \epsilon'_x & = \partial^x \otimes \partial^y\\
   {d'}^T & = (\mathbbm{1} \ \  \mathbbm{1})^T
\end{align*}
In order to see which logicals this mapping cone gauges, we need to consider $\epsilon_x(\Ker \epsilon_x')$. $\Ker \epsilon_x'$ is the vector space of metachecks (i.e., redundancies) of code $\epsilon_x'$, which is $\bm{\Gamma}_x(\bm e_{x,p})$. These metachecks are labeled by straight lines $l^*$ in $x$ and $y$-direction on the dual lattice,
\begin{align*}
    \prod_{p\in l^*} \bm{\Gamma}_x(\bm e_{x,p}) = 1
\end{align*}
The image of $l^*$ under $\epsilon_x$ is a pair of parallel lines in $x$ and $y$-direction. For example, if $l^*$ extends in $x$-direction, then
\begin{align*}
    \prod_{p\in l^*} \bm e_{x,p}
    \hspace{5pt}
    =
    \hspace{5pt}
    \begin{tikzpicture}
    [baseline=0.8cm, scale = 0.8]
        \node at (1, 0) {$X$};
        \node at (1,2) {$X$};
        \draw[ step = 2] (-0.5,0) grid (2,2);
        \node at (3, 0) {$X$};
        \node at (3,2) {$X$};
        \draw[ step = 2] (2,0) grid (4.5,2);
    \draw[dotted] (-1,1) -- (5,1);
    \node at (5.5,1) {$l^*$};
    \end{tikzpicture}
\end{align*}
Thus, we see that the symmetry we choose to gauge consists of pairs of adjacent lines, instead of individual lines. Indeed, no product of plaquette terms $\bm e_{x,p}$ can reproduce a single line.

To perform the gauging, we introduce an ancilla qubit at every vertex initialized in $Z_v=1$ and enforce the Gauss law
\begin{align*}
\begin{tikzpicture}[baseline=0.8cm, scale = 0.8]
    \node at (0, 0) {$X$};
    \node at (2,0) {$X$};
    \node at (0,2) {$X$};
    \node at (2,2) {$X$};
    \node at (1, 0) {$X$};
    \node at (2,1) {$X$};
    \node at (1,2) {$X$};
    \node at (0,1) {$X$};
    \draw[ step = 2] (0,0) grid (2,2);
\end{tikzpicture}
\end{align*}
Indeed, the product of these operators reproduces the symmetries we wish to gauge without the action on ancillas. The $ZZ$ checks in Eq.~\eqref{eq:XcubepregaugeIsingterms} get minimally coupled into
\begin{align*}
 \begin{tikzpicture}
    [baseline=0cm, scale = 0.8]
        \node at (0,0) {$Z$};
        \node at (-0.7, 0) {$Z$};
        \node at (0.7,0) {$Z$};
        \node at (0, 0.7) {};
        \node at (0,-0.7) {};
        \draw[step = 2] (-1.2,-1.2) grid (1.2,1.2);
    \end{tikzpicture}
    \hspace{15pt}
    \begin{tikzpicture}
    [baseline=0cm, scale = 0.8]
        \node at (0,0) {$Z$};
        \node at (-0.7, 0) {};
        \node at (0.7,0) {};
        \node at (0, 0.7) {$Z$};
        \node at (0,-0.7) {$Z$};
        \draw[step = 2] (-1.2,-1.2) grid (1.2,1.2);
    \end{tikzpicture}
\end{align*}
The resulting CSS code above is exactly the mapping cone \eqref{eqn:Xcube_MC}. The $X$-stabilizers correspond to plaquettes $P$, while the two $Z$-stabilizers correspond to the two vertex terms in $V \oplus V$.

\section{CSS Coupled-Layer Codes}\label{sec:cssclc}

In this section, we consider a simple generalization of the tensor product on the chain complex level using the mapping cone. Our construction is closely related to partial block reading in code surgery~\cite{cowtan2025fast}. However, a more appropriate interpretation in our case is condensation of excitations, which does not necessarily perform a readout of logicals if those logicals are not gauged. We also note that a related coupled-layer construction to construct interesting topological orders is mentioned in~\cite{Liu23, Garre-Rubio24, Cuiper2025systematic,Liu26}. However, our construction is not restricted to layering along one dimension. The construction in this section is a special case of the coupled-layer code we will later present in Sec.~\ref{sec:clc}, where we restrict both inputs and outputs to be CSS codes, and hence can be fully described in terms of chain complexes. As a result, we refer to it as the \textit{CSS coupled-layer code}.

As discussed in Sec.~\ref{sec:review_g}, the tensor product relies on a particular mapping cone. Namely, using the mapping cone that gauges all logicals in a single layer given in Fig.~\ref{fig:gaugecomplex}, one can construct a mapping cone in a layered system of CSS$_1$ which gauges combinations of these logicals according to the checks in CSS$_2$, as in Figure~\ref{fig:gaugecomplex}. A straightforward generalization is to simply replace the gauging mapping cone in a single layer by a different mapping cone as in Figure~\ref{fig:MappingConeX}.

In analogy to tensor product, we need to promote this mapping cone to a complex similar to Figure~\ref{fig:tensormappingcone}. The resulting chain complex is shown in Figure~\ref{fig:CoupLayerCSS}. The chain maps between red or blue stripes and the gray stripe are well-defined, since the two parallelograms commute. To see this for the top parallelogram,
\begin{align*}
    (\epsilon_z'^T\otimes id)(\delta_1' \otimes d_2^T) = \epsilon_z'^T \delta_1' \otimes  d_2^T = (\epsilon_z^T\otimes d_2^T)(\delta_1 \otimes id),
\end{align*}
and similarly for the bottom parallelogram. By taking the middle column in Figure~\ref{fig:CoupLayerCSS} to be the new qubits, left column as $X$-stabilizers, and right column to be $Z$-stabilizers, we obtain a new CSS code by taking two CSS codes as input, and supplemented by data provided by the mapping cones.

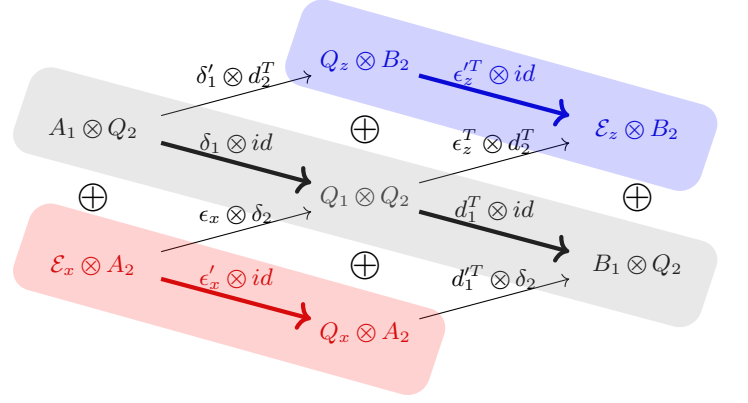
\begin{figure}
\centering
\hspace{0.8cm}
\begin{tikzpicture}[scale=0.9]
    \def\L{-4}
    \def\C{0}
    \def\R{4}
    \def\shadedwidth{60}
    \def\shadedheight{160}
    \def\H{1}
    \def\HH{2}
    \def\offs{2}

    \node[font=\Large, text=black] at (\L, 0) {$\oplus$};
    \node[font=\Large, text=black] at (\C, \H) {$\oplus$};
    \node[font=\Large, text=black] at (\C, -\H) {$\oplus$};
    \node[font=\Large, text=black] at (\R, 0) {$\oplus$};

    \node  (A1Q2) at (\L,  \H) {};
    \node at ($(\L,  \H) + (0, 0)$) {$A_{1}\otimes Q_{2}$};

    \node (Q1A2) at (\L, -\H) {};
    \node[text=red!80!black] at ($(\L, -\H) +(0, 0)$) {$\mathcal{E}_{x}\otimes A_{2}$};

    \node[text=blue!80!black] (A1B2) at (\C,  \HH) {$Q_{z}\otimes B_{2}$};
    \node[text=black!80]      (Q1Q2) at (\C,       0.0) {$Q_1 \otimes Q_{2}$};
    \node[text=red!80!black]      (B1A2) at (\C, -\HH) {$Q_{x}\otimes A_{2}$};

    \node[text=blue!80!black] at ($(\R,  \H) + (0, 0)$) {$\mathcal{E}_{z}\otimes B_{2}$};

    \node[text=black] (B1Q2) at (\R, -\H) {};
    \node[text=black] at ($(\R, -\H) + (0, 0)$) {$B_{1}\otimes Q_{2}$};
    
    \node (Q1B2) at (\R, \H) {};

    \draw[->] (-3, 1.2) -- (A1B2) node[midway, above] {$\delta_1' \otimes d_2^T$};
    \draw[->, ultra thick] (-3, 0.8) -- (Q1Q2) node[midway, above] {$\delta_1 \otimes id$};
    \draw[->] (-3, -0.8) -- (Q1Q2) node[midway, above] {$\epsilon_x \otimes \delta_2$};
    \draw[->,red!80!black, ultra thick] (-3, -1.2) -- (B1A2) node[midway, above] {$\epsilon_x' \otimes id$};

    \draw[->,blue!80!black, ultra thick] (A1B2) -- (3, 1.2) node[midway, above] {$\epsilon_z'^T \otimes id$};
    \draw[->] (Q1Q2) -- (3, 0.8) node[midway, above] {$\epsilon_z^T \otimes d_2^T$};
    \draw[->,ultra thick] (Q1Q2) -- (3, -0.8) node[midway, above] {$d_1^T \otimes id$};
    \draw[->] (B1A2) -- (3, -1.2) node[midway, above] {$d_1'^T \otimes \delta_2$};

    \def\pad{0.35}   
    \def\th{0.21}    
  
    \fill[black!40, fill opacity=0.22, draw opacity=0, rounded corners=8pt]
    ($ (A1Q2) + ({-3*\pad},{\pad}) + ({\th},{3*\th}) $) --
    ($ (B1Q2) + ({ 3*\pad},{-\pad}) + ({\th},{3*\th}) $) --
    ($ (B1Q2) + ({ 3*\pad},{-\pad}) + ({-\th},{-3*\th}) $) --
    ($ (A1Q2) + ({-3*\pad},{\pad}) + ({-\th},{-3*\th}) $) -- cycle;

    \fill[red!80, fill opacity=0.22, draw opacity=0, rounded corners=8pt]
    ($ (Q1A2) + ({-3*\pad},{\pad}) + ({\th},{3*\th}) $) --
    ($ (B1A2) + ({ 3*\pad},{-\pad}) + ({\th},{3*\th}) $) --
    ($ (B1A2) + ({ 3*\pad},{-\pad}) + ({-\th},{-3*\th}) $) --
    ($ (Q1A2) + ({-3*\pad},{\pad}) + ({-\th},{-3*\th}) $) -- cycle;

    \fill[blue!80, fill opacity=0.22, draw opacity=0, rounded corners=8pt]
    ($ (A1B2) + ({-3*\pad},{\pad}) + ({\th},{3*\th}) $) --
    ($ (Q1B2) + ({ 3*\pad},{-\pad}) + ({\th},{3*\th}) $) --
    ($ (Q1B2) + ({ 3*\pad},{-\pad}) + ({-\th},{-3*\th}) $) --
    ($ (A1B2) + ({-3*\pad},{\pad}) + ({-\th},{-3*\th}) $) -- cycle;
\end{tikzpicture}
\caption{The chain complex of the CSS coupled-layer code, which generalizes the tensor product. The gray stripe is still CSS$_1 \otimes Q_2$. However, a more general mapping cone is constructed using the two chain complexes in the red and blue stripes.}
\label{fig:CoupLayerCSS}
\end{figure}

\subsection{Coupled-Layer Construction}\label{sec:cssclc_clc}
We discuss the coupled-layer construction of the CSS code corresponding to the total complex of Figure~\ref{fig:CoupLayerCSS}, which can be carried out similarly as the tensor product code. For each $q_2 \in Q_2$, introduce a copy of CSS$_1$. For each $a_2\in A_2$, introduce a copy of the Hilbert space $\bm Q_x$ with stabilizer $\bm{\Gamma}_x(\bm{b_1})$ for each $b_1\in B_1$, denoted $\bm{\Gamma}_x^{(a_2)}(\bm{b_1})$. For each $b_2\in B_2$, introduce a copy of Hilbert space $\bm Q_z$ with stabilizer $\bm{\Gamma}_z(\bm{a_1})$ for each $a_1\in A_1$ denoted $\bm{\Gamma}_z^{(b_2)}(\bm{a_1})$. The initial stabilizer group is therefore
\begin{align*}
    \mathcal{S}_0 = 
    \left\langle\,
    \bm{a_1}^{(q_2)}, \bm{b_1}^{(q_2)}, \bm{\Gamma}_x^{(a_2)}(\bm{b_1}), \bm{\Gamma}_z^{(b_2)}(\bm{a_1})
    \,\right\rangle.
\end{align*}
Next, we perform a code switching:
\vspace{-0.1cm}
\begin{itemize}
    \item For each pair $e_x\in \mathcal{E}_x$ and $a_2\in A_2$, add $X$-stabilizer
    \begin{align*}
        \bm\alpha(e_x,a_2) = \bm{\Gamma}_x^{(a_2)}(\bm{e_x}) \prod_{q_2\in a_2} \bm{e_x}^{(q_2)}.
    \end{align*}
    \vspace{-0.5cm}
    \item For each pair $e_z\in \mathcal{E}_z$ and $b_2\in B_2$, add $Z$-stabilizer
    \begin{align*}
        \bm\beta(e_z,b_2) = \bm{\Gamma}_z^{(b_2)}(\bm{e_z}) \prod_{q_2\in b_2} \bm{e_z}^{(q_2)}.
    \end{align*}
\end{itemize}
\vspace{-0.3cm}
$\bm\alpha(e_x,a_2)$ and $\bm\beta(e_z,b_2)$ commute due to the fact that $\bm{a_2}$ and $\bm{b_2}$ commute. These terms span the vector space $\mathcal{E}_x\otimes A_2$ and $\mathcal{E}_z\otimes B_2$ in Figure~\ref{fig:CoupLayerCSS}, respectively. Upon enforcing these stabilizers, the commuting terms in $\mathcal{S}_0$ remain, which are
\vspace{-0cm}
\begin{itemize}
    \item For each pair $a_1\in A_1$ and $q_2\in Q_2$, $X$-stabilizer
    \begin{align*}
        \bm\xi(a_1,q_2) = \bm{a_1}^{(q_2)} \prod_{b_2\ni q_2} \bm{\Gamma}_z^{(b_2)}(\bm{a_1}).
    \end{align*}
    \vspace{-0.5cm}
    \item For each pair $b_1\in B_1$ and $q_2\in Q_2$, $Z$-stabilizer
    \begin{align*}
        \bm\zeta(b_1,q_2) = \bm{b_1}^{(q_2)} \prod_{a_2\ni q_2} \bm{\Gamma}_x^{(a_2)}(\bm{b_1}).
    \end{align*}
\end{itemize}
\vspace{-0.2cm}
It is obvious that both terms are indeed in $\mathcal{S}_0$, and they mutually commute because $\bm{a_1}$ and $\bm{b_1}$ commute. To see that $\bm\xi(a_1, q_2)$ commutes with the code switching, we only need to check between $\bm\xi(a_1, q_2)$ and $\bm\beta(e_z, b_2)$. Their supports only overlap when $\bm{a_1}$ and $\bm{e_z}$ overlap and $q_2\in b_2$, in which case
\begin{align*}
    [\bm\xi(a_1, q_2), \bm\beta(e_z, b_2)] = [\bm{a_1}^{(q_2)} \bm{\Gamma}_z^{(b_2)}(\bm{a_1}),\, \bm{\Gamma}_z^{(b_2)}(\bm{e_z}) \bm{e_z}^{(q_2)}] = 0,
\end{align*}
where the last equality is due to $\bm{\Gamma}_z$ preserving the commutation relation. Similarly, $\bm\zeta(b_1, q_2)$ commutes with the code switching. Hence, we indeed have a stabilizer group
\begin{align}\label{eq:mc_CSS_stabilizer}
    \mathcal{S} = 
    \left\langle\,
    \bm\alpha(e_x, a_2),\,
    \bm\beta(e_z, b_2),\,
    \bm\xi(a_1, q_2),\,
    \bm\zeta(b_1, q_2)
    \,\right\rangle.
\end{align}
We enumerate the logicals of this code in Appendix~\ref{sec:CSSlogical}.

It is worth noting that if there is a product $\prod_{b_1\in \nu} \bm{b_1}$ of $Z$-stabilizers for some vector $\nu\in B_1$, which commutes with all the $X$-excitations $\bm{e_x}$, then the product $\prod_{b_1 \in \nu} \bm{\Gamma}_x^{(a_2)}(\bm{b_1})$ in each layer $a_2$ also commutes with the code switching. This is analogous to a metacheck in the conventional tensor product~\cite{zhang2026coupled}. In terms of commutative diagram, this is a modification of the mapping cone in Figure~\ref{fig:MappingConeX},

\begin{equation*}
\begin{tikzpicture}[scale=0.8, every node/.style={scale=1}]
    \def\L{-3.5}
    \def\C{0}
    \def\R{3.5}
    \def\shadedwidth{60}
    \def\shadedheight{160}
    \def\H{1}
    \def\HH{2}
    \def\HHH{3}

    \def\shift{0.3} 

    \node[text=red!80!black] (A) at (\L, \H) {$A$};
    \node (Q) at (0 , 0) {$Q$};
    \node[text=blue!80!black] (B) at (\R, -\H) {$B$};

    \node[text=red!80!black] (aux0) at (\L, -\H) {$\mathcal{E}_x$};
    \node (aux1) at (0, -\HH) {$Q_x$};
    \node[text=blue!80!black] (aux2) at (\R, -\HHH) {$M_x$};

    \node[font=\Large] at ($(Q)!0.5!(aux1)$) {$\circlearrowleft$};

    \node[font=\Large] at (2.6, -2.1) {$\circlearrowleft$};

    \draw[->] (A) -- (Q);
    \draw[->] (Q) -- (B); 
    \draw[->] (aux0) -- (aux1);
    \draw[->] (aux1) -- (aux2);

    \draw[->] (aux0) -- (Q); 
    \draw[->] (aux1) -- (B); 
    \draw[dashed, <-] (aux2) -- (B);

\end{tikzpicture}
\end{equation*}
where the bottom row is also a chain complex, and the dashed line does not participate in defining stabilizers. Similarly for the reflected diagram. In this case, the commutative diagram Figure~\ref{fig:CoupLayerCSS} is augmented to Figure~\ref{fig:CoupLayerCSSaug}.

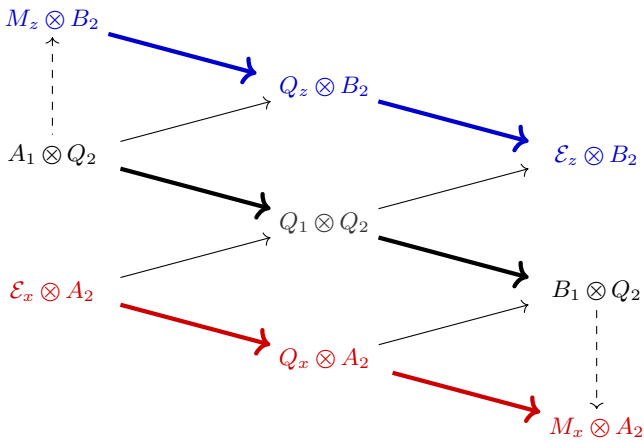
\begin{figure}[h]
\centering
\hspace{-0.9cm}
\begin{tikzpicture}[scale=0.9]
    \def\L{-4}
    \def\C{0}
    \def\R{4}
    \def\shadedwidth{60}
    \def\shadedheight{160}
    \def\H{1}
    \def\HH{2}
    \def\HHH{3}
    \def\offs{2}

    \node  (A1Q2) at (\L,  \H) {};
    \node at ($(\L,  \H) + (0, 0)$) {$A_{1}\otimes Q_{2}$};

    \node (Q1A2) at (\L, -\H) {};
    \node[text=red!80!black] at ($(\L, -\H) +(0, 0)$) {$\mathcal{E}_{x}\otimes A_{2}$};
    
    \node[text=blue!80!black] (MzB2) at ($(\L,  \HHH) + (0, 0)$) {$M_{z}\otimes B_{2}$};

    \node[text=blue!80!black]  (A1B2) at (\C,  \HH) {$Q_{z}\otimes B_{2}$};
    \node[text=black!80]      (Q1Q2) at (\C,       0.0) {$Q_1 \otimes Q_{2}$};
    \node[text=red!80!black](B1A2) at (\C, -\HH) {$Q_{x}\otimes A_{2}$};

    \node[text=blue!80!black] at ($(\R,  \H) + (0, 0)$) {$\mathcal{E}_{z}\otimes B_{2}$};
        
    \node[text=black] (B1Q2) at ($(\R, -\H) + (0, 0)$) {$B_{1}\otimes Q_{2}$};

    \node[text=red!80!black] (MxA2) at ($(\R,  -\HHH) + (0, 0)$) {$M_{x}\otimes A_{2}$};

    \draw[->] (-3, 1.2) -- (A1B2) node[midway, above] {};
    \draw[->, ultra thick] (-3, 0.8) -- (Q1Q2) node[midway, above] {};
    \draw[->] (-3, -0.8) -- (Q1Q2) node[midway, above] {};
    \draw[->,red!80!black, ultra thick] (-3, -1.2) -- (B1A2) node[midway, above] {};
    \draw[->,red!80!black, ultra thick] (1, -2.2) -- (MxA2) node[midway, above] {};
    \draw[->,dashed] (B1Q2) -- (MxA2) node[midway, above] {};

    \draw[->,blue!80!black, ultra thick] (A1B2) -- (3, 1.2) node[midway, above] {};
    \draw[->,blue!80!black, ultra thick] (MzB2) -- (-1, 2.2) node[midway, above] {};
    \draw[<-,dashed] (MzB2) -- (-4,1.3) node[midway, above] {};
    
    \draw[->] (Q1Q2) -- (3, 0.8) node[midway, above] {};
    \draw[->,ultra thick] (Q1Q2) -- (3, -0.8) node[midway, above] {};
    \draw[->] (B1A2) -- (3, -1.2) node[midway, above] {};

    \def\pad{0.35}   
    \def\th{0.21}    

\end{tikzpicture}
\caption{Under the coupled layer construction, low-weight $X$-operators labeled by $M_z\otimes B_2$ and $Z$-operators $M_x\otimes A_2$ can arise in the stabilizer group. These are analogs of the metachecks in tensor product.}
\label{fig:CoupLayerCSSaug}
\end{figure}

\subsection{Concatenated and Subsystem Coupled-Layer Codes}

In this section, we relate the coupled-layer code to a generalization of concatenated codes and a related subsystem code. This generalizes the relation between the tensor product code, the concatenated code and the subsystem tensor product code given in our previous work Ref.~\cite{zhang2026coupled}.

First, we recall the code concatenation of two CSS codes. We call CSS$_1$ the outer code, and CSS$_2$ the inner code. The concatenated code is constructed by replacing physical qubits of CSS$_2$ by the logical qubits of CSS$_1$, and the Pauli $X$ and $Z$ on physical qubits of CSS$_2$ by logical $X$ and $Z$ of CSS$_1$. The result is a CSS code with Hilbert space $Q_1\otimes Q_2$. Denoting the $X$-logicals in CSS$_1$ by $\mathcal{X}(\mu):= \prod_{q_1\in \mu} X_{q_1}$, for $\mu \in \Ker d_1^T $; and $Z$-logicals by $\mathcal{Z}(\nu):= \prod_{q_1\in \nu} Z_{q_1}$, for $\nu \in \Ker \delta_1^T $, the stabilizer group of the concatenated code is
\begin{align*}
    &\Big\langle\,
    \bm{a_1}^{(q_2)},\,
    \bm{b_1}^{(q_2)},\,
    \prod_{q_2\in a_2}\mathcal{X}(\mu)^{(q_2)},\,
    \prod_{q_2\in b_2}\mathcal{Z}(\nu)^{(q_2)}
    \\&
    \hspace{40pt}\,\Big|\,
    a_1\in A_1,\,
    b_1\in B_1,\,
    \mu \in \Ker d_1^T,\,
    \nu \in \Ker \delta_1^T
    \,\Big\rangle
\end{align*}
where $\mathcal{X}(\mu)^{(q_2)}$ and $\mathcal{Z}(\nu)^{(q_2)}$ are CSS$_1$ logical operators in layer $q_2$. The stabilizers $\prod_{q_2\in a_2}\mathcal{X}(\mu)^{(q_2)}$ and $\prod_{q_2\in b_2} \mathcal{Z}(\nu)^{(q_2)}$ have large weight if the logicals of CSS$_1$ have large weight, which makes them potentially not LDPC. In Ref.~\cite{zhang2026coupled}, we showed gauging these stabilizers results in exactly the tensor product code, which is LDPC.

The analogue of the concatenated code in relation to the CSS coupled-layer code is the following. Instead of using all the $X$-logicals $\Ker d_1^T$ and $Z$-logicals $\Ker \delta_1^T$ in CSS$_1$, we only use a subset of the logicals in constructing the stabilizer group. That is, choose a subspace of $X$-logicals $L_X\subseteq \Ker d^T_1$, and a subspace of $Z$-logicals $L_Z \subseteq \Ker \delta_1^T$, we define the stabilizer group
\begin{align*}
    \mathcal{S}_{\text{concat}} = 
    &\Big\langle\,
    \bm{a_1}^{(q_2)},\,
    \bm{b_1}^{(q_2)},\,
    \prod_{q_2\in a_2}\mathcal{X}(\mu)^{(q_2)},\,
    \prod_{q_2\in b_2}\mathcal{Z}(\nu)^{(q_2)}
    \\&
    \hspace{40pt}\,\Big|\,
    a_1\in A_1,\,
    b_1\in B_1,\,
    \mu \in L_X,\,
    \nu \in L_Z
    \,\Big\rangle
\end{align*}
This stabilizer group clearly has fracton-like behavior in the following sense. If there exists an $X$-logical $\mu \in \Ker d^T_1$ in CSS$_1$ such that $\mu \notin L_X + \Imaa \delta_1$, and there is no corresponding conjugate $Z$-logical in $L_Z$, then $\mathcal{X}(\mu)^{(q_2)}$ is a non-trivial $X$-logical of $\mathcal{S}_{\text{concat}}$ for each $q_2$. Moreover, $\mathcal{X}(\mu)^{(q_2)}$ is fully supported in the single layer $q_2$, and can be deformed using $\bm{a_1}^{(q_2)}$, but can never move to another layer.

Analogous to the usual concatenated code, if the logicals in $L_X$ and $L_Z$ have large supports, then we might need to perform code surgery to recover the LDPC property. In order to do this, we choose a pair of mapping cones as in Figure~\ref{fig:MappingConeX} and Figure~\ref{fig:MappingConeZ}, such that
\begin{align*}
    \epsilon_x(\Ker \epsilon_x') = L_X
    \hspace{15pt}
    \text{and}
    \hspace{15pt}
    \epsilon_z(\Ker \epsilon_z') = L_Z
\end{align*}
We then use these mapping cones to construct the complex in Figure~\ref{fig:CoupLayerCSS}, leading exactly to the CSS coupled-layer code, with stabilizer group in Eq.~\ref{eq:mc_CSS_stabilizer}. In this way, we see $\bm \alpha(e_x, a_2)$ and $\bm \beta(e_z, b_2)$ product exactly to $\prod_{q_2\in a_2}\mathcal{X}(\mu)^{(q_2)}$ and $\prod_{q_2\in b_2}\mathcal{Z}(\nu)^{(q_2)}$ by canceling ancillas.

On the other hand, instead of adding ancillas, and imposing the Gauss laws $\bm \alpha(e_x, a_2)$ and $\bm \beta(e_z, b_2)$, one can directly measure the operators 
\begin{align*}
    & \bm{e_x}^{(a_2)} := \prod_{q_2 \in a_2} \bm{e_x}^{(q_2)} =
    \prod_{q_1 \in e_x} \prod_{q_2\in a_2} X_{q_1,q_2},\\
    & \bm{e_z}^{(b_2)} := \prod_{q_2 \in b_2} \bm{e_z}^{(q_2)} =
    \prod_{q_1 \in e_z} \prod_{q_2\in b_2} Z_{q_1,q_2}.
\end{align*}
These operators do not commute with $\bm{a_1}^{(q_2)}$ and $\bm{b_1}^{(q_2)}$ in $\mathcal{S}_{\text{concat}}$, hence they do not define a stabilizer group together. Instead, they form the gauge group of a subsystem code~\cite{poulin2005stabilizer}
\begin{align*}
    \mathcal{G} = \langle\,
    i ,\,
    \bm{a_1}^{(q_2)}, \, \bm{b_1}^{(q_2)}, \, 
    \bm{e_x}^{(a_2)}, \,
    \bm{e_z}^{(b_2)} 
    \,\rangle.
\end{align*}
The relation of this subsystem code and previous codes is via gauge fixing the subsystem code\footnote{As an unfortunate terminology, it is worth noting that the word ``gauge" in gauge group and gauge fixing used in the context of subsystem codes has no relation to the concept of gauging in Sec.~\ref{sec:review_g}.}. Operationally, we keep add operators $\bm{a_1}^{(q_2)}$ and  $\bm{b_1}^{(q_2)}$, and remove all operators in $\mathcal{G}$ that do not commute, keeping only the commuting ones. This gives the stabilizer group (modulo $i$)
\begin{align*}
    &\Big\langle\, 
    \bm{a_1}^{(q_2)},\,
    \bm{b_1}^{(q_2)},\,
    \prod_{e_x \in \mu} \bm{e_x}^{(a_2)}\, ,
    \prod_{e_z \in \nu} \bm{e_z}^{(b_2)}
    \,|\,\\
    & \hspace{35pt}
    a_1\in A_1,\,
    b_1\in B_1,\,
    \mu \in \epsilon_x^{-1}(\Ker d^T),\,
    \nu \in \epsilon_z^{-1}(\Ker \delta ^T)
    \,\Big\rangle
\end{align*}
If the maps are chosen such that $\epsilon_x^{-1}(\Ker d^T_1) = \Ker \epsilon_x'$ and $\epsilon_z^{-1}(\Ker \delta_1^T) = \Ker \epsilon_z'$, then we see this is precisely the stabilizer group $\mathcal{S}_{\text{concat}}$.

To summarize, we have defined a generalization of concatenated code, called the concatenated coupled-layer code, by only using a subset of logical qubits of the outer code to replace the physical qubits of the inner code. This potentially leads to a CSS code that is not LDPC. Gauging the high-weight stabilizers, recovers the CSS coupled-layer code. Moreover, the generalized concatenated code can be thought of as a specific gauge fixing of the subsystem coupled-layer code.

\subsection{Example: X-Cube}\label{sec:cssclc_xc}

We illustrate how the X-cube model is a CSS coupled-layer code. This particular construction has appeared in \cite{Liu23} in relation to the bulk-boundary correspondence of fracton phases \cite{Liu23,schuster2023holographic,Liu26}.

We take CSS$_1$ to be the classical code defined in Eq.~\ref{eq:XcubepregaugeIsingterms}, which we repeat here
\begin{align*}
  \bm b_{v,x} &=    \begin{tikzpicture}
    [baseline=0cm, scale = 0.8]
        \node at (0,0) {};
        \node at (-0.7, 0) {$Z$};
        \node at (0.7,0) {$Z$};
        \node at (0, 0.7) {};
        \node at (0,-0.7) {};
        \draw[step = 2] (-1.2,-1.2) grid (1.2,1.2);
    \end{tikzpicture} & \bm b_{v,y}&=\begin{tikzpicture}
    [baseline=0cm, scale = 0.8]
        \node at (0,0) {};
        \node at (-0.7, 0) {};
        \node at (0.7,0) {};
        \node at (0, 0.7) {$Z$};
        \node at (0,-0.7) {$Z$};
        \draw[step = 2] (-1.2,-1.2) grid (1.2,1.2);
    \end{tikzpicture}
\end{align*}
This corresponds to the chain complex
\begin{align*}
    E \xrightarrow{d_1^T} V \oplus V
\end{align*}
Next, we define CSS$_2$ to be a repetition code in the $X$-basis. The chain complex is $\delta_2 = \partial^z$ where
\begin{equation*}
\begin{tikzpicture}[scale=0.8, every node/.style={scale=1}]
    \def\L{-3.5}
    \def\R{3.5}
    \node at (-0.75, -0.05) {$\partial^z:$};
    \node (Q2) at (\R, 0) {$V^z$};
    \node [text=red!80!black] (A2) at (0 , 0) {$E^z$};
    
    \draw[<-] (Q2) -- (A2);
\end{tikzpicture},
\end{equation*}
where $A_2 = E^z$ and $Q_2= V^z$ are edges and vertices of a 1D lattice, respectively. 

The mapping cone for a single layer of CSS$_1$, which partially gauges the logicals (in particular, only pairs of logicals) is constructed in Eq.~\eqref{eqn:Xcube_MC}. We now use this mapping cone to obtain the X-cube model. We augment this into the following mapping cone:

\begin{equation*}
\begin{tikzpicture}[scale=0.8, every node/.style={scale=1}]
    \def\L{-3.5}
    \def\R{3.5}
    \def\H{1}
    \def\HH{2}

    \node (Ex) [text=red!80!black] at (\L, -\H) {$P \otimes E^z$};
    \node (Qx) at (0, -\HH) {$V \otimes  E^z$};
    \node (Q1) at (0, 0) {$E \otimes V^z$};
    \node [text=blue!80!black] (B1) at (\R , -\H) {$(V \oplus V)\otimes V^z$};
    
    \draw[->] (Q1) -- (B1);
    \draw[->] (Ex) -- (Q1);
    \draw[->] (Ex) -- (Qx);
    \draw[->] (Qx) -- (B1);

\end{tikzpicture}
\end{equation*}
The qubits can be identified with edges of the 3D cubic lattice in the $x-y$ plane and the $z$-direction, respectively. The $X$-stabilizers corresponds to cubes
\begin{align*}
    \bm\alpha(p, e^z) 
    \hspace{5pt}
    =
    \hspace{5pt}
    \begin{tikzpicture}[scale = 1, baseline = 0.1cm, scale = 0.8]
    \coordinate (O) at (0,0,0); 
    \coordinate (A) at (2,0,0);
    \coordinate (B) at (2,2,0);
    \coordinate (C) at (0,2,0);
    \coordinate (D) at (0,0,2);
    \coordinate (E) at (2,0,2);
    \coordinate (F) at (2,2,2);
    \coordinate (G) at (0,2,2);
    %
    \node  at (1, 0, 0) {$X$};
    \node  at (1, 0, 2) {$X$};
    \node at (0, 0, 1) {$X$};
    \node at (2, 0, 1) {$X$};
    \node at (0, 1, 0) {$X$};
    \node at (0, 1, 2) {$X$};
    \node at (2, 1, 0) {$X$};
    \node at (2, 1, 2) {$X$};
    \node  at (1, 2, 0) {$X$};
    \node  at (1, 2, 2) {$X$};
    \node at (0, 2, 1) {$X$};
    \node at (2, 2, 1) {$X$};
%
    \draw [dotted] (O) -- (A);
    \draw (A) -- (B) -- (C); 
    \draw [dotted] (O) -- (C);
    \draw (D) -- (E) -- (F) -- (G) -- cycle; 
    \draw [dotted] (O) -- (D); 
    \draw (A) -- (E);
    \draw (B) -- (F);
    \draw (C) -- (G);
    \end{tikzpicture}
\end{align*}
While the two $Z$-stabilizers are defined for each vertex as
\begin{align*}
    \bm\zeta(v_x, v^z) 
    \hspace{5pt}
    =
    \hspace{5pt}
    \begin{tikzpicture}
    [baseline=0.cm, scale = 0.8]
        \node at (0,0) {};
        \node at (-0.7, 0) {$Z$};
        \node at (0.7,0) {$Z$};
        \node at (0, 0.7) {$Z$};
        \node at (0,-0.7) {$Z$};
        \draw (0,0,-1.2) -- (0,0,1.2)
        (0,-1.2,0) -- (0,1.2,0)
        (-1.2,0,0) -- (1.2,0,0);
    \end{tikzpicture}
    \hspace{30pt}
    \bm\zeta(v_y, v^z) 
    \hspace{5pt}
    =
    \hspace{5pt}
    \begin{tikzpicture}
    [baseline=0.cm, scale = 0.8]
        \node at (0,0) {};
        \node at (0,-0.7, 0) {$Z$};
        \node at (0,0.7,0) {$Z$};
        \node at (0,0, 0.7) {$Z$};
        \node at (0,0,-0.7) {$Z$};
        \draw (0,0,-1.2) -- (0,0,1.2)
        (0,-1.2,0) -- (0,1.2,0)
        (-1.2,0,0) -- (1.2,0,0);
    \end{tikzpicture}
\end{align*}
Together they generate the stabilizer group of X-cube model.

Now, to perform the coupled-layer construction, the checks of the repetition code CSS$_2$ tells us that we should actually perform the partial gauging simultaneously for two consecutive stacks. This means that we are gauging not just pairs of line symmetries in a single $xy$ plane, but pairs of such pairs for every adjacent plane in the $z$ direction. Thus, we are gauging four line symmetries that bound a stack of cubes where the stacking is either in the $x$ or $y$ direction.

Hence, starting with a stack of CSS$_1$ given by the cochain complex
\begin{equation*}
\begin{tikzpicture}[scale=0.8, every node/.style={scale=1}]
    \def\L{-3.5}
    \def\R{3.5}
    \def\H{2.5}

    \node (Ex)  at (\L, \H) {$E \otimes V^z$};
    \node (Qx) [text=blue!80!black]at (0, \H) {$(V \oplus V) \otimes C^z$};
    
    \draw[->] (Ex) -- (Qx);

\end{tikzpicture},
\end{equation*}
which consists of 1D Ising models along every $x$ and $y$ line of the cubic lattice. In addition, for each pair $v\in V$ and $e^z \in E^z$ (i.e. every $z$-edge of the cubic lattice), introduce a $Z$-ancilla. The stabilizer group is
\begin{align*}
    \mathcal{S}_0 = \langle \bm b_{v,x}^{(v^z)}, 
    \bm b_{v,y}^{(v^z)}, Z_{v, e^z} \rangle,
\end{align*}
Code switching is given by the twelve-body term
\begin{align*}
    \bm\alpha(p,e^z) =  \bm{\Gamma}^{(e^z)}_x(\bm e_{x,p})\prod_{v^z \in e^z}\bm e_{x,p}^{(v^z)} 
\end{align*} 
This implements the Gauss law of the aforementioned gauging. After the code switching, the commuting terms in $\mathcal{S}_0$ are 
\begin{align*}
    &\bm\zeta(v_x ,v^z) = \bm b_{v,x}^{(v^z)} \prod_{e^z \ni v^z}Z_{v, e^z},\\
    &\bm\zeta(v_y, v^z) = \bm b_{v,y}^{(v^z)} \prod_{e^z \ni v^z}Z_{v, e^z}.
\end{align*}
Together with $\alpha(p,e^z)$, these generate the stabilizer group of the X-cube model.

As discussed in the general case, the following string logicals of the X-cube model are generated. Let $l$ be a straight line in 2D extending in the $x$ or $y$-direction, then $\prod_{v\in l}Z_{v, \Tilde{e}}$
is a $Z$-logical of the X-cube in $\mathcal{S}_0$, whose end point correspond to lineons.

We remark that by simply replacing the three repetition codes $E^\alpha \rightarrow V^\alpha$ for $\alpha = x,y,z$ into an arbitrary classical code, the resulting chain complex is exactly the 3-fold product generalization of X-cube presented in ~\cite{tan2025fracton,RakovszkyLDPC2}. Thus the above  generalized X-cube 3-fold product is also a CSS coupled-layer code. However, we note that \cite{RakovszkyLDPC2} presents a different coupled-layer construction which is more akin to p-string condensation \cite{MaLakeChenHermele2017,Vijay17}.

Next, we discuss the subsystem version of the X-cube model, which is given by
\begin{align*}
    \mathcal{G} = 
    \Big\langle i, \,\bm b_{v,x}^{(v^z)}, \,
    \bm b_{v,y}^{(v^z)}, \, \bm a_{e^z}^{(e_{x,p})}
    \,\Big \rangle,
\end{align*}
\vspace{-0.5cm}
\begin{align*} \begin{tikzpicture}
    [baseline=-0.2cm, scale = 0.8]
        \node at (0, 2){$\underline{\, \bm b_{v,x}^{(\tilde v)} \,}$};
        \node at (0,0) {};
        \node at (-0.7, 0) {$Z$};
        \node at (0.7,0) {$Z$};
        \node at (0, 0.7) {};
        \node at (0,-0.7) {};
        \draw[step = 2] (-1.2,-1.2) grid (1.2,1.2);
    \end{tikzpicture}
    \hspace{10pt}
    ,
    \hspace{15pt}
    \begin{tikzpicture}
    [baseline=-0.2cm, scale = 0.8]
        \node at (0, 2, 0) {$\underline{\, \bm b_{v,y}^{(\tilde v)}\,} $};
        \node at (0,0) {};
        \node at (0,-0.7, 0) {};
        \node at (0,0.7,0) {};
        \node at (0,0, 0.7) {$Z$};
        \node at (0,0,-0.7) {$Z$};
        \draw (0,0,-1.2) -- (0,0,1.2)
        (0,-1.2,0) -- (0,1.2,0);
    \end{tikzpicture}
    \hspace{10pt}
    ,
    \hspace{15pt}
    \begin{tikzpicture}[scale = 1, baseline = 0.3cm, scale = 0.8]
    \node at (1.1,3.0,1) {$\underline{\,\bm a_{\tilde e}^{(e_{x,p})}\,}$};
    \coordinate (O) at (0,0,0); 
    \coordinate (A) at (2,0,0);
    \coordinate (B) at (2,2,0);
    \coordinate (C) at (0,2,0);
    \coordinate (D) at (0,0,2);
    \coordinate (E) at (2,0,2);
    \coordinate (F) at (2,2,2);
    \coordinate (G) at (0,2,2);
    %
    \node  at (1, 0, 0) {$X$};
    \node  at (1, 0, 2) {$X$};
    \node at (0, 0, 1) {$X$};
    \node at (2, 0, 1) {$X$};
    \node  at (1, 2, 0) {$X$};
    \node  at (1, 2, 2) {$X$};
    \node at (0, 2, 1) {$X$};
    \node at (2, 2, 1) {$X$};
%
    \draw [dotted] (O) -- (A);
    \draw (A) -- (B) -- (C); 
    \draw [dotted] (O) -- (C);
    \draw (D) -- (E) -- (F) -- (G) -- cycle; 
    \draw [dotted] (O) -- (D); 
    \draw (A) -- (E);
    \draw (B) -- (F);
    \draw (C) -- (G);
    \end{tikzpicture}
\end{align*}
We perform the gauge fixing by keeping the $Z$-stabilizers $\bm b_{v,x}^{(\tilde v)}$ and $\bm b_{v,y}^{(\tilde v)}$, and keeping only the commuting $X$-operators in $\mathcal{G}$. In order to do this, consider straight lines $l^*$ on the dual of 2D lattice, extending in either $x$ or $y$-direction, then the operator $\mathcal{X}_{l^*, \tilde e} = \prod_{\tilde e \in l^*} \bm a_{\tilde e}^{(e_{x,p})}$
commutes with every term in $\mathcal{G}$. For example, if $l^*$ extends in the $x$-direction, then the operator is
\begin{align*}
    \mathcal{X}_{l^*, \tilde e} 
    \hspace{5pt}
    =
    \hspace{5pt}
    \begin{tikzpicture}[scale = 1, baseline = 0.4cm, scale = 0.8]
    \coordinate (O) at (0,0,0); 
    \coordinate (A) at (2,0,0);
    \coordinate (B) at (2,2,0);
    \coordinate (C) at (0,2,0);
    \coordinate (D) at (0,0,2);
    \coordinate (E) at (2,0,2);
    \coordinate (F) at (2,2,2);
    \coordinate (G) at (0,2,2);
    %
    \node  at (1, 0, 0) {$X$};
    \node  at (1, 0, 2) {$X$};
    \node  at (1, 2, 0) {$X$};
    \node  at (1, 2, 2) {$X$};
%
    \draw [dotted] (O) -- (A);
    \draw (B)-- (C); 
    \draw [dotted] (O) -- (C);
    \draw (D) -- (E) -- (F) -- (G) -- cycle; 
    \draw [dotted] (O) -- (D); 
    \draw (B) -- (F);
    \draw (C) -- (G);
    %
    %
    %
    \coordinate (O) at (2,0,0); 
    \coordinate (A) at (4,0,0);
    \coordinate (B) at (4,2,0);
    \coordinate (C) at (2,2,0);
    \coordinate (D) at (2,0,2);
    \coordinate (E) at (4,0,2);
    \coordinate (F) at (4,2,2);
    \coordinate (G) at (2,2,2);
    %
    \node  at (3, 0, 0) {$X$};
    \node  at (3, 0, 2) {$X$};
    \node  at (3, 2, 0) {$X$};
    \node  at (3, 2, 2) {$X$};
%
    \draw [dotted] (O) -- (A);
    \draw (A) -- (B) -- (C); 
    \draw [dotted] (O) -- (C);
    \draw (D) -- (E) -- (F) -- (G) -- cycle; 
    \draw [dotted] (O) -- (D); 
    \draw (A) -- (E);
    \draw (B) -- (F);
    \draw (C) -- (G);
     \draw[dotted]
     (-1,1,1) -- (5,1,1);
     \node[fill = black,circle, inner sep = 0.7pt] at (0,1,1) {};
     \node[fill = black,circle, inner sep = 0.7pt] at (4,1,1) {};
     \node at (5.3,1,1) {$l^*$}; \end{tikzpicture}
\end{align*}
We thus obtain the corresponding concatenated coupled-layer code
\begin{align*}
    \mathcal{S}_{\text{concat}} = \langle\,
    \bm b_{v,x}^{(\tilde v)}, \,
    \bm b_{v,y}^{(\tilde v)}, \,
    \mathcal{X}_{l^*, \tilde e}
    \,\rangle
\end{align*}
However, the $X$-stabilizers have large weight. In order to recover an LDPC code, we need to gauge these stabilizers. We do this by introducing $Z$-ancilla on all vertical edges, and impose the Gauss law
\begin{equation*}
    \begin{tikzpicture}[scale = 1, baseline = 0.1cm, scale = 0.8]
    \coordinate (O) at (0,0,0); 
    \coordinate (A) at (2,0,0);
    \coordinate (B) at (2,2,0);
    \coordinate (C) at (0,2,0);
    \coordinate (D) at (0,0,2);
    \coordinate (E) at (2,0,2);
    \coordinate (F) at (2,2,2);
    \coordinate (G) at (0,2,2);
    %
    \node  at (1, 0, 0) {$X$};
    \node  at (1, 0, 2) {$X$};
    \node at (0, 0, 1) {$X$};
    \node at (2, 0, 1) {$X$};
    \node at (0, 1, 0) {$X$};
    \node at (0, 1, 2) {$X$};
    \node at (2, 1, 0) {$X$};
    \node at (2, 1, 2) {$X$};
    \node  at (1, 2, 0) {$X$};
    \node  at (1, 2, 2) {$X$};
    \node at (0, 2, 1) {$X$};
    \node at (2, 2, 1) {$X$};
%
    \draw [dotted] (O) -- (A);
    \draw (A) -- (B) -- (C); 
    \draw [dotted] (O) -- (C);
    \draw (D) -- (E) -- (F) -- (G) -- cycle; 
    \draw [dotted] (O) -- (D); 
    \draw (A) -- (E);
    \draw (B) -- (F);
    \draw (C) -- (G);
    \end{tikzpicture}
\end{equation*}
which is nothing but the X-cube stabilizer $\bm\alpha(p, \tilde e)$. It follows the $Z$-stabilizers $\bm b_{v,x}^{(\tilde v)}$ and $\bm b_{v,y}^{(\tilde v)}$ in $\mathcal{S}_{\text{concat}}$ couples minimally to the ancillas, giving the X-cube $Z$-stabilizers $\bm \xi(v_x, \tilde v)$ and $\bm \xi(v_y, \tilde v)$. Together, we recover the X-cube model.

\subsection{Example: 4D Quantum Code}\label{sec:cssclc_4dqc}

As a second example, we take both inputs to be quantum codes. Hence both $X$ and $Z$ excitations are condensed. 

We take both CSS$_1$ and CSS$_2$ to be TC. For the $X$-type mapping cone, we choose it as the partial gauging which gauges all pairs of $X$-logical extending in $x$-direction, which is specified by
\begin{itemize}
    \item Excitations: spanned by vertices $\{e_{x,v} \,|\, v\in V\}$, with
    \begin{align*}
        \bm e_{x,v} = \begin{tikzpicture}
    [baseline=0cm, scale = 0.8]
        \node at (0,0) {};
        \node at (-0.7, 0) {};
        \node at (0.7,0) {};
        \node at (0, 0.7) {$X$};
        \node at (0,-0.7) {$X$};
        \node at (0.2,0.2) {$v$};
        \draw[step = 2] (-1.2,-1.2) grid (1.2,1.2);
    \end{tikzpicture}
    \end{align*}
    \item Algebra preserving map $\bm{\Gamma}_x$,
    \begin{align*}
        \bm e_{x,v}        \hspace{5pt}
        \xrightarrow{\bm{\Gamma}_x}   \hspace{5pt}
        \begin{tikzpicture}[baseline=0cm, scale = 0.8]
        \node at (1.2, 1) {$X$};
        \node at (2.8, 1) {$X$};
        \node at (2.3,0.3) {$v$};
        \draw[ step = 2] (0.5,0) grid (3.5,2);
        \draw (2,-1.5) -- (2,0);
    %
    %
    \end{tikzpicture}
    \hspace{25pt}
        \bm b_{p}
        \hspace{5pt}
        \xrightarrow{\bm{\Gamma}_x}
        \hspace{5pt}
        \begin{tikzpicture}[scale = 0.8, baseline = 1.5cm]
            \node at (1,1) {$Z$};
            \node at (1,3) {$Z$};
            \node at (1.5, 0.4) {$p$};
            \draw[ step = 2] (0,0) grid (2,4);
        \end{tikzpicture}
    \end{align*}
\end{itemize}
Indeed, this is equivalent to the following mapping cone 
\begin{equation*}
    \begin{tikzpicture}[scale=0.8, every node/.style={scale=1}]
    \def\L{-3.5}
    \def\C{0}
    \def\R{3.5}
    \def\shadedwidth{60}
    \def\shadedheight{160}
    \def\H{1}
    \def\HH{2}

    \def\shift{0.3} 

    \node[text=red!80!black] (A) at (\L, \H) {$V$};
    \node (Q) at (0 , 0) {$E$};
    \node[text=blue!80!black] (B) at (\R, -\H) {$P$};

    \node[text=red!80!black] (aux0) at (\L, -\H) {$V$};
    \node (aux1) at (0, -\HH) {$P$};

    \draw[->] (A) -- (Q);
    \draw[->] (Q) -- (B); 
    \draw[->] (aux0) -- (aux1);

    \draw[->] (aux0) -- (Q); 
    \draw[->] (aux1) -- (B);

    \node at ($(A)!0.5!(Q) + (0,\shift)$) {$\delta$};
    \node at ($(B)!0.5!(Q) + (0,\shift)$) {$d^T$};
    \node at ($(aux0)!0.5!(aux1) - (0,\shift)$) {$\epsilon_x'$};

    \node at ($(aux0)!0.5!(Q) + (0.2,\shift)$) {$\epsilon_x$};
    \node at ($(aux1)!0.5!(B) - (0.3,\shift)$) {$d'^T$};
\end{tikzpicture}
\end{equation*}
where the maps $\delta$ and $d^T$ are the usual coboundary maps from vertices to edges, and edges to plaquettes. The rest of the linear maps are given by
\begin{align*}
    &\epsilon_x:
    \hspace{5pt}
    \begin{tikzpicture}[scale = 0.6, baseline = -0.1cm]
        \draw[step = 2] (-1.5, -1.5)grid (1.5, 1.5);
        \node[fill = red,circle, inner sep = 1pt] at (0,0) {};
    \end{tikzpicture}
    \hspace{5pt}
    \mapsto
    \hspace{5pt}
    \begin{tikzpicture}[scale = 0.6, baseline = -0.1cm]
        \draw (0,-1) -- (0,1);
        \node at (0.2, 0) {$e$};
        \node[fill = red,circle, inner sep = 1pt] at (0,-1) {};
    \end{tikzpicture}
    \hspace{5pt}
    +
    \hspace{5pt}
    \begin{tikzpicture}[scale = 0.6, baseline = -0.1cm]
        \draw (0,-1) -- (0,1);
        \node at (0.2, 0) {$e$};
        \node[fill = red,circle, inner sep = 1pt] at (0,1) {};
    \end{tikzpicture}\\
    & \epsilon_x':
    \hspace{5pt}
    \begin{tikzpicture}[scale = 0.6, baseline = -0.1cm]
        \draw[step = 2] (-1.5, -1.5)grid (1.5, 1.5);
        \node[fill = red,circle, inner sep = 1pt] at (0,0) {};
    \end{tikzpicture}
    \hspace{5pt}
    \mapsto
    \hspace{5pt}
    \begin{tikzpicture}[scale = 0.6, baseline = 0.5cm]
        \draw[step = 2] (0, 0) grid (2,2);
        \node at (1, 1) {$p$};
        \node[fill = red,circle, inner sep = 1pt] at (0,0) {};
    \end{tikzpicture}
    \hspace{5pt}
    +
    \hspace{5pt}
    \begin{tikzpicture}[scale = 0.6, baseline = 0.5cm]
        \draw[step = 2] (0, 0) grid (2,2);
        \node at (1, 1) {$p$};
        \node[fill = red,circle, inner sep = 1pt] at (2,0) {};
    \end{tikzpicture}\\
    & d':
    \hspace{5pt}
    \begin{tikzpicture}[scale = 0.6, baseline = 0.5cm]
        \draw[step = 2] (0, 0) grid (2,2);
        \node[fill = red,circle, inner sep = 1pt] at (1,1) {};
    \end{tikzpicture}
    \hspace{5pt}
    \mapsto
    \hspace{5pt}
    \begin{tikzpicture}[scale = 0.6, baseline = 0.5cm]
        \draw[step = 2] (0, 0) grid (2,2);
        \node[fill = red,circle, inner sep = 1pt] at (1,1) {};
        \node at (1.4, 1.4) {$p$};
    \end{tikzpicture}
    \hspace{5pt}
    +
    \hspace{5pt}
    \begin{tikzpicture}[scale = 0.6, baseline = 1cm]
        \draw[step = 2] (0, 0) grid (2,4);
        \node[fill = red,circle, inner sep = 1pt] at (1,1) {};
        \node at (1, 3) {$p$};
    \end{tikzpicture}
\end{align*}
and extend by linearity.

Similarly for the $Z$-type mapping cone, we choose it to be the partial gauging which gauges all pairs of $Z$-logicals extending in $y$-direction, which is given 
\begin{itemize}
    \item Excitations: spanned by plaquettes $\{e_{z,p} \,|\, p\in P\}$, with
    \begin{align*}
        \bm e_{z,p} = \begin{tikzpicture}[scale = 0.8, baseline = 0.7cm]
            \node at (0,1) {$Z$};
            \node at (2,1) {$Z$};
            \node at (1,1) {$p$};
            \draw[ step = 2] (0,0) grid (2,2);
        \end{tikzpicture}
    \end{align*}
    \item Algebra preserving map $\bm{\Gamma}_z$,
    \begin{align*}
        \bm e_{z,p}
        \hspace{5pt}
        \xrightarrow{\bm{\Gamma}_z}
        \hspace{5pt}
        \begin{tikzpicture}[scale = 0.8, baseline = 0.7cm]
            \node at (2,0) {$Z$};
            \node at (2,2) {$Z$};
            \node at (1,1) {$p$};
            \draw[ step = 2] (0,0) grid (2,2);
    \end{tikzpicture}
    \hspace{25pt}
        \bm a_{v}
        \hspace{5pt}
        \xrightarrow{\bm{\Gamma}_z}
        \hspace{5pt}
        \begin{tikzpicture}[scale = 0.8, baseline = -0.1cm]
            \draw[step = 2] (-1.5, -1.5)grid (1.5, 1.5);
            \node at (0,0) {$X$};
            \node at (1.5,0) {$X$};
            \node at (0.4, 0.4) {$v$};
        \end{tikzpicture}
    \end{align*}
\end{itemize}
This is equivalent to the mapping cone
\begin{equation*}
    \begin{tikzpicture}[scale=0.8, every node/.style={scale=1}]
    \def\L{-3.5}
    \def\C{0}
    \def\R{3.5}
    \def\shadedwidth{60}
    \def\shadedheight{160}
    \def\H{1}
    \def\HH{2}

    \def\shift{0.3}
    
    \node [text=red!80!black] (A) at (\L, \H) {$V$};
    \node (Q) at (0 , 0) {$E$};
    \node [text=blue!80!black] (B) at (\R, -\H) {$P$};

    \node (aux0) at (0, \HH) {$V$};
    \node [text=blue!80!black] (aux1) at (\R, \H) {$P$};

    \draw[<-] (A) -- (Q);
    \draw[<-] (Q) -- (B); 
    \draw[<-] (aux0) -- (aux1);

    \draw[->] (aux0) -- (A); 
    \draw[->] (aux1) -- (Q);

    \node at ($(A)!0.5!(Q) - (0,\shift)$) {$\delta^T$};
    \node at ($(B)!0.5!(Q) - (0,\shift)$) {$d$};
    \node at ($(aux0)!0.5!(aux1) + (0,\shift)$) {$\epsilon_z'$};

    \node at ($(aux0)!0.5!(A) + (0,\shift)$) {$\delta'^T$};
    \node at ($(aux1)!0.5!(Q) + (0,\shift)$) {$\epsilon_z$};
\end{tikzpicture}
\end{equation*}
where $d$ and $\delta^T$ are again the usual boundary maps from plaquettes to edges, and edges to vertices. The rest of the linear maps are
\begin{align*}
    &\epsilon_z:
    \hspace{5pt}
    \begin{tikzpicture}[scale = 0.6, baseline = 0.5cm]
        \draw[step = 2] (0, 0) grid (2,2);
        \node[fill = red,circle, inner sep = 1pt] at (1,1) {};
    \end{tikzpicture}
    \hspace{5pt}
    \mapsto
    \hspace{5pt}
    \begin{tikzpicture}[scale = 0.6, baseline = 0.5cm]
        \draw[step = 2] (0, 0) grid (2,2);
        \node[fill = red,circle, inner sep = 1pt] at (1,1) {};
        \node at (0,1) {$e$};
    \end{tikzpicture}
    \hspace{5pt}
    +
    \hspace{5pt}
    \begin{tikzpicture}[scale = 0.6, baseline = 0.5cm]
        \draw[step = 2] (0, 0) grid (2,2);
        \node[fill = red,circle, inner sep = 1pt] at (1,1) {};
        \node at (2,1) {$e$};
    \end{tikzpicture}\\
    & \epsilon_z':
    \hspace{5pt}
    \begin{tikzpicture}[scale = 0.6, baseline = 0.5cm]
        \draw[step = 2] (0, 0) grid (2,2);
        \node[fill = red,circle, inner sep = 1pt] at (1,1) {};
    \end{tikzpicture}
    \hspace{5pt}
    \mapsto
    \hspace{5pt}
    \begin{tikzpicture}[scale = 0.6, baseline = 0.5cm]
        \draw[step = 2] (0, 0) grid (2,2);
        \node[fill = red,circle, inner sep = 1pt] at (1,1) {};
        \node at (2,0) {$v$};
    \end{tikzpicture}
    \hspace{5pt}
    +
    \hspace{5pt}
    \begin{tikzpicture}[scale = 0.6, baseline = 0.5cm]
        \draw[step = 2] (0, 0) grid (2,2);
        \node[fill = red,circle, inner sep = 1pt] at (1,1) {};
        \node at (2,2) {$v$};
    \end{tikzpicture}\\
    & \delta': 
    \hspace{5pt}
    \begin{tikzpicture}[scale = 0.6, baseline = -0.1cm]
        \draw[step = 2] (-1.5, -1.5)grid (1.5, 1.5);
        \node[fill = red,circle, inner sep = 1pt] at (0,0) {};
    \end{tikzpicture}
    \hspace{5pt}
    \mapsto
    \hspace{5pt}
    \begin{tikzpicture}[scale = 0.6, baseline = -0.1cm]
        \draw[step = 2] (-1.5, -1.5)grid (1.5, 1.5);
        \node[fill = red,circle, inner sep = 1pt] at (0,0) {};
        \node at (0.3, 0.3) {$v$};
    \end{tikzpicture}
    \hspace{5pt}
    +
    \hspace{5pt}
    \begin{tikzpicture}[scale = 0.6, baseline = -0.1cm]
        \draw[step = 2] (-1.5, -1.5)grid (1.5, 1.5);
        \node[fill = red,circle, inner sep = 1pt] at (0,0) {};
        \node at (1.5, 0) {$v$};
    \end{tikzpicture}
\end{align*}
and extend by linearity.

Recall that using the mapping cones that gauges all the logicals in each layer, the usual coupled-layer construction that recovers the tensor product gives the 4D toric code. Instead, we will use the above mapping cones to perform the coupled-layer construction. We label the vertices, edges and plaquettes in the two TC factors by $v_i$, $e_i$ and $p_i$ for $i = 1,2$ respectively. For each $e_2$, introduce a copy of TC. For each 
$v_2$, introduce a copy of the Hilbert space with qubits on plaquettes with terms $\bm{\Gamma}_x^{(v_2)}(\bm b_{p_1})$ for each plaquette $p_1$. Similarly, for each $p_2$, introduce a copy of the Hilbert space with qubits on vertices with terms $\bm{\Gamma}_z^{(p_2)}(\bm a_{v_1})$ for each vertex $v_1$. We see the total Hilbert space is
\begin{align*}
    Q \,=\, (E_1\otimes E_2) \oplus (V_1\otimes P_2) \oplus (P_1\otimes V_2)
\end{align*}
which exactly spans the $2$-cells of the $4$D hypercubic lattice. The stabilizer group of the layered system is
\begin{align*}
    \mathcal{S}_0 = 
    \langle \,
    \bm a_{v_1}^{(e_2)},\,
    \bm b_{p_1}^{(e_2)},\,
    \bm{\Gamma}_x^{(v_2)}(\bm b_{p_1}),\,
    \bm{\Gamma}_z^{(p_2)}(\bm a_{v_1})
    \,\rangle
\end{align*}

The code switching is given by
\begin{align*}
    &\bm \alpha(v_1, v_2) = \bm{\Gamma}_x^{(v_2)}(\bm e_{x,{v_1}}) \prod_{e_2 \in v_2} \bm e_{x,v_1}^{(e_2)}\\
    &\bm \beta(p_1, p_2) = \bm{\Gamma}_z^{(p_2)}(\bm e_{z,{p_1}}) \prod_{e_2 \in p_2} \bm e_{z,p_1}^{(e_2)}
\end{align*}
The $X$-stabilizers $\bm \alpha(v_1, v_2)$ span all $0$-cells, and $Z$-stabilizers $\bm \beta(p_1, p_2)$ span all $4$-cells of the hypercubic lattice. The remaining commuting terms $\mathcal{S}_0$ are
\begin{align*}
    & \bm \xi(v_1, e_2) = \bm a_{v_1}^{(e_2)} \prod_{p_2\ni e_2} \bm{\Gamma}_z^{(p_2)}(\bm a_{v_1})\\
    & 
    \bm \zeta(p_1, e_2) = \bm b_{p_1}^{(e_2)} \prod_{v_2\ni e_2} \bm{\Gamma}_x^{(v_2)}(\bm b_{p_1})
\end{align*}
The $X$-stabilizers $\bm \xi(v_1, e_2)$ and $Z$-stabilizers $\bm \zeta(p_1, e_2)$ each span half of $1$-cells and $3$-cells respectively.

This code can be understood in the following way. Starting with $2$D square lattice, with a copy of TC placed on each edge. For each vertex, we partial gauge pairs of $X$-logicals along $x$-direction simultaneously in all four copies of TC on the edge neighboring the vertex. Similarly, for each plaquette, we partial gauge pairs of $Z$-logicals along $y$-direction simultaneously in all four copies of TC on the edge neighboring the plaquette. Due to the difficulty in drawing the final form of the code, we instead opt to express it in polynomial form~\cite{haah2013commuting}. Since the qubits live on $2$-cells of the hypercubic lattice, there are $6$ qubits per unit cell, and up to an ordering, the polynomials of the stabilizers are given by
\onecolumngrid
\begin{align*}
    \begin{blockarray}{cccccc}
    \bm\alpha(v_1,v_2) & \bm\xi(v_1,e_2^{H}) &\bm \xi(v_1,e_2^{V})& \bm \beta(p_1,p_2)&\bm\zeta(p_1,e_2^{H}) &\bm \zeta(p_1,e_2^{V}) \\
    \vspace{-0.2cm}\\
    \begin{block}{(cccccc)}
      1 + \bar x 
      & 0 
      & 0    
      & 0
      & 0
      & 0\\
      0
      & 1+\bar x 
      & 0   
      & 0
      & 0
      & 0\\
      0 
      & 0 
      & 1+\bar x     
      & 0
      & 0
      & 0\\
      (1+\bar y)(1+\bar z) 
      & 1+\bar y 
      & 0 
      & 0
      & 0
      & 0\\
      (1+\bar y)(1+\bar w)
      & 0 
      & 1+\bar y 
      & 0
      & 0
      & 0\\
      0
      & (1+\bar x)(1+\bar w)
      & (1+\bar x)(1+\bar z)
      & 0
      & 0
      & 0\\
      \vspace{-0.3cm}\\
      \cline{1-6}
      \vspace{-0.3cm}\\
      0
      & 0
      & 0
      & 0
      & (1+y)(1+z)
      & (1+y)(1+w)\\
      0
      & 0
      & 0
      & 0
      & 1+y
      & 0\\
      0
      & 0
      & 0
      & 0
      & 0
      & 1+y\\
      0
      & 0
      & 0
      & (1+x)(1+w)
      & 1+x
      & 0\\
      0
      & 0
      & 0
      & (1+x)(1+z)
      & 0
      & 1+x\\
      0 
      & 0 
      & 0 
      & 1+y
      & 0
      & 0\\
    \end{block}
  \end{blockarray}
\end{align*}
\twocolumngrid
Note that we do not claim the codes constructed in this way are guaranteed to have good distances. This is not surprising, since the vanilla tensor product can also give rise to codes with constant distance (see Ref.~\cite{zhang2026coupled} for further discussion). Indeed, for the above code, one can check that the following are weight-5 logicals:
\begin{align*}
    (\,
    1 \hspace{5pt} 1+\bar z \hspace{5pt} 1+\bar w \hspace{5pt} 0 \hspace{5pt} 0 \hspace{5pt} 0 
    \hspace{5pt} |
    \hspace{5pt} 0 \hspace{5pt} 0 \hspace{5pt} 0 \hspace{5pt} 0 \hspace{5pt} 0 
    \hspace{5pt} 0 
    \,)^T,\\
    (\,
    0 \hspace{5pt} 0 \hspace{5pt} 0 \hspace{5pt} 0 \hspace{5pt} 0 
    \hspace{5pt} 0 
    \hspace{5pt} |
    \hspace{5pt} 0 \hspace{5pt} 1+w \hspace{5pt} 1+z \hspace{5pt} 0 \hspace{5pt} 0 
    \hspace{5pt} 1 
    \,)^T.
\end{align*}
This makes the above code have constant distance. Nevertheless, one can show that adding the above logicals as stabilizers results in a code with distance linear in the system size.

\section{Coupled-Layer Codes}\label{sec:clc}

In this section, we present a product construction inspired by the coupled-layer construction in Sec.~\ref{sec:cssclc_clc}. From the coupled-layer perspective, there is no reason that the code being stacked and condensed must be a CSS code. Indeed, as we will present below, the two input codes in this construction are not required to be CSS. However, the idea is the same as before. Excitations in layers of the first code are being condensed, while the second code provides a recipe for the condensation to commute. For this reason, we refer the resulting code as the \textit{coupled-layer code}. Finally, we will see the coupled-layer code indeed reduce to the CSS coupled-layer code as a special case.

In order to generalize to the non-CSS case, we recall the mapping cone used to generalize tensor product contains the same data as the algebra-preserving map $\bm \Gamma_x$ (resp. $\bm\Gamma_z$) acting on the excitation algebra $\bm{\mathcal{A}}_x$ (resp. $\bm{\mathcal{A}}_z$). While the former can only be defined when a chain complex is available, the latter is not. Instead, they can be defined on any operator algebra. Treating the operator maps as fundamental objects, it is then straightforward to generalize to non-CSS codes.

We replace CSS$_1$ with an arbitrary stabilizer code with stabilizer group $\mathcal{H} = \langle\bm h_t \rangle$ defined on qubits $Q_1$, and replace CSS$_2$ with a stabilizer group generated from a set that only consists of pure $X$,$Y$, or $Z$ type stabilizers. We will denote these sets $\{\mathcal{X}_\alpha\}$, $\{\mathcal{Y}_\beta\}$ and $\{\mathcal{Z}_\gamma\}$, defined on an auxiliary set of qubits $Q_2$, where $\alpha$, $\beta$, $\gamma$ are indices for these Pauli operators.

Next, instead of defining the mapping cone, we specify the condensation algebras and algebra-preserving maps as appropriate in the non-CSS case,
\vspace{-0.5em}
\begin{itemize}
    \item Excitations: For each $\mathcal{X}_\alpha$ (resp. $\mathcal{Y}_\beta$, $\mathcal{Z}_\gamma$), choose a set of $X$ (resp. $Y$, $Z$) type Pauli operators $\{
    \bm a_{\alpha, i}
    \}_{i\in I_\alpha}$ $\left( \text{resp. } \{
    \bm b_{\beta, j}
    \}_{j\in J_\beta}\right.$, $ \left.\{
    \bm c_{\gamma,k}
    \}_{k\in K_\gamma} \right)$ in the system $\mathcal{H}$. $I_\alpha$, $J_\beta$ and $K_\gamma$ are index sets associated by $\mathcal{X}_\alpha$, $\mathcal{Y}_\beta$ and $\mathcal{Z}_\gamma$ respectively. These excitations define the excitation algebras
    \begin{gather*}
        \hspace{25pt}
        \mathcal{A}^{(\alpha)}_x := \langle\, \bm h_t ,\, \bm a_{\alpha, i} \,\rangle\\
        \hspace{25pt}
        \left(\text{resp. }
        \hspace{5pt}
        \mathcal{A}^{(\beta)}_y := \langle \bm h_t ,\, \bm b_{\beta, j} \,\rangle
        ,\hspace{10pt}
        \mathcal{A}_z^{(\gamma)} := \langle\,
        \bm h_t ,\, \bm c_{\gamma, k} \,\rangle
        \right)
    \end{gather*}
    \vspace{-0.5em}
    \item Algebra-preserving maps: For each $\mathcal{X}_\alpha$ (resp. $\mathcal{Y}_\beta$, $\mathcal{Z}_\gamma$), choose map $\bm{\Gamma}^{(\alpha)}_x$ $\left( \text{resp. }\bm{\Gamma}^{(\beta)}_y\right.$, $\left.\bm{\Gamma}^{(\gamma)}_z\right)$ defined on the excitation algebra $\mathcal{A}^{(\alpha)}_x$ (resp. $\mathcal{A}^{(\beta)}_y$, $\mathcal{A}^{(\gamma)}_z$), which takes each generator to a Pauli operator on an auxiliary Hilbert space $Q_x^{(\alpha)}$ (resp. $Q_y^{(\beta)}$, $Q_z^{(\gamma)}$), while preserving commutation relations. 
\end{itemize}
\vspace{-0.5em}

Now we describe stacking and code switching.  
For each $q_2 \in Q_2$, introduce a copy of $\mathcal{H}$. For each $\mathcal{X}_\alpha$, $\mathcal{Y}_\beta$ and $\mathcal{Z}_\gamma$, introduce a copy of $\left\{\bm{\Gamma}^{(\alpha)}_x(\bm h_t)\right\}$, $\left\{\bm{\Gamma}^{(\beta)}_y(\bm h_t)\right\}$ and $\left\{\bm{\Gamma}^{(\gamma)}_z(\bm h_t)\right\}$ respectively. Thus the stacked stabilizer group is defined on the total Hilbert space
\begin{align*}
    \left(Q_1\otimes Q_2\right)
    \oplus
    \left(\bigoplus_\alpha Q_x^{(\alpha)}\right)
    \oplus
    \left(\bigoplus_\beta Q_y^{(\beta)}\right)
    \oplus
    \left(\bigoplus_\gamma Q_z^{(\gamma)}\right)
\end{align*}
whose stabilizer group is generated by
\begin{align*}
    \mathcal{S}_0 = \left\langle \bm h_t^{(q)}, \bm{\Gamma}^{(\alpha)}_x(\bm h_t), \bm{\Gamma}^{(\beta)}_y(\bm h_t), \bm{\Gamma}^{(\gamma)}_z(\bm h_t) \right\rangle,
\end{align*}
where $\bm h_t^{(q)}$ is the stabilizer $\bm h_t$ in layer $q$.

Next we perform code switching by enforcing 
\begin{align*}
    \bm h(i, \alpha) = \bm{\Gamma}^{(\alpha)}_x(\bm a_{\alpha, i}) \prod_{q\in \mathcal{X}_\alpha} \bm a_{\alpha, i}^{(q)},
\end{align*} 
where $\bm a_{\alpha, i}^{(q)}$ is the excitation $\bm a_{\alpha, i}$ in layer $q$. Physically, this condenses the excitations created by $\bm a_{\alpha, i}$ in multiple layers determined by the $X$-check $\mathcal{X}_\alpha$. Similarly, we condense the $Y$ and $Z$-type excitations by enforcing 
\begin{align*}
    &\bm h(j, \beta) = \bm{\Gamma}^{(\beta)}_y (\bm b_{\beta, j}) \prod_{q\in \mathcal{Y}_\beta} \bm b_{\beta, j}^{(q)},
    \\
    &\bm h(k, \gamma) = \bm{\Gamma}^{(\gamma)}_z(\bm c_{\gamma, k}) \prod_{q\in \mathcal{Z}_\gamma} \bm c_{\gamma, k}^{(q)}.
\end{align*}
We need to check that these condensations commute with each other. Given $\bm h(i, \alpha)$ and $\bm h(i', \alpha')$, they only overlap on layers of $\mathcal{H}$ labeled by supp$(\mathcal{X}_\alpha) \cap$ supp$(\mathcal{X}_{\alpha'})$. But they only act as Pauli $X$ on the intersection, hence they commute. Similarly for $\bm h(j, \beta)$ and $\bm h(k, \gamma)$. Next, consider $\bm h(i, \alpha)$ and $\bm h(j, \beta)$. Again they overlap on layers supp$(\mathcal{X}_\alpha) \cap$ supp$(\mathcal{Y}_\beta)$. But $\mathcal{X}_\alpha$ and $\mathcal{Y}_\beta$ always have even overlap since they were chosen to commute with each other. Hence $\bm h(i, \alpha)$ and $\bm h(j, \beta)$ always commute. 

After the code switching, only commuting terms in $\mathcal{S}_0$ remain. Consider
\begin{align*}
    \bm h(t,q) = \bm h_t^{(q)}
    \prod_{\mathcal{X}_\alpha \ni q} \bm{\Gamma}^{(\alpha)}_x (\bm h_t) 
    \prod_{\mathcal{Y}_\beta \ni q} \bm{\Gamma}^{(\beta)}_y (\bm h_t)
    \prod_{\mathcal{Z}_\gamma \ni q} \bm{\Gamma}^{(\gamma)}_z (\bm h_t),
\end{align*}
which obviously belongs to $\mathcal{S}_0$. To see it commutes with $\bm h(i,\alpha)$, note that they overlap if and only if $q\in \mathcal{X}_\alpha$ and $a_{\alpha,i}$ overlaps with $\bm h_t$, in which case 
\begin{align*}
    \text{sgn}(\bm h(t,q), \bm h(i,\alpha)) = \text{sgn}(\bm h_t, \bm a_{\alpha,i})\text{sgn}(\bm{\Gamma}^{(\alpha)}_x (\bm h_t), \bm{\Gamma}^{(\alpha)}_x (\bm a_{\alpha,i})).
\end{align*}
But $\bm{\Gamma}^{(\alpha)}_x$ preserves commutation relation, so this is always $1$. Similarly, $\bm h(t,q)$ commutes with $\bm h(j,\beta)$ and $\bm h(k,\gamma)$. We have
\begin{align*}
    \mathcal{S} \supseteq \langle 
    \bm h(t,q), \bm h(i,\alpha), \bm h(j,\beta), \bm h(k,\gamma)
    \rangle.
\end{align*}
We take the above four terms as defining the stabilizer group of the coupled-layer code. Note that from the code switching perspective, there can be more terms in $\mathcal{S}_0$ that remain in $\mathcal{S}$, as in the case of the tensor product. For example, if certain product $\prod_{t\in T} \bm h_t$ commutes with $\bm a_{\alpha,i}$ for all $i$ with some fixed $\alpha$, then $\prod_{t\in T} \bm{\Gamma}^{(\alpha)}_x(\bm h_t)$ remains in $\mathcal{S}$. In particular, if there is a nontrivial relation between the stabilizers $\prod_{t\in T} \bm h_t = 1$, then $\prod_{t\in T} \bm{\Gamma}^{(\alpha)}_x (\bm h_t)$ is in $\mathcal{S}$ for all $\alpha$. We do not attempt to write down all such terms in the most general form.

This construction reduces to the CSS coupled-layer code in the following way: $\mathcal{H}$ is taken to be CSS$_1$ with stabilizers $\bm{a_1}$ and $\bm{b_1}$. The auxiliary operators $\{\mathcal{Y}_\beta\}$ are trivial, while $\{\mathcal{X}_\alpha\}$ and $\{\mathcal{Z}_\gamma\}$ generate the stabilizer group of CSS$_2$, hence we replace the label $\alpha$ and $\gamma$ by $a_2$ and $b_2$ respectively. The excitation algebras are independent of $\alpha$ and $\gamma$, that is, $\bm{\mathcal{A}}^{(a_2)}_x \equiv \bm{\mathcal{A}}_x$ and $\bm{\mathcal{A}}^{(b_2)}_z \equiv \bm{\mathcal{A}}_z$ for all $a_2$ and $b_2$. We label the excitations in $\bm{\mathcal{A}}_x$ and $\bm{\mathcal{A}}_z$ by $\{\bm e_{x,i}\}$ and $\{\bm e_{z,k}\}$ respectively. Similarly, the algebra-preserving maps are chosen $\bm \Gamma_x^{(a_2)} \equiv \bm \Gamma_x$ for all $a_2$ (resp. $\bm \Gamma_z^{(b_2)} \equiv \bm \Gamma_z$ for all $b_2$) such that all $\bm{a_1}$ (resp. $\bm{b_1}$) are mapped to identity. In this way, as discussed in Sec.~\ref{sec:review_mc}, the pair $(\bm{\mathcal{A}}_x, \bm\Gamma_x)$ and $(\bm{\mathcal{A}}_z, \bm\Gamma_z)$ each defines a mapping cone as in Figure~\ref{fig:MappingConeX} and Figure~\ref{fig:MappingConeZ}. It is straightforward to check
\begin{align*}
    &\bm h(i, a_2) = \bm \alpha(e_{x,i}, a_2)\\
    &\bm h(k, b_2) = \bm \beta(e_{z,k}, b_2)\\
    &\bm h(a_1, q_2) = \bm \xi(a_1, q_2)\\
    &\bm h(b_1, q_2) = \bm \zeta(b_1, q_2)
\end{align*}
Hence the CSS coupled-layer code is indeed recovered.


\subsection{Example: XYZ Product Code} 
In this section, we give a concrete example by showing that the quantum XYZ product code \cite{2022Quant...6..766L} is a coupled-layer code. The XYZ product code was inspired from the Chamon fracton code~\cite{chamon2005quantum}. In particular, the quantum XYZ product takes three classical codes as input, and reproduces the Chamon model when the classical codes are 1D repetition codes. We show how to obtain this construction using our framework.

Let us first review the XYZ product code.  Given three classical codes $\mathcal{C}_i$ with parity checks $H_i: Q_i \rightarrow C_i$. By taking tensor products of the length-1 chain complexes, one obtains a length-4 chain complex. The XYZ product is given by the following diagram,
\begin{align*}
    \begin{tikzpicture}[scale = 1.0]
        \node at (0,0) {$QQQ$};
        \node at (2.5,1) {$CQQ$};    \node at (2.5,0) {$QCQ$};
        \node at (2.5,-1) {$QQC$};
        \node at (5,1) {$CCQ$};
        \node at (5,0) {$CQC$};
        \node at (5,-1) {$QCC$};
        \node at (7.5,0) {$CCC$};
        \node at (3.85,1.2) {$y$};
        \node at (3.8,0.7) {$z$};
        \node at (3.8,0.2) {$x$};
        \node at (3.8,-0.2) {$z$};
        \node at (3.8,-0.7) {$x$};
        \node at (3.85,-1.2) {$y$};
        \draw[->]
        (0.5,0) -- (1.9,1)  node[midway, above] {$x$};
        \draw[->]
        (0.5,0) -- (1.9,0) node[midway, above] {$y$};
        \draw[->]
        (0.5,0) -- (1.9,-1) node[midway, above] {$z$};
        \draw[->]
        (3.1,1) -- (4.5,1);
        \draw[->]
        (3.1,1) -- (4.5,0.1);
        \draw[->]
        (3.1,0) -- (4.5,0.8);
        \draw[->]
        (3.1,0) -- (4.5,-0.8);
        \draw[->]
        (3.1,-1) -- (4.5,-0.1);
        \draw[->]
        (3.1,-1) -- (4.5,-1);
        \draw[->]
        (5.6,1) -- (7,0.1)  node[midway, above] {$z$};
        \draw[->]
        (5.6,0) -- (7,0) node[midway, above] {$y$};
        \draw[->]
        (5.6,-1) -- (7,-0.1) node[midway, above] {$x$};
    \end{tikzpicture}
\end{align*}
where $ABC$ is an abbreviation for $A_1\otimes B_2 \otimes C_3$. One then interprets $QQQ$, $CCQ$, $CQC$, $QCC$ as qubits, and $CQQ$, $QCQ$, $QQC$, $CCC$ as stabilizers. The $x$, $y$, and $z$ indicate which Pauli type a stabilizer acts on the qubits. All four types of stabilizers can be easily written. For example, take $c_1q_2q_3\in CQQ$, the corresponding stabilizer term is 
\begin{align*}
    \bm h_{c_1q_2q_3} = \prod_{q_1\in c_1}X_{q_1q_2q_3} \prod_{c_2\ni q_2}Y_{c_1c_2q_3} \prod_{c_3\ni q_3}Z_{c_1q_2c_3}.
\end{align*}
Indeed, when all three inputs are the 1D classical Ising model, the output is Chamon's fracton model.

The coupled-layer construction of the XYZ product is most easily described in terms of products between two classical codes, in particular the hypergraph product and homological product, which we review first. Each of these products is a special case of tensor product at the chain complex level. Given a pair of classical codes $\mathcal{C}_1$ and $\mathcal{C}_2$ with parity checks $H_1: Q_1 \ra C_1$ and $H_2: Q_2 \ra C_2$, each can be viewed as a $2$-term chain complex. Taking the tensor product gives the following chain complex
\begin{equation*}
    \begin{tikzpicture}[scale=0.8, every node/.style={scale=1}]
    \def\L{-4}
    \def\R{4}
    \def\H{1}
    \def\HH{2}
    \def\shift{0.3}
    
    \node (QQ) [text=red!80!black] at (\L, -\H) {$Q_1\otimes Q_2$};
    \node (CQ) at (0, -\HH) {$C_1\otimes Q_2$};
    \node (QC) at (0, 0) {$Q_1\otimes C_2$};
    \node [text=blue!80!black] (CC) at (\R , -\H) {$C_1 \otimes C_2$};
    
    \draw[->] (QQ) -- (QC);
    \draw[->] (QQ) -- (CQ);
    \draw[->] (QC) -- (CC);
    \draw[->] (CQ) -- (CC);
    \node at ($(QC)!0.5!(CC) + (0.1,\shift)$) {$H_1 \otimes id$};
    \node at ($(QQ)!0.5!(CQ) - (0.4,\shift)$) {$H_1 \otimes id$};

    \node at ($(QQ)!0.5!(QC) + (-0.4,\shift)$) {$id \otimes H_2$};
    \node at ($(CQ)!0.5!(CC) - (-0.3,\shift)$) {$id \otimes H_2$};
\end{tikzpicture}
\end{equation*}
where the homological product corresponds to putting the qubits on the middle layer\footnote{The tensor product in our convention would instead put the qubits on $Q_1 \otimes Q_2$}.

The qubits of the homological product are labeled by $(Q_1\otimes C_2) \oplus (C_1\otimes Q_2$), the $X$-stabilizers and $Z$-stabilizers are given respectively as
\begin{align}
    \bm s_{q_1q_2}' = \prod_{c_1\ni q_1} X_{c_1q_2}\prod_{c_2\ni q_2} X_{q_1 c_2}\\
    \bm s_{c_1c_2}' = \prod_{q_1\in c_1} Z_{q_1c_2}\prod_{q_2\in c_2} Z_{c_1 q_2}
    \label{eq:homologicalproduct}
\end{align}
On the other hand, the hypergraph product only differs from the homological product by taking the transpose of the second parity check before taking the tensor product. That is, the total chain complex is
\begin{equation*}
    \begin{tikzpicture}[scale=0.8, every node/.style={scale=1}]
    \def\L{-4}
    \def\R{4}
    \def\H{1}
    \def\HH{2}
    \def\shift{0.3}
    
    \node (QC) [text=red!80!black] at (\L, -\H) {$Q_1\otimes C_2$};
    \node (CC) at (0, -\HH) {$C_1\otimes C_2$};
    \node (QQ) at (0, 0) {$Q_1\otimes Q_2$};
    \node [text=blue!80!black] (CQ) at (\R , -\H) {$C_1 \otimes Q_2$};
    
    \draw[->] (QC) -- (QQ);
    \draw[->] (QC) -- (CC);
    \draw[->] (QQ) -- (CQ);
    \draw[->] (CC) -- (CQ);
    \node at ($(QQ)!0.5!(CQ) + (0.1,\shift)$) {$H_1 \otimes id$};
    \node at ($(QC)!0.5!(CC) - (0.4,\shift)$) {$H_1 \otimes id$};

    \node at ($(QC)!0.5!(QQ) + (-0.5,\shift)$) {$id \otimes H_2^T$};
    \node at ($(CC)!0.5!(CQ) - (-0.3,\shift)$) {$id \otimes H_2^T$};
\end{tikzpicture}
\end{equation*}
The qubits of the hypergraph product are labeled by $Q_1\otimes Q_2$ and $C_1\otimes C_2$, $X$-stabilizers and $Z$-stabilizers are given respectively as
\begin{align*}
    &\bm t_{q_1c_2}' = \prod_{c_1\ni q_1} X_{c_1c_2}\prod_{q_2\in c_2} X_{q_1q_2}\\
    &\bm t_{c_1q_2}' = \prod_{q_1\in c_1} Z_{q_1q_2}\prod_{c_2\ni q_2} Z_{c_1c_2}
\end{align*}

Now we are ready to describe coupled-layer construction of the XYZ product. This is given by a two-step process. In the first step, we simply take the hypergraph product between $\mathcal{C}_1$ and $\mathcal{C}_2$, followed by a particular Clifford deformation in order to match our final result with the original definition in \cite{2022Quant...6..766L}. The Clifford rotation is given by
\begin{align}
  \prod_{q_1 ,q_2} \frac{X_{q_1q_2}+Y_{q_1q_2}}{\sqrt{2}}   \prod_ {q\in Q} \frac{Y_q+Z_q}{\sqrt{2}}
  \label{eq:Cliffordrotation}
\end{align}
Thus the stabilizers of the hypergraph product in this basis is
\begin{align*}
    &\bm t_{q_1c_2} = \prod_{c_1\ni q_1} X_{c_1c_2}\prod_{q_2\in c_2} Y_{q_1q_2}\\
    &\bm t_{c_1q_2} =\prod_{q_1\in c_1} X_{q_1q_2}\prod_{c_2\ni q_2} Y_{c_1c_2}
\end{align*}
We note that since the hypergraph product is a special case of tensor product with classical codes as inputs, it can be obtained with a coupled-layer construction by following the procedure in Sec.~\ref{sec:review_cltp}. This concludes the first step.

In the second step, we perform another coupled layer construction using the above code and  $\mathcal{C}_3$. In order to define the excitation algebra and the algebra-preserving map, it is insightful to define a second ``dual" code which we will map the current code into. In fact, this turns out to be precisely the homological product Eq.~\eqref{eq:homologicalproduct} after the Clifford transformation Eq.~\eqref{eq:Cliffordrotation}
\begin{align*}
    \bm s_{q_1q_2}=\prod_{c_1\ni q_1} X_{c_1q_2}\prod_{c_2\ni q_2} Y_{q_1 c_2}\\
    \bm s_{c_1c_2} = \prod_{q_1\in c_1} X_{q_1c_2}\prod_{q_2\in c_2} Y_{c_1 q_2}
\end{align*}

Taking $\mathcal{C}_3$ to be in the $Z$ basis, the excitations are single Pauli $Z$, that is, $Z_{q_1q_2}$ and $Z_{c_1c_2}$. The action of algebra-preserving map on excitation algebra $\mathcal{A}_z = \langle\, \bm t_{q_1c_2},\, \bm t_{c_1q_2},\, Z_{q_1q_2}, \,Z_{c_1c_2}  \,\rangle$ is given by
\begin{align*}
    \bm{\Gamma}_z(\bm t_{q_1c_2}) &= Z_{q_1c_2},
    &\bm{\Gamma}_z(\bm t_{c_1q_2}) &= Z_{c_1q_2},\\
    \bm{\Gamma}_z(Z_{q_1q_2}) &= \bm s_{q_1q_2},
    &\bm{\Gamma}_z(Z_{c_1c_2}) &=\bm s_{c_1c_2}.
\end{align*}

Having defined $\mathcal{A}_z$ and $\bm\Gamma_z$, we proceed as usual. For each qubit $q_3\in Q_3$, introduce a copy of the hypergraph product $\langle \bm t_{q_1c_2}, \bm t_{c_1q_2} \rangle$. For each check $c_3\in C_3$, introduce a copy of the dual system $\langle Z_{q_1c_2}, Z_{c_1q_2} \rangle$. The total Hilbert space is thus $QQQ\oplus CCQ \oplus QCC \oplus CQC$, which agrees with that of the XYZ product. The stacked stabilizer group is
\begin{align*}
    \mathcal{S}_0 = \langle \bm t_{q_1c_2}^{(q_3)}, \bm t_{c_1q_2}^{(q_3)}, Z_{q_1c_2c_3}, Z_{c_1q_2c_3}  \rangle,
\end{align*}
We code switch by adding the following terms to the stabilizer group is given by (see, e.g., $\bm h(k,\gamma)$ above)
\begin{align*}
    \bm h_{q_1q_2c_3} 
    :=
    \bm h(q_1q_2, c_3) = 
    \bm s^{(c_3)}_{q_1q_2}
    \prod_{q_3\in c_3} Z_{q_1q_2q_3},\\
    \bm h_{c_1c_2c_3} 
    :=
    \bm h(c_1c_2, c_3) = 
    \bm s^{(c_3)}_{c_1c_2} \prod_{q_3\in c_3} Z_{c_1c_2q_3}.
\end{align*}
Keeping only the commuting terms, following the general construction, this includes (see, e.g., $\bm h(t,q)$ above)
\begin{align*}
    \bm h_{q_1c_2q_3} :=
    \bm h(q_1c_2, q_3) = 
    \bm t_{q_1c_2}^{(q_3)} \prod_{c_3\ni q_3} Z_{q_1c_2c_3},\\
    \bm h_{c_1q_2q_3} := 
    \bm h(c_1q_2, q_3) = 
    \bm t_{c_1q_2}^{(q_3)} \prod_{c_3\ni q_3} Z_{c_1q_2c_3}.
\end{align*}
These four terms precisely give the stabilizers in the XYZ product. 
We now give an example construction of the Chamon model using our condensation procedure.

\subsubsection{Example: Chamon Code}
    The Chamon code is the XYZ product constructed out of three repetition codes. We now show that it arises explicitly from a coupled-layer construction. 
    
    We take the qubits and checks of the repetition code to correspond to vertex and edges of a 1D lattice. The first step is to take the hypergraph product of the first two repetition codes. Qubits live on vertex and faces of the square lattice, and the stabilizers are defined on the edges. After the Clifford deformation, this gives the Wen plaquette model, (which can also be called the XYYX surface code)\cite{Wenplaquette03,XZZX}.
\begin{equation}
    \begin{aligned}
        \bm t_{q_1c_2} &= \begin{tikzpicture}[baseline = 0.8cm, scale = 0.8]
        \node at (1,1) {$X$};
        \node at (3, 1) {$X$};
        \node at (2,0) {$Y$};
        \node at (2,2) {$Y$};
        \draw[step = 2] (0,0) grid (4,2);
    \end{tikzpicture},\\
    \bm t_{c_1q_2} &= \begin{tikzpicture}[baseline = 1.6cm, scale = 0.8]
        \node at (1,1) {$Y$};
        \node at (1, 3) {$Y$};
        \node at (0, 2) {$X$};
        \node at (2,2) {$X$};
        \draw[ step = 2] (0,0) grid (2,4);
    \end{tikzpicture}
    \end{aligned}
\end{equation}

For the second step, we chose the excitations to be a single $Z$ on every qubit (vertices and plaquettes). The action of the algebra-preserving map $\bm{\Gamma}_z$ on the TC terms is
\begin{align*}
\begin{tikzpicture}[baseline = 0.8cm, scale = 0.8]
        \node at (1,1) {$X$};
        \node at (3, 1) {$X$};
        \node at (2,0) {$Y$};
        \node at (2,2) {$Y$};
        \draw[step = 2] (0,0) grid (4,2);
    \end{tikzpicture}
    \hspace{5pt}
    &\xmapsto{\bm{\Gamma}_z}
    \hspace{5pt}
    \begin{tikzpicture}[baseline = 0.8cm, scale = 0.8]
        \node at (0,1) {$Z$};
        \draw (0,0) -- (0,2);      
    \end{tikzpicture}
    \\
      \begin{tikzpicture}[baseline = 1.6cm, scale = 0.8]
        \node at (1,1) {$Y$};
        \node at (1, 3) {$Y$};
        \node at (0, 2) {$X$};
        \node at (2,2) {$X$};
        \draw[ step = 2] (0,0) grid (2,4);
    \end{tikzpicture}
    \hspace{5pt}
    &\xmapsto{\bm{\Gamma}_z}
    \hspace{5pt}
      \begin{tikzpicture}[baseline = 0cm, scale = 0.8]
        \node at (1, 0) {$Z$};
        \draw (0,0) -- (2, 0);      
    \end{tikzpicture}
\end{align*}
and its action on the excitations are
\begin{align*}
    \begin{tikzpicture}
    [baseline = 0cm, scale = 0.8]
        \node at (0,0) {$Z$};
        \draw[step = 2] (-1.2,-1.2) grid (1.2,1.2) ;
    \end{tikzpicture}
    \hspace{5pt}
    &\xmapsto{\bm{\Gamma}_z}
    \hspace{5pt}
         \begin{tikzpicture}
    [baseline = 0cm, scale = 0.8]
        \node at (-0.7,0) {$X$};
        \node at (0.7,0) {$X$};
        \node at (0,-0.7) {$Y$};
        \node at (0,0.7) {$Y$};
        \draw[step = 2] (-1.2,-1.2) grid (1.2,1.2) ;
    \end{tikzpicture}
    \\
    \begin{tikzpicture}
    [baseline = 0.8cm, scale = 0.8]
        \node at (1,1) {$Z$};
        \draw[step = 2] (0,0) grid (2,2) ;
    \end{tikzpicture}
    \hspace{5pt}
    &\xmapsto{\bm{\Gamma}_z}
    \hspace{5pt}
        \begin{tikzpicture}
    [baseline = 0.8cm, scale = 0.8]
        \node at (1,0) {$Y$};
        \node at (0,1) {$X$};
        \node at (2,1) {$X$};
        \node at (1,2) {$Y$};
        \draw[step = 2] (0,0) grid (2,2) ;
    \end{tikzpicture}
\end{align*}
That is, the Wen plaquette is mapped to a trivial code $Z$, while each set of excitations created by $Z$ in the Wen plaquette model is mapped to four charges. One can check this indeed preserves the commutation relations. 

Now we stack the Wen plaquette using the third repetition code, which is in Pauli $Z$ basis. For each qubit in the repetition code, we stack a copy of Wen plaquette model in the $xy$ plane. For each check in the repetition code, we stack a layer of $Z$-paramagnet, whose qubits are placed on the $xy$ and $xz$ plaquettes. Together, the qubits live on the vertices and plaquettes of the cubic lattice, which form the FCC lattice. The condensation terms are
\begin{align*}
    \bm h_{q_1q_2c_3}  &= \begin{tikzpicture}[scale = 1, baseline = 0.5cm, scale = 0.8]
    \coordinate (O) at (0,0,0); 
    \coordinate (A) at (2,0,0);
    \coordinate (B) at (-2,0,0);
    \coordinate (C) at (0,0,2);
    \coordinate (D) at (0,0,-2);
    \coordinate (E) at (0,2,0);
    \coordinate (F) at (2,2,0);
    \coordinate (G) at (-2,2,0);
    \coordinate (H) at (0,2,2);
    \coordinate (I) at (0,2,-2);
    \node  at (0, 1, 1) {$Y$};
    \node  at (0, 1, -1) {$Y$};
    \node at (1, 1, 0) {$X$};
    \node at (-1, 1, 0) {$X$};
    \node at (0, 0, 0) {$Z$};
    \node at (0, 2, 0) {$Z$};
    \draw (A) -- (O);
    \draw [dotted] (O) -- (B);
    \draw (O) -- (C);
    \draw [dotted] (O) -- (D);
    \draw (O) -- (E);
    \draw (F) -- (G);
    \draw (H) -- (E) -- (I);
    \draw 
    (A) -- (F)
    (B) -- (G)
    (C) -- (H);
    \draw [dotted] (D) -- (I);
    \draw (0,1.23, -2) -- (0,2,-2)
    (-2,0,0) -- (-0.77,0,0);
    \end{tikzpicture}
    \\
    \bm h_{c_1c_2c_3}  &=  \begin{tikzpicture}[scale = 1, baseline = 0cm, scale = 0.8]
    \coordinate (O) at (0,0,0); 
    \coordinate (A) at (2,0,0);
    \coordinate (B) at (2,2,0);
    \coordinate (C) at (0,2,0);
    \coordinate (D) at (0,0,2);
    \coordinate (E) at (2,0,2);
    \coordinate (F) at (2,2,2);
    \coordinate (G) at (0,2,2);
    \node  at (1, 1, 0) {$Y$};
    \node  at (1, 1, 2) {$Y$};
    \node at (0, 1, 1) {$X$};
    \node at (2, 1, 1) {$X$};
    \node at (1, 0, 1) {$Z$};
    \node at (1, 2, 1) {$Z$};
    \draw [dotted] (O) -- (A);
    \draw (A) -- (B) -- (C); 
    \draw [dotted] (O) -- (C);
    \draw (D) -- (E) -- (F) -- (G) -- cycle; 
    \draw [dotted] (O) -- (D); 
    \draw (A) -- (E);
    \draw (B) -- (F);
    \draw (C) -- (G);
    \end{tikzpicture}
\end{align*}
and the remaining stabilizer terms are
\begin{align*}
   \bm h_{q_1c_2q_3} &= \begin{tikzpicture}[baseline = -0.3cm, scale = 0.8]
        \node  at (0, 0, 0) {$Y$};
        \node  at (0, 0, 2) {$Y$};
        \node at (-1, 0, 1) {$X$};
        \node at (1, 0, 1) {$X$};
        \node at (0, 1, 1) {$Z$};
        \node at (0, -1, 1) {$Z$};
        \draw [dotted]
        (-1, 0, 0) -- (0, 0, 0);
        \draw (-2, 0, 0) -- (-0.77,0,0);
        \draw 
        (0,0,0) -- (2, 0, 0)
        (-2, 0, 2) -- (2, 0, 2);
        \draw (-2,0,0) -- (-2,0,2)
        (2,0,0) -- (2,0,2);
        \draw [dotted] (0,-1, 0) -- (0,0, 0);
        \draw (0,-2, 0) -- (0, -0.77,0);
        \draw
        (0, 2, 0) -- (0,0,0)
        (0,-2, 2) -- (0,2, 2);
        \draw (0,-2,0) -- (0,-2,2)
        (0,2,0) -- (0,2,2);
        \draw
        (0,0,0) -- (0,0,2);
    \end{tikzpicture}\\
    \bm h_{c_1q_2q_3} &=\begin{tikzpicture}[baseline = 0cm, scale = 0.8]
        \node  at (1, 0, -1) {$Y$};
        \node  at (1, 0, 1) {$Y$};
        \node at (0, 0, 0) {$X$};
        \node at (2, 0, 0) {$X$};
        \node at (1, 1, 0) {$Z$};
        \node at (1, -1, 0) {$Z$};
        \draw (0,0,0) -- (2,0,0);
        \draw [dotted](0, -2, 0) -- (0,0,0);
        \draw 
        (0,0,0) -- (0,2,0)
        (2,-2,0) -- (2,2,0);
        \draw (0,-2,0) -- (0, -0.77,0);
        \draw (0,-2,0) -- (2, -2,0) (0,2,0) -- (2, 2,0);
        \draw [dotted] (0, 0, 0) -- (0,0,-2)
        (0,0,-2) -- (2, 0,-2);
        \draw 
        (0,0,0) -- (0,0,2)
        (2,0,-2) -- (2,0,2);
        \draw
        (2,0.77,0) -- (2,0, -2);
        \draw 
        (0,0,2) -- (2, 0,2);
    \end{tikzpicture}
\end{align*}
These are the stabilizers of the Chamon code. 

Note that the code switching also generates diagonal lines of $Z$ in a fixed layer of paramagnet. This non-local operator is a logical of the Chamon code.


\section{Generalized Quotient Code}\label{sec:gqc}

In this section, we discuss a new quotient construction on Pauli stabilizer codes with certain symmetry. This quotient construction will be applied in Sec.~\ref{sec:embp} to define a generalization of the balanced product construction, therefore we first study it here. More specifically, we allow the symmetry to act by a permutation of qubits, accompanied by an onsite unitary. As we will see below, such a quotient (and product in Sec.~\ref{sec:embp}) naturally produces non-CSS codes even if the input code is CSS.

Given a Pauli stabilizer group $\mathcal{S} = \langle \, \bm s \,\rangle$ generated by a set of stabilizers $\{\bm s\}$, and defined on qubits $Q$, which is not necessarily CSS. We write $\bm s = \prod_{q\in s} \bm s_q$, where $\bm s_q$ is the Pauli operator acting on qubit $q$ in $\bm s$. Now, suppose there exists a right $G$-symmetry on the code that has two components. First, $G$ freely permutes the qubits and the supports of stabilizers. That is, $q \mapsto q\cdot  g$. Second, each $g\in G$ applies a transversal single-qubit Clifford unitary $\bar U_g$, which acts by $U_g$ on all qubits $q \in Q$, satisfying $U_{gg'}  =U_g U_{g'}$. We denote the total action by $\mathcal{U}_g$, and explicitly it acts on stabilizers by
\begin{align*}
    \mathcal{U}_g :
        \,\bm s 
        \,\mapsto \,
        \mathcal{U}_g(\bm s)
        \hspace{25pt}
        (\mathcal{U}_g(\bm s))_{q\cdot g} = U_g^\dagger \bm s_q U_g
\end{align*}
We can quotient $\mathcal S$ by $G$ in the following way. The qubits and stabilizers of $\mathcal S/G$ are simply orbits of the qubits and stabilizers in $\mathcal S$. Fixing a set of representatives $\{\tilde q\}$ and $\{\tilde{\bm s}\}$ of $Q/G$ and $\mathcal{S}/G$ respectively, define the stabilizer of quotient code
\begin{align}\label{eq:em_quotient}
    \tilde{\bm s}_G := \prod_{\tilde q} \bm T^{\tilde{\bm s}}_{\tilde q},
    \hspace{35pt}
    \bm T^{\tilde{\bm s}}_{\tilde q} = \prod_{\substack{q\in \tilde{\bm s} \\ q\sim \tilde q}} U_{g_q} \tilde{\bm s}_q U_{g_q}^\dagger.
\end{align}
where $g_q$ is the unique group element relating $q$ and its representative $\tilde q$ by $q =  \tilde q \cdot g_q$. In words, $\tilde{\bm s}_G$ is a Pauli operator acting on qubits $Q/G = \{\tilde q\}$, with $\bm T^{\tilde{\bm s}}_{\tilde q}$ the single-qubit Pauli in $\tilde{\bm s}_G$ acting on qubit $\tilde q$. Concretely, $\bm T^{\tilde{\bm s}}_{\tilde q}$ is a product of all Pauli's $\tilde{\bm s}_q$ with acts on qubits in the orbit $[\tilde q]$ after conjugating by an appropriate $U_g$. Note that this definition of $\tilde{\bm s}_G$ is ambiguous, but only up to a phase factor. For instance, it could happen that $U_{g_q}$ is the Hadamard $\mathfrak{h}$ which switches $X$ and $Z$, so that the terms defining $\bm T_{\tilde q}^{\tilde{\bm s}}$ are not commuting, and we need to choose an ordering. Nonetheless, this does not affect the commutation relations of stabilizers, and one can show that they will always commute:
\begin{lemma}
    Given $\tilde{\bm s}$ and $\tilde{\bm t}$, pairs $S = \{ (q, q') \,|\, q\in \tilde{\bm s}, q'\in \tilde{\bm t},\, q \sim q' \}$ is in bijection with pairs $P = \{(g, q) \,|\, q \in \supp{\tilde{\bm s}} \cap \supp{\tilde{\bm t} \cdot g} \}$
\end{lemma}
\begin{proof}
    Given a pair $(q, q')$ in $S$, there exists a unique $g$ such that $q = q' \cdot g$, hence $q \in \supp{\tilde{\bm s}} \cap \supp{\tilde{\bm t} \cdot g}$, and we map $(q,q') \mapsto (g, q)$. Conversely, given $(g,q)$ in $P$, we have $q\cdot g^{-1} \in \tilde{\bm t}$, so we map $(g, q) \mapsto (q, q\cdot g^{-1})$. These maps are inverses of each other, giving bijections.
\end{proof}
Assume we have fixed an ordering to define $\tilde{\bm s}_G$ and $\tilde{\bm t}_G$, we can now show they commute as follows
\begin{align*}
    \sgn{&\tilde{\bm s}_G, \tilde{\bm t}_G} = \prod_{(q, q') \in S} \sgn{U_{g_q} \tilde{\bm s}_q U_{g_q}^\dagger, U_{g_{q'}} \tilde{\bm t}_{q'} U_{g_{q'}}^\dagger}\\
    &= \prod_{(g, q) \in P} \sgn{U_{g_q} \tilde{\bm s}_q U_{g_q}^\dagger, U_{g_q g^{-1}} \tilde{\bm t}_{q\cdot g^{-1}} U_{g_q g^{-1}}^\dagger}\\
    &= \prod_{(g, q) \in P} \sgn{\tilde{\bm s}_q, U_g^\dagger \tilde{\bm t}_{q\cdot g^{-1}} U_g}\\
    &= \prod_{g\in G}
    \,
    \prod_{q \in \supp{\tilde{\bm s}} \cap \supp{\tilde{\bm t} \cdot g} } \sgn{\tilde{\bm s}_q, U_g^\dagger \tilde{\bm t}_{q\cdot g^{-1}} U_g}\\
    &= \prod_{g\in G}
    \sgn{\tilde{\bm s}, \tilde{\bm t} \cdot g} =1
\end{align*}

We may choose different representatives to define the quotient $\mathcal S/G$, but the outcome is independent of this choice. To see this explicitly, assume we choose a different qubit representative $\tilde q' \sim \tilde q$, related by $\tilde q = \tilde q' \cdot g$. Then
\begin{align*}
    \bm T_{\tilde q'}^{\tilde{\bm s}} = \prod_{\substack{q\in \tilde{\bm s} \\ q\sim \tilde q'}}  U_{gg_q} \tilde{\bm s}_q U_{gg_q}^\dagger = U_g \bm T_{\tilde q}^{\tilde{\bm s}} U_g^\dagger.
\end{align*}
We thus see that choosing a different representative for the orbit $\tilde q$ amounts to applying a single unitary $U_g$ on this qubit. On the other hand, if we choose a different representative for stabilizers $\tilde{\bm s} \sim \tilde{\bm s}' $ related by $\tilde{\bm s} = \tilde{\bm s}' \cdot g$, then
\begin{align*}
    \bm T_{\tilde q}^{\tilde{\bm s}'} &= \prod_{\substack{q\in \tilde{\bm s} \cdot g^{-1} \\ q\sim \tilde q}}  U_{g_q} \tilde{\bm s}_q' U_{g_q}^\dagger 
    \\
    &=
    \prod_{\substack{q\in \tilde{\bm s} \\ q\sim \tilde q}}  U_{g_q g^{-1}} \tilde{\bm s}'_{q\cdot g^{-1}}   U_{g_q g^{-1}}^\dagger
    \\
    &=
    \prod_{\substack{q\in \tilde{\bm s} \\ q\sim \tilde q}}  U_{g_q g^{-1}} U_{g^{-1}}^\dagger \tilde{\bm s}_q U_{g^{-1}}  U_{g_q g^{-1}}^\dagger
    \\
    &=
    \prod_{\substack{q\in \tilde{\bm s} \\ q\sim \tilde q}}  U_{g_q} \tilde{\bm s}_q  U_{g_q}^\dagger = \bm T_{\tilde q}^{\tilde{\bm s}}.\\
\end{align*}
We thus see that a different representative for $\tilde{\bm s}$ has no effect on the definition.

In the rest of the paper, we focus on a special case of the action where $U_g$ is either the Hadamard $\mathfrak{h}$ or identity. We refer to it as an $\ee-\mm$ symmetry, whose quotient is referred to as the $\ee-\mm$ quotient. The same notion is introduced in \cite{Breuckmann2024foldtransversal}, where an $\ee-\mm$ symmetry is referred to as a ZX-duality, and a code with ZX-duality is self-ZX-dual.

We illustrate this with the simple 2D TC, which has $\ee-\mm$ symmetry by exchanging the $\ee$ and $\mm$ anyons. On the lattice, this can be achieved by a $G=\mathbb{Z}$ action, generated by a shift
\begin{align*}
    \begin{tikzpicture}
        [scale = 0.8, baseline=0.7cm]
            \node [circle,fill, inner sep=1.5pt] at (1,0) {};
            \node [circle,fill, inner sep=1.5pt]at (0, 1) {};
            \node [circle,fill, inner sep=1.5pt]at (2, 1) {};
            \node [circle,fill, inner sep=1.5pt]at (1, 2) {};
            \draw[step = 2] (-0.5,-0.5) grid (2.5,2.5);
            \draw[->] (1.1,1.9) -- (1.9,1.1);
            \draw[->] (0.1,0.9) -- (0.9,0.1);
        \end{tikzpicture}
\end{align*}
followed by $\mathfrak{h}$. We perform the quotient by choosing representatives along the $x$-axis, then we obtain (modulo phases)
\begin{align*}
    \text{TC}/\mathbb{Z} = -\sum_i Y_i Y_{i+1}.
\end{align*}

A more interesting example, which we will discuss in detail later is to take the quotient of the $\ee-\mm$ in the 4D toric code. Surprisingly, this gives rise to the 3D fermionic toric code! We will revisit this point in Sec.~\ref{sec:quotient4D}.

\section{$\ee-\mm$ Balanced product}\label{sec:embp}

In this section, we apply the $\ee-\mm$ quotient defined in Sec.~\ref{sec:gqc} to a product construction in an analogous way as the balanced product construction \cite{breuckmann2021balanced}, for which reason we refer to this product as the $\ee-\mm$ balanced product. We will then give a coupled-layer construction to the $\ee-\mm$ balanced product code. Finally, we illustrate the construction using the 3D fermionic TC as an example, as well as constructing a family of non-CSS codes.

Consider a pair of CSS codes CSS$_1$ and CSS$_2$, with a group $G$ acting freely from the right and left respectively, and the onsite unitary part of the action is either trivial or the Hadamard gate $\mathfrak{h}$, defining a right and left $\ee-\mm$ symmetry. The corresponding tensor product code is naturally equipped with an induced $\ee-\mm$ action, the form of which is similar to the usual balanced product
\begin{align*}
    (x\otimes y) \cdot g = x\cdot g \otimes g^{-1} \cdot y,
\end{align*}
where $x\otimes y$ is from one of the spaces $Q_1\otimes Q_2$, $A_1\otimes B_2$ or $B_1\otimes A_2$. The difference in this case is that the action can switch between qubits in $A_1\otimes B_2$ and $B_1\otimes A_2$. We now check that  this action permutes the stabilizers. For example, consider $\bm\alpha(q_1, a_2)$, 
\begin{align*}
    \bm\alpha(q_1, a_2) \cdot g 
    &= U_g^{\dagger} \left(
    \prod_{b_1 \ni q_1} X_{b_1 \cdot g, g^{-1} \cdot a_2} \prod_{q_2\in a_2} X_{q_1\cdot g, g^{-1} \cdot  q_2} \right)
    U_g
    \\
    &=
    \begin{cases}
        \bm \alpha(q_1 \cdot g, g^{-1} \cdot a_2), 
        \hspace{5pt}
        U_g = id,\\
        \bm \beta(q_1 \cdot g, g^{-1} \cdot a_2), 
        \hspace{5pt}
        U_g = \mathfrak{h}.
    \end{cases}
\end{align*}
The other cases are similar. Moreover, due to the boundary map in the product chain complex on $x\otimes y$ always fixes one side, while the group action always acts on both sides simultaneously, the tensor product has the property that qubits checked by a single stabilizer all belong to different orbits. As a result, the quotient of stabilizers is well-defined in this case. We denote this the $\ee-\mm$ balanced product CSS$_1 \otimes_{\ee-\mm}$ CSS$_2$.

\subsection{Coupled-Layer Construction}

The coupled layer construction here generalizes that of the usual balance product construction in our previous work~\cite{zhang2026coupled}.
Choose a set of qubit representatives $\{\tilde q_2\}$ of $G\backslash Q_2$  
and stabilizer representatives $\tilde{ s}$ of $G\backslash A_2\oplus B_2$. The stabilizer representatives are partitioned into two sets, either $X$-types $\{\tilde{ s}_{ a}\} \subseteq A_2$ or $Z$-types $\{\tilde{ s}_{ b}\} \subseteq B_2$. Note that the group action now mixes the $X$-stabilizers and $Z$-stabilizers, given $a_2 \in A_2$, we still denote by $\tilde{ a}_{ 2}$ the representative of $a_2$ which might be a $Z$-stabilizer. Similarly, given $b_2 \in B_2$, its representative $\tilde{ b}_{2}$ might be an $X$-stabilizer. The stabilizer group of the layered system is
\begin{align*}
    \mathcal{S}_0 = 
    \left\langle
    \bm{a_1}^{(\tilde q_2)}, \, \bm{b_1}^{(\tilde q_2)}, \, X_{a_1, \tilde s_b},\, Z_{b_1, \tilde s_a}
    \right\rangle.
\end{align*}
Next, the code switching is performed by imposing
\begin{align*}
    & \bm\alpha(q_1, \tilde s_a) = \prod_{b_1\ni q_1} X_{b_1, \tilde s_a} \prod_{q_2\in \tilde s_a} U^\dagger _{{g_{q_2}}}
    X_{q_1\cdot g_{q_2}, \tilde q_2} U_{{g_{q_2}}}\\
    & \bm \beta(q_1, \tilde s_b) = \prod_{a_1\ni q_1} Z_{a_1, \tilde s_b} \prod_{q_2\in \tilde s_b} U_{{g_{q_2}}}^\dagger Z_{q_1\cdot g_{q_2}, \tilde q_2} U_{{g_{q_2}}}
\end{align*}
After the code switching, the commuting terms are
\begin{align*}
    & \bm\xi(a_1, \tilde q_2) = \bm{a_1}^{(\tilde q_2)} \prod_{b_2\ni \tilde q_2} U_{{g_{b_2}}}^\dagger X_{a_1 \cdot g_{b_2} , \tilde b_2} U_{{g_{b_2}}}, \\
    & \bm\zeta(b_1, \tilde q_2) = \bm{b_1}^{(\tilde q_2)} \prod_{a_2\ni \tilde q_2} U_{{g_{a_2}}}^\dagger Z_{b_1 \cdot g_{a_2} , \tilde a_2} U_{{g_{a_2}}}.
\end{align*}
It is not obvious that the terms in the products are in $\mathcal{S}_0$ in this form. First consider the $\bm \xi(a_1, \tilde q_2)$, given $b_2 \ni \tilde q_2$, if $b_{2}$ is represented by a stabilizer $\tilde{s}_{b}$, that is $\tilde b_2 \in B_2$, then $U_{g_{b_2}} = id$ and $a_1\cdot g_{b_2} \in A_1$,  so we have $X_{a_1\cdot g_{b_2}, \tilde s_b} \in \mathcal{S}_0$; otherwise $ b_{ 2}$ is represented by a $\tilde s_a$ in $A_2$, in which case $U_{g_{b_2}} = \mathfrak{h}$ and $a_1\cdot g_{b_2} \in B_1$, hence we have $Z_{a_1\cdot g_{b_2}, \tilde s_a} \in \mathcal{S}_0$. Therefore, this shows $\bm \xi(a_1, \tilde q_2)$ is in $\mathcal{S}_0$. A similar argument holds for $\bm \zeta(b_1, \tilde q_2)$. This shows the final stabilizer contains
\begin{align*}
    \mathcal{S} \supseteq 
    \left\langle\,
    \bm\alpha(q_1, \tilde s_a),
    \bm\beta(q_1, \tilde s_b),
    \bm\xi(a_1, \tilde q_2),
    \bm\zeta(b_1, \tilde q_2)
    \,\right\rangle.
\end{align*}
One can see that this is CSS$_1\otimes_{\ee-\mm}$ CSS$_2$, with qubit representatives chosen to be $q_1 \otimes \tilde q_2$, $a_1 \otimes \tilde s_b$ and $b_1 \otimes \tilde s_a$, and stabilizer representatives chosen to be $q_1\otimes \tilde s_a$, $q_1\otimes \tilde s_b$, $a_1\otimes \tilde q_2$ and $b_1\otimes \tilde q_2$.

As shown in Sec.\ref{sec:gqc}, the result does not depend on the choice of representatives. If we choose a different representative for qubit $\tilde q_2' \sim \tilde q_2$ related by $\tilde q_2 = g \cdot \tilde q_2'$, we only need to relabel qubits $q_1\otimes \tilde q_2$ by $q_1 \cdot g\otimes \tilde q_2'$, and apply $U_g$ transversally in the layer of CSS$_1$ labeled by $\tilde q_2'$. Similarly for stabilizers, without loss of generality, assume we replace $\tilde{ s}_{ a}$ by another representative $\tilde{s}'$, related by $\tilde{ s}_{ a} = g \cdot \tilde{ s}'$. If $U_g = id$ so that $\tilde{ s}'$ is also an $X$-stabilizer, then we only need to relabel qubits $b_1\otimes \tilde s_a$ by $b_1 \cdot g \otimes \tilde s'$. Otherwise $U_g = \mathfrak{h}$, and $\tilde{ s}'$ is a $Z$-stabilizer, so we need to relabel qubits $b_1\otimes \tilde s_a$ by $b_1 \cdot g \otimes \tilde s'$, followed by a transversal $\mathfrak{h}$ in this layer of ancilla labeled by $\tilde s'$. Moreover, it is not obvious that the coupled-layer construction is symmetric upon exchanging the two CSS codes. However, it is shown from the last section that the result is symmetric, hence we can equivalently couple layers of CSS$_2$ using the stabilizers of CSS$_1$.

\subsection{Example: 3D Fermionic TC}
\label{sec:embp_3dftc}

In this section, we illustrate the above abstraction by taking both CSS$_1$ and CSS$_2$ to be the 2D TC. Their tensor product is the 4D TC, which admits an $\ee-\mm$ duality symmetry. Interestingly, performing the $\ee-\mm$ balancing produces the 3D fermionic TC.

\subsubsection{Coupled-Layer Construction}\label{sec:embp_3dftc_clc}

It turns out that in order to reproduce exactly the stabilizers for the fermionic TC, we will need to express the two input codes in a particular way on the lattice. We express CSS$_1$ as the usual 2D TC on the square lattice with coordinates $(x,y)$, with stabilizer group
\begin{align*}
    \text{TC} = \langle\,
    \bm a_v, \,\bm b_p
    \,|\,
    v\in V,\, p\in P
    \,\rangle
\end{align*}
with qubits on edges, $X$-stabilizers on vertices, and $Z$-stabilizers on plaquettes. On the other hand, we will write CSS$_2$ as a rotated and slanted toric code. More specifically, the qubits are vertices on the square lattice with coordinates $(w,z)$, and the stabilizer group is
\begin{align*}
    \text{CSS}_2 =  \langle\,
    \bm s_{(w,z)} 
    \,|\,
    w,\,z\in \mathbb{Z} 
    \,\rangle
\end{align*}
\begin{equation*}
\bm s_{(w,z)} = \begin{cases}
    \begin{tikzpicture}
    [scale = 0.8, baseline=0.7cm]
        \node at (-1, 0) {$z$};
        \node at (-1.2, 2) {$z+1$};
        \node at (0, 0) {$Z$};
        \node at (2, 0) {$Z$};
        \node at (2, 2) {$Z$};
        \node at (4, 2) {$Z$};

        \node at (2,2.7) {$w$};
        
        \draw[step = 2] (-0.5,-0.5) grid (4.5,2.5);
    \end{tikzpicture} 
    \hspace{15pt}; w \text{ even}\\
    \vspace{0.5pt}
    \hdashrule{0.8\linewidth}{0.5pt}{2pt}
    \vspace{0.5pt}\\
    \begin{tikzpicture}
    [scale = 0.8, baseline=0.7cm]
         \node at (-1, 0) {$z$};
         \node at (-1.2, 2) {$z+1$};
         \node at (0, 0) {$X$};
        \node at (2, 0) {$X$};
        \node at (2, 2) {$X$};
        \node at (4, 2) {$X$};

        \node at (2,2.7) {$w$};
        
        \draw[step = 2] (-0.5,-0.5) grid (4.5,2.5);
    \end{tikzpicture}
    \hspace{15pt}; w \text{ odd}
\end{cases}
\end{equation*}
We take group $\mathbb{Z}$, and define the action of its generator to be consisted of the following:
\vspace{-0.5em}
\begin{itemize}
    \item Clifford gates are chosen to be $U_{\text{even}} = id$, $U_{\text{odd}} = \mathfrak{h}$, where $\mathfrak{h}$ is the Hadamard gate.
    \item Permutation of qubits is chosen to be translation
    \begin{equation*}
        \text{CSS}_1: \,
        \begin{tikzpicture}
        [scale = 0.8, baseline=0.7cm]
            \node [circle,fill, inner sep=1.5pt] at (1,0) {};
            \node [circle,fill, inner sep=1.5pt]at (0, 1) {};
            \node [circle,fill, inner sep=1.5pt]at (2, 1) {};
            \node [circle,fill, inner sep=1.5pt]at (1, 2) {};

            \draw[step = 2] (-0.5,-0.5) grid (2.5,2.5);
    
            \draw[->] (1.1,1.9) -- (1.9,1.1);
            \draw[->] (0.1,0.9) -- (0.9,0.1);
        \end{tikzpicture}
        \hspace{35pt}
        \text{CSS}_2: 
        \begin{tikzpicture}
        [scale = 0.8, baseline=0cm]
            \node [circle,fill, inner sep=1.5pt] at (0,0) {};
            \node [circle,fill, inner sep=1.5pt] at (2,0) {};
            
            \draw[step = 2] (-0.5,-0.5) grid (2.5,0.5);

            \draw[->,  yscale=0.7] (0.1,0.1) arc[start angle=180, end angle=5, radius=0.9cm];

        \end{tikzpicture}
    \end{equation*}
\end{itemize}
\vspace{-0.5em}
which is in fact the $\ee$-$\mm$ duality for both toric codes. The representatives of the quotient CSS$_2/\mathbb{Z}$ are chosen to be:
\vspace{-0.5em}
\begin{itemize}
    \item Qubits in CSS$_2/\mathbb{Z}$ are represented by $(0,z)$. That is, those with coordinates $w=0$
    \item Stabilizers in CSS$_2/\mathbb{Z}$ are represented by $s_{(0,z)}$
\end{itemize}
\vspace{-0.5em}

Next we describe stacking and code switching. For each integer $z$, stack a copy of TC in the $x-y$ plane. Since all the representatives have been chosen to be $Z$ stabilizers in CSS$_2$, we only need to introduce $X$-ancillas, which are defined on pairs $(v,z)$, which we associate to $z$-edges $e_z$. Thus our original stabilizer code lives on edges of the cubic lattice with stabilizers
\begin{align*}
    \mathcal{S}_0 = \langle \bm a_v^{(z)}, \bm b_p^{(z)}, X_{e_z} \rangle.
\end{align*}
The code switching terms are defined for each $x$ or $y$ edge, along with a $z$ edge. Thus they can be associated to $x-z$ and $y-z$ plaquettes of the cubic lattice.
\begin{align*}
    \bm\beta\left(\hspace{5pt}
    \begin{tikzpicture}[scale = 0.8, baseline = -0.4cm]
            \draw (0,0,0) -- (0,0,2);
            \node [fill, circle, black, inner sep = 1pt] at (0,0, 1) {};            
    \end{tikzpicture}
    \hspace{5pt}    
    ,
    \hspace{5pt}
    s_{(0,z)}
    \right)
    \hspace{5pt}
    = 
    \hspace{5pt}
    \begin{tikzpicture}[scale = 0.8, baseline = 0.5cm]
            \coordinate (A) at (0,0,0);
            \coordinate (B) at (0,0,2);
            \coordinate (C) at (0,2,2);
            \coordinate (D) at (0,2,0);
%
            \coordinate (U) at (-2,0,0);
            \coordinate (O) at (2,2,2);
%
%
            %
            \draw (A) -- (B) --(C) -- (D) -- cycle;
            \draw [dotted] (A) -- (U);
            \draw (C) -- (O);
            \foreach \X/\Y in { A/B, B/C, C/D, D/A}
            \node at ($( \X )!0.5!( \Y )$) {$Z$}; 
            \node at (-1.1, 0, 0) {$X$};
            \node at (1.1, 2, 2) {$X$};
    \end{tikzpicture}
\end{align*}
\begin{align*}
    \bm\beta\left(
    \hspace{5pt}
    \begin{tikzpicture}[scale = 0.5, baseline =0 cm]
            \draw (0,0) -- (2,0);
            \node [fill, circle, black, inner sep = 1pt] at (1, 0) {};            
    \end{tikzpicture}
    \hspace{5pt}
    ,
    \hspace{5pt}
s_{(0,z)}    
    \right)
    \hspace{5pt}
    = 
    \hspace{5pt}
    \begin{tikzpicture}[scale = 0.8, baseline = 0.7cm]
%
            \coordinate (A) at (0,0,0);
            \coordinate (B) at (0,2,0);
            \coordinate (C) at (2,2,0);
            \coordinate (D) at (2,0,0);
%
            \coordinate (U) at (0,0,-2);
            \coordinate (O) at (2,2,2);
%
%
  %
            \draw (A) -- (B) --(C) -- (D) -- cycle;
            \draw [dotted] (A) -- (U);
            \draw (C) -- (O);
            \foreach \X/\Y in { A/B, B/C, C/D, D/A}
            \node at ($( \X )!0.5!( \Y )$) {$Z$}; 
            \node at (0, 0, -1.1) {$X$};
            \node at (2, 2, 1.1) {$X$};
    \end{tikzpicture}
\end{align*}
and the remaining commuting terms in $\mathcal{S}_0$ are
\begin{align*}
    \bm\xi\left( a_v, z
    \right)
    \hspace{5pt}
    =
    \hspace{5pt}
    \begin{tikzpicture}[scale = 0.8, baseline = -0.1cm]
    %
        \coordinate (O) at (0,0,0);
        \coordinate (A) at (-1.5,0,0);
        \coordinate (B) at (1.5,0,0);
        \coordinate (C) at (0,-1.5,0);
        \coordinate (D) at (0,1.5,0);
        \coordinate (E) at (0,0,1.5);
        \coordinate (F) at (0,0,-1.5);
        %
        \draw (A) -- (B);
        \draw (C) -- (D); 
        \draw (E) -- (F);
        \foreach \X/\Y in { O/E, O/F}
        \node at ($( \X )!0.5!( \Y )$) {$X$};
        \foreach \X/\Y in { O/C, O/D}
        \node at ($( \X )!0.5!( \Y )$) {$X$};
        \foreach \X/\Y in {O/A, O/B}
        \node at ($( \X )!0.5!( \Y )$) {$X$}; 
    \end{tikzpicture}
\end{align*}
\begin{align*}
    \bm\zeta\left(b_p, z
    \right)
    \hspace{5pt}
    =
    \hspace{5pt}
    \begin{tikzpicture}[scale = 0.8, baseline = -0.4cm]
        \coordinate (A) at (0,0,0);
        \coordinate (B) at (2,0,0);
        \coordinate (C) at (2,0,2);
        \coordinate (D) at (0,0,2);
%
        \coordinate (O) at (2,2,2);
        \coordinate (U) at (0,-2,0);
%
%
        \draw  (A) -- (B) -- (C) -- (D) -- cycle;
        \draw [dotted] (A) -- (U);
        \draw (C) -- (O);
%
 %
        \foreach \X/\Y in { A/B, B/C, C/D, D/A}
        \node[circle, inner sep=1.5pt] at ($( \X )!0.5!( \Y )$) {$Z$};
        \node[circle, inner sep=1.5pt] at (2,1.1,2) {$X$};
        \node[circle, inner sep=1.5pt] at (0,-1.1,0) {$X$};        
    \end{tikzpicture}
\end{align*}
Together, they form the stabilizer group of 3D fermionic TC \cite{2023ScPP...14...65B}. Indeed, the condensation terms can be physically interpreted as the condensation of anyons $\ff_z \ff_{z+1}$ in the stack of 2D toric codes, which indeed produces the 3D fermionic toric code.

\subsubsection{Different Representatives}\label{sec:embp_3dftc_dr}

As shown in Sec.~\ref{sec:gqc}, the resulting code is independent of the representatives. To demonstrate this in this particular example, let us choose another set of representatives of qubits and stabilizers in TC when setting up the coupled-layer construction, and show that we recover the same code, up to a basis transformation.

Let us write the stabilizers of CSS$_2$ as that of the rotated (but not slanted) toric code, and the group action generated by translation to the right by one site, followed by a global Hadamard. The stabilizers are now
\begin{equation*}
\bm s_{(w,z)} = \begin{cases}
    \begin{tikzpicture}[baseline = 0cm]
            \node at (-1.6, -0.8) {$z$};
            \node at (-1.6, 0.8) {$z+1$};
    
            \node at (-0.8,-0.8) {$Z$};
            \node at (-0.8,0.8) {$Z$};
            \node at (0.8,-0.8) {$Z$};
            \node at (0.8,0.8) {$Z$};
            
            \node at (-0.8,1.4) {$w$};
    
            \draw (-0.8,-0.8) -- (0.8,-0.8) -- (0.8,0.8) -- (-0.8,0.8) -- cycle;
        \end{tikzpicture}
        \hspace{15pt}
        ; w+z \text{ even}\\
    \vspace{0.5pt}
    \hdashrule{0.68\linewidth}{0.5pt}{2pt}
    \vspace{0.5pt}\\
        \begin{tikzpicture}[baseline = 0cm]
            \node at (-1.6, -0.8) {$z$};
            \node at (-1.6, 0.8) {$z+1$};
    
            \node at (-0.8,-0.8) {$X$};
            \node at (-0.8,0.8) {$X$};
            \node at (0.8,-0.8) {$X$};
            \node at (0.8,0.8) {$X$};
            
            \node at (-0.8,1.4) {$w$};
    
            \draw (-0.8,-0.8) -- (0.8,-0.8) -- (0.8,0.8) -- (-0.8,0.8) -- cycle;
        \end{tikzpicture} 
        \hspace{15pt}
        ; w+z \text{ odd}
\end{cases}
\end{equation*}
where we label each plaquette by the vertex on the bottom left corner of that plaquette. We say a vertex, or qubit, is of $Z$-type or $X$-type if its corresponding plaquette is a $Z$-stabilizer or $X$-stabilizer, respectively.

Now, we choose the representatives for the quotient.
\vspace{-0.5em}
\begin{itemize}
    \item Qubits: representatives are along the $w=0$-axis, labeled by integers $z$. With this, even integers are $Z$-type and odd integers are $X$-type.
    \vspace{-0.5em}
    \item Stabilizers: representives are $s_{(0,z)}$, which we will label by $z$ edges. Even integers are $Z$-stabilizers and odd integers are $X$-stabilizers.
\end{itemize}
\vspace{-0.5em}
Note that the qubit and stabilizer representatives chosen in Sec.~\ref{sec:embp_3dftc_clc} correspond respectively to the vertices and plaquettes on the diagonal line $z=-w+1$, whereas in this section the representatives live along the $w=0$ axis.

Next, we describe the stacking and code switching. For convenience, we define $z = 2\ell + 1$ for $z$ odd and $z=2\ell$ for $z$ even. Again, for each integer $z$, we introduce a copy of the TC in the $x-y$ plane. For each vertex $v$ and interval $[2l,2l+1]$, we introduce an $X$-ancilla, which lives a $z$-edge of the cubic lattice. For each plaquette $p$ and interval $[2l+1,2l+2]$, we introduce a $Z$-ancilla, which lives a cube of the cubic lattice. We can shift the qubits in cubes in $(1,1,0)$-direction (in 3D coordinates with $x$-axis points out-of-page, $y$-axis points to the right, and $z$-axis points up), so that they also live on a $z$-edge. In this way, all qubits in the stacked system live on edges of the 3D cubic lattice. Together, the TC layers and the ancillas form the stabilizer group $\mathcal{S}_0$.

The code switching terms are
\begin{align*}
    \bm \alpha\left(
    \hspace{5pt}
    \begin{tikzpicture}[scale = 0.5, baseline =-0.1 cm]
            \draw (0,0) -- (2,0);
            \node[fill, circle, black, inner sep = 1pt] at (1, 0) {};     
    \end{tikzpicture}
    \hspace{5pt}
    ,
    \hspace{5pt}
    s_{(0,2\ell+1)}
    \right) 
    \hspace{5pt}
    &= 
    \hspace{5pt}
    \begin{tikzpicture}[scale = 0.8, baseline = 0.5cm]
            \coordinate (A) at (0,0,0);
            \coordinate (B) at (0,0,2);
            \coordinate (C) at (0,2,2);
            \coordinate (D) at (0,2,0);
%
            \coordinate (U) at (-2,0,0);
            \coordinate (O) at (-2,2,0);
%
%
 %
            \draw (A) -- (B) --(C) -- (D) -- cycle;
            \draw [dotted] (A) --  (-0.8,0,0);\draw (-0.78,0,0) -- (U);
            \draw (D) -- (O);
            \foreach \X/\Y in {  B/C,  D/A}
            \node at ($( \X )!0.5!( \Y )$) {$X$};
            \foreach \X/\Y in { A/B,  C/D}
            \node at ($( \X )!0.5!( \Y )$) {$Z$};
            \node at (-1.1, 0, 0) {$X$};
            \node at (-1.1, 2, 0) {$X$};
            \node at (1.3, 0 ,0) {$2\ell+1$};
            \node at (1.3, 2 ,0) {$2\ell+2$};
    \end{tikzpicture}
\end{align*}
\begin{align*}
    \bm \alpha\left(
    \hspace{5pt}
    \begin{tikzpicture}[scale = 0.75, baseline = -0.4cm]
            \draw (0,0,0) -- (0,0,2);
            \node [fill, circle, black, inner sep = 1pt] at (0,0, 1) {};    
    \end{tikzpicture}
    \hspace{5pt},
    \hspace{5pt}
    s_{(0,2\ell+1)}
    \right)
    \hspace{5pt}
    &= 
    \hspace{5pt}
    \begin{tikzpicture}[scale = 0.8, baseline = 0.7cm]
            \coordinate (A) at (0,0,0);
            \coordinate (B) at (0,2,0);
            \coordinate (C) at (2,2,0);
            \coordinate (D) at (2,0,0);
%
            \coordinate (U) at (0,0,-2);
            \coordinate (O) at (0,2,-2);
%
%
 %
            \draw (A) -- (B) --(C) -- (D) -- cycle;
            \draw [dotted](A) -- (U);
            \draw (B) -- (O);
            \foreach \X/\Y in { B/C, D/A}
            \node at ($( \X )!0.5!( \Y )$) {$Z$}; 
            \foreach \X/\Y in { A/B, C/D}
            \node at ($( \X )!0.5!( \Y )$) {$X$}; 
            \node at (0, 0, -1.1) {$X$};
            \node at (0, 2, -1.1) {$X$};
            \node at (3.3, 0 ,0) {$2\ell+1$};
            \node at (3.3, 2 ,0) {$2\ell+2$};
    \end{tikzpicture}
\end{align*}
\begin{align*}
    \bm \beta\left(
    \hspace{5pt}
    \begin{tikzpicture}[scale = 0.5, baseline = -0.1 cm]
            \draw (0,0) -- (2,0);
            \node[fill, circle, black, inner sep = 1pt] at (1, 0) {};            
    \end{tikzpicture}
    \hspace{5pt}
    ,
    \hspace{5pt}
    s_{(0,2\ell)}
    \right)
    \hspace{5pt}
    &=
    \hspace{5pt}
    \begin{tikzpicture}[scale = 0.8, baseline = 0.7cm]     
            \coordinate (A) at (0,0,0);
            \coordinate (B) at (0,2,0);
            \coordinate (C) at (2,2,0);
            \coordinate (D) at (2,0,0);
%
            \coordinate (U) at (2,0,2);
            \coordinate (O) at (2,2,2);
%
%
 %
            \draw (A) -- (B) --(C) -- (D) -- cycle;
            \draw (D) -- (U);
            \draw (C) -- (O);
            \foreach \X/\Y in { A/B, B/C, C/D, D/A}
            \node at ($( \X )!0.5!( \Y )$) {$Z$}; 
            \node at (2, 0, 1.1) {$X$};
            \node at (2, 2, 1.1) {$X$};
            \node at (3, 0 ,0) {$2\ell$};
            \node at (3, 2 ,0) {$2\ell+1$};
    \end{tikzpicture}
\end{align*}
\begin{align*}
    \bm \beta\left(
    \hspace{5pt}
    \begin{tikzpicture}[scale = 0.75, baseline = -0.4cm]
            \draw (0,0,0) -- (0,0,2);
            \node[fill, circle, black, inner sep = 1pt] at (0,0, 1) {};     
    \end{tikzpicture}
    \hspace{5pt},
    \hspace{5pt}s_{(0,2\ell)}
    \right) 
    \hspace{5pt}
    &= 
    \hspace{5pt}
    \begin{tikzpicture}[scale = 0.8, baseline = 0.1cm]
            \coordinate (A) at (0,0,0);
            \coordinate (B) at (0,0,2);
            \coordinate (C) at (0,2,2);
            \coordinate (D) at (0,2,0);
%
            \coordinate (U) at (2,0,2);
            \coordinate (O) at (2,2,2);
%
%
 %
            \draw (B) -- (C) -- (D);
            \draw [dotted](A) -- (B);
            \draw [dotted] (A) -- (0.78,2,2);
            \draw (0.78,2,2)--(D);
            \draw (B) -- (U);
            \draw (C) -- (O);
            \foreach \X/\Y in { A/B, B/C, C/D, D/A}
            \node at ($( \X )!0.5!( \Y )$) {$Z$}; 
            \node at (1.1, 0, 2) {$X$};
            \node at (1.1, 2, 2) {$X$};
            \node at (3, 0 ,2) {$2\ell$};
            \node at (3, 2 ,2) {$2\ell+1$};
    \end{tikzpicture}
\end{align*}
The remaining commuting terms in $\mathcal{S}_0$ are
\begin{align*}
    \bm\xi\left( a_v, 2\ell
    \right)
    \hspace{5pt}
    = 
    \hspace{5pt}
    \begin{tikzpicture}[scale = 0.8, baseline = 0cm]
        \coordinate (O) at (0,0,0);
        \coordinate (A) at (-1.5,0,0);
        \coordinate (B) at (1.5,0,0);
        \coordinate (C) at (0,-1.5,0);
        \coordinate (D) at (0,1.5,0);
        \coordinate (E) at (0,0,1.5);
        \coordinate (F) at (0,0,-1.5);
 %
        \draw (A) -- (B);
        \draw (C) -- (D); 
        \draw (E) -- (F);
        \foreach \X/\Y in { O/E, O/F}
        \node at ($( \X )!0.5!( \Y )$) {$X$};
        \node at ($(O)!0.5!(C)$) {$Z$};
        \node at ($(O)!0.5!(D)$) {$X$};
        \foreach \X/\Y in {O/A, O/B}
        \node at ($( \X )!0.5!( \Y )$) {$X$}; 
        \node at (2,0,0) {$2\ell$};
    \end{tikzpicture}
\end{align*}
\begin{align*}
    \bm\xi\left( a_v, 2\ell+1
    \right)
    \hspace{5pt}
    = 
    \hspace{5pt}
    \begin{tikzpicture}[scale = 0.8, baseline = 0cm]
        \coordinate (O) at (0,0,0);
        \coordinate (A) at (-1.5,0,0);
        \coordinate (B) at (1.5,0,0);
        \coordinate (C) at (0,-1.5,0);
        \coordinate (D) at (0,1.5,0);
        \coordinate (E) at (0,0,1.5);
        \coordinate (F) at (0,0,-1.5);
 %
        \draw (A) -- (B);
        \draw (C) -- (D); 
        \draw (E) -- (F);
        \foreach \X/\Y in { O/E, O/F}
        \node at ($( \X )!0.5!( \Y )$) {$X$};
        \node at ($(O)!0.5!(C)$) {$X$};
        \node at ($(O)!0.5!(D)$) {$Z$};
        \foreach \X/\Y in {O/A, O/B}
        \node at ($( \X )!0.5!( \Y )$) {$X$};
        \node at (2.4,0,0) {$2\ell+1$};
    \end{tikzpicture}
\end{align*}
\vspace{-0.3cm}
\begin{align*}
    \bm\zeta\left(b_p, 2\ell
    \right)
    \hspace{5pt}
    = 
    \hspace{5pt}
    \begin{tikzpicture}[scale = 0.8, baseline = -0.5cm]
        \coordinate (A) at (0,0,0);
        \coordinate (B) at (2,0,0);
        \coordinate (C) at (2,0,2);
        \coordinate (D) at (0,0,2);
%
        \coordinate (O) at (0,1.5,0);
        \coordinate (U) at (2,-1.5,2);
%
%
        \draw  (A) -- (B) -- (C) -- (D) -- cycle;
        \draw (C) -- (U);
        \draw (A) -- (O);
%
 %
        \foreach \X/\Y in { A/B, B/C, C/D, D/A}
        \node[circle, inner sep=1.5pt] at ($( \X )!0.5!( \Y )$) {$Z$};
        \node at (0,0.8,0) {$X$};
        \node at (2,-0.8,2) {$Z$};
        \node at (3,0,1) {$2\ell$};
    \end{tikzpicture}
\end{align*}
\begin{align*}
    \bm\zeta\left(b_p, 2\ell+1
    \right)
    \hspace{5pt}
    = 
    \hspace{5pt}
    \begin{tikzpicture}[scale = 0.8, baseline = -0.5cm]   
        \coordinate (A) at (0,0,0);
        \coordinate (B) at (2,0,0);
        \coordinate (C) at (2,0,2);
        \coordinate (D) at (0,0,2);
%
        \coordinate (O) at (2,2,2);
        \coordinate (U) at (0,-2,0);
%
%
        \draw  (A) -- (B) -- (C) -- (D) -- cycle;
        \draw [dotted] (A) -- (U);
        \draw (C) -- (O);
%
 %
        \foreach \X/\Y in { A/B, B/C, C/D, D/A}
        \node[circle, inner sep=1.5pt] at ($( \X )!0.5!( \Y )$) {$Z$};
        \node at (2,1.1,2) {$Z$};
        \node at (0,-1.1,0) {$X$};
        \node at (3.2,0,1) {$2\ell+1$};
    \end{tikzpicture}
\end{align*}

In order to recover the usual form of the stabilizers in fermionic TC according to the procedure in previous section, we need to act by $\mathbb{Z}$ to shift the representatives on the $z=-w+1$ line to the $w=0$ axis. In the horizontal layer $z$, we shift the qubits $(e,z)$ $z-1$ times to $(e \cdot t^{z-1}, z)$, followed by $\mathfrak{h}^{z-1}$, where $t$ is the generator of $\mathbb{Z}$. On the vertical edges $(v, [2l, 2l+1])$, we shift $2l$ times to $(v\cdot t^{2l}, [2l, 2l+1])$. Note $v \cdot t$ is a plaquette, so an even number of actions take vertices to vertices. Similarly, on vertical edges $(v, [2l + 1, 2l+2])$, we need to shift $2l + 1$ times according to the general procedure in Sec.~\ref{sec:gqc}. However, recall we have already shifted once when we move the qubits in the cubes $(p, [2l + 1, 2l+2])$ to edges, hence we only need to shift $2l$ more times to $(v\cdot t^{2l}, [2l +1, 2l+2])$, followed by $\mathfrak{h}$. It is easy to check

\begin{align*}
    \left.
    \begin{array}{ll}
         & \bm \alpha\left(
        \hspace{5pt}
        \begin{tikzpicture}[scale = 0.5, baseline =-0.1 cm]
            \draw (0,0) -- (2,0);
            \node[fill, circle, black, inner sep = 1pt] at (1, 0) {};     
        \end{tikzpicture}
        \hspace{5pt}
        ,
        \hspace{5pt}
        s_{(0,2\ell+1)}
        \right) 
        \\[0.5cm]
        & \bm \beta\left(
        \hspace{5pt}
        \begin{tikzpicture}[scale = 0.75, baseline = -0.4cm]
                \draw (0,0,0) -- (0,0,2);
                \node[fill, circle, black, inner sep = 1pt] at (0,0, 1) {};     
        \end{tikzpicture}
        \hspace{5pt},
        \hspace{5pt}s_{(0,2\ell)}
        \right) 
    \end{array}
    \right\}
    \hspace{5pt}
    \mapsto
    \begin{tikzpicture}[scale = 0.8, baseline = 0.5cm]
            %
            \coordinate (A) at (0,0,0);
            \coordinate (B) at (0,0,2);
            \coordinate (C) at (0,2,2);
            \coordinate (D) at (0,2,0);
%
            \coordinate (U) at (-2,0,0);
            \coordinate (O) at (2,2,2);
%
%
            %
            \draw (A) -- (B) --(C) -- (D) -- cycle;
            \draw [dotted] (A) -- (U);
            \draw (C) -- (O);
            \foreach \X/\Y in { A/B, B/C, C/D, D/A}
            \node at ($( \X )!0.5!( \Y )$) {$Z$}; 
            \node at (-1.1, 0, 0) {$X$};
            \node at (1.1, 2, 2) {$X$};
    \end{tikzpicture}
\end{align*}

\begin{align*}
    \left. 
    \begin{array}{ll}
         &  \bm \alpha\left(
        \hspace{5pt}
        \begin{tikzpicture}[scale = 0.75, baseline = -0.4cm]
                \draw (0,0,0) -- (0,0,2);
                \node [fill, circle, black, inner sep = 1pt] at (0,0, 1) {};    
        \end{tikzpicture}
        \hspace{5pt},
        \hspace{5pt}
        s_{(0,2\ell+1)}
        \right) 
        \\[0.5cm]
        & \bm \beta\left(
        \hspace{5pt}
        \begin{tikzpicture}[scale = 0.5, baseline = -0.1 cm]
                \draw (0,0) -- (2,0);
                \node[fill, circle, black, inner sep = 1pt] at (1, 0) {};            
        \end{tikzpicture}
        \hspace{5pt}
        ,
        \hspace{5pt}
        s_{(0,2\ell)}
        \right)
    \end{array}
    \right\}
    \hspace{5pt}
    \mapsto
    \begin{tikzpicture}[scale = 0.8, baseline = 0.8cm]
%
            \coordinate (A) at (0,0,0);
            \coordinate (B) at (0,2,0);
            \coordinate (C) at (2,2,0);
            \coordinate (D) at (2,0,0);
%
            \coordinate (U) at (0,0,-2);
            \coordinate (O) at (2,2,2);
%
%
  %
            \draw (A) -- (B) --(C) -- (D) -- cycle;
            \draw [dotted] (A) -- (U);
            \draw (C) -- (O);
            \foreach \X/\Y in { A/B, B/C, C/D, D/A}
            \node at ($( \X )!0.5!( \Y )$) {$Z$}; 
            \node at (0, 0, -1.1) {$X$};
            \node at (2, 2, 1.1) {$X$};
    \end{tikzpicture}
\end{align*}

\begin{align*}
    \left. 
    \begin{array}{ll}
         &  \bm\xi(a_v, 2l)
        \\[0.5cm]
        & \bm\zeta(b_p, 2l+1)
    \end{array}
    \right\}
    \hspace{5pt}
    \mapsto
    \hspace{5pt}
    \begin{tikzpicture}[scale = 0.8, baseline = -0.4cm]
    %
        \coordinate (A) at (0,0,0);
        \coordinate (B) at (2,0,0);
        \coordinate (C) at (2,0,2);
        \coordinate (D) at (0,0,2);
%
        \coordinate (O) at (2,2,2);
        \coordinate (U) at (0,-2,0);
%
%
        \draw  (A) -- (B) -- (C) -- (D) -- cycle;
        \draw [dotted] (A) -- (U);
        \draw (C) -- (O);
%
 %
        \foreach \X/\Y in { A/B, B/C, C/D, D/A}
        \node[circle, inner sep=1.5pt] at ($( \X )!0.5!( \Y )$) {$Z$};
        \node[circle, inner sep=1.5pt] at (2,1.1,2) {$X$};
        \node[circle, inner sep=1.5pt] at (0,-1.1,0) {$X$};        
    \end{tikzpicture}
\end{align*}

\begin{align*}
    \left. 
    \begin{array}{ll}
         &  \bm\xi(a_v, 2l+1)
        \\[0.5cm]
        & \bm \zeta(b_p, 2l)
    \end{array}
    \right\}
    \hspace{5pt}
    \mapsto
    \hspace{5pt}
    \begin{tikzpicture}[scale = 0.8, baseline = -0.1cm]
    %
        \coordinate (O) at (0,0,0);
        \coordinate (A) at (-1.5,0,0);
        \coordinate (B) at (1.5,0,0);
        \coordinate (C) at (0,-1.5,0);
        \coordinate (D) at (0,1.5,0);
        \coordinate (E) at (0,0,1.5);
        \coordinate (F) at (0,0,-1.5);
        %
        \draw (A) -- (B);
        \draw (C) -- (D); 
        \draw (E) -- (F);
        \foreach \X/\Y in { O/E, O/F}
        \node at ($( \X )!0.5!( \Y )$) {$X$};
        \foreach \X/\Y in { O/C, O/D}
        \node at ($( \X )!0.5!( \Y )$) {$X$};
        \foreach \X/\Y in {O/A, O/B}
        \node at ($( \X )!0.5!( \Y )$) {$X$}; 
    \end{tikzpicture}
\end{align*}

Indeed, we recover the stabilizers of fermionic TC.

\subsubsection{Quotient 4D TC}
\label{sec:quotient4D}

The usual tensor product of two 2D TC is the 4D loop-only TC, which has an $\ee-\mm$ duality inherited from the $\ee-\mm$ duality in 2D TC. To see this, recall the qubits of 4D TC live on $2$-cells, which are labeled by $(e,e')$, $(v,p')$ and $(p, v')$, where $v$, $e$, $p$ are the vertices, edges and plaquettes in 2D, the un-primed and primed labels correspond to the first and second copy respectively. The $X$-stabilizers live on $1$-cells, labeled by $\bm a(v,e')$ and $\bm a(e, v')$. The $Z$-stabilizers live on $3$-cells, labeled by $\bm b(p,e')$ and $\bm b(e, p')$. The generator of the $\mathbb{Z}$ permutes the qubits by shifts $(x,y') \mapsto (x \cdot t, \, t^{-1} \cdot y')$, where $(x,y')$ is one of $(e,e')$, $(v,p')$ or $(p, v')$, followed by a global $\mathfrak{h}$. To see the action on stabilizers, consider $\bm a(v,e')$ as an example
\begin{align*}
    \bm a(v,e') \cdot t = \prod_{e\in v} Z_{e\cdot t, \,t^{-1} \cdot e'} \prod_{p'\ni e'} Z_{v\cdot t, \,t^{-1} \cdot p'} = \bm b(v\cdot t, \, t^{-1}\cdot e').
\end{align*}
Similarly for the other stabilizers. Thus indeed the action switches the $X$-stabilizers and $Z$-stabilizers in 4D TC. Physically this exchanges the $\ee$ and $\mm$ loop excitations.

If we further quotient the 4D TC  by this $\ee-\mm$ duality, we obtain the 3D fermionic TC. Even though the explicit form of stabilizers depends on the representatives, the phase does not. If we choose to move all group actions onto the second copy of TC, then this corresponds to the coupled-layer construction.


\subsection{3-fermion Walker-Wang and a family of codes}

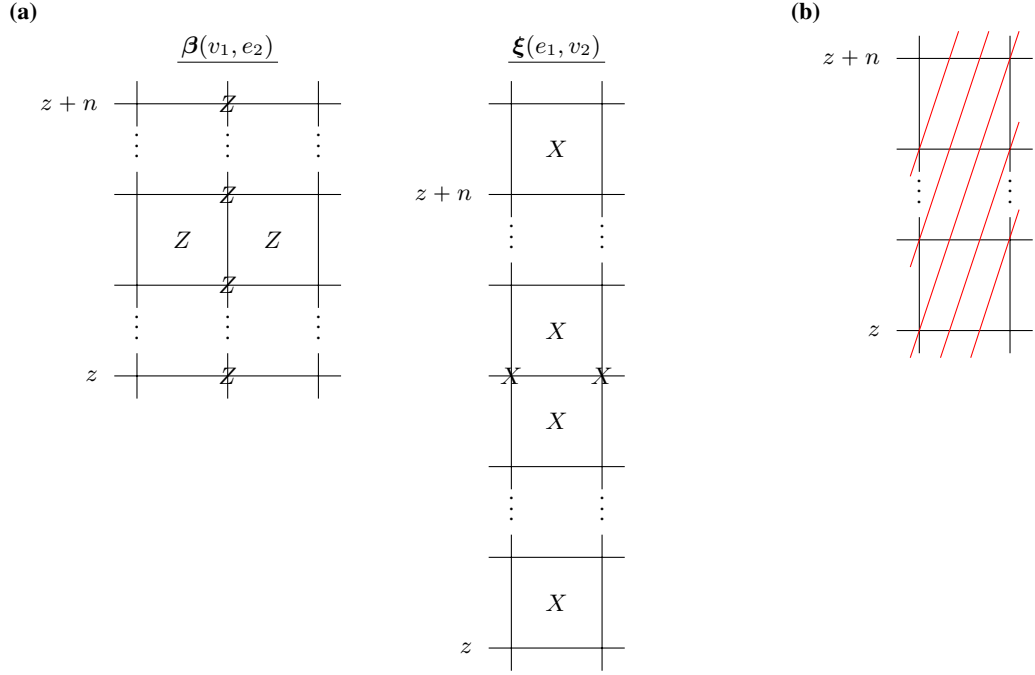
\begin{figure*}[t]
    \begin{tikzpicture}
    [scale = 0.6, baseline=-3.5cm]
        \node at (-2.5, 8) {\textbf{(a)}};

        \node at (2,7.2) {$\underline{\,\bm \beta(v_1,e_2)\,}$};

        \node at (2, 0) {$Z$};
        \node at (3, 3) {$Z$};
        
        \node at (1, 3) {$Z$};
        \node at (2, 6) {$Z$};
        \node at (2, 4) {$Z$};
        \node at (2, 2) {$Z$};

        \node at (0 , 5.25) {$\vdots$};
        \node at (2 , 5.25) {$\vdots$};
        \node at (4 , 5.25) {$\vdots$};

        \node at (0 , 1.25) {$\vdots$};
        \node at (2 , 1.25) {$\vdots$};
        \node at (4 , 1.25) {$\vdots$};

        \draw[step = 2] (-0.5,-0.5) grid (4.5,0.5);
        \draw[step = 2] (-0.5,1.5) grid (4.5,4.5);
        \draw[step = 2] (-0.5,5.5) grid (4.5,6.5);

        \node at (-1, 0) {$z$};
        \node at (-1.5, 6) {$z+n$};
    \end{tikzpicture}
    \hspace{20pt}
    \begin{tikzpicture}
    [scale = 0.6, baseline=0.1cm]

        \node at (1,13.2) {$\underline{\,\bm \xi(e_1,v_2)\,}$};
        
        \node at (1, 1) {$X$};
        \node at (1, 5) {$X$};
        \node at (1, 7) {$X$};
        \node at (1, 11) {$X$};
        \node at (0, 6) {$X$};
        \node at (2, 6) {$X$};

        \node at (0 , 9.25) {$\vdots$};
        \node at (2 , 9.25) {$\vdots$};

        \node at (0 , 3.25) {$\vdots$};
        \node at (2 , 3.25) {$\vdots$};

        \draw[step = 2] (-0.5,-0.5) grid (2.5,2.5);
        \draw[step = 2] (-0.5,3.5) grid (2.5,8.5);
        \draw[step = 2] (-0.5,9.5) grid (2.5,12.5);

        \node at (-1, 0) {$z$};
        \node at (-1.5, 10) {$z+n$};
    \end{tikzpicture}
    \hspace{55pt}
    \begin{tikzpicture}
    [scale = 0.6, baseline = -4.1cm]
        \node at (-2.5, 7) {\textbf{(b)}};

        \node at (0 , 3.25) {$\vdots$};
        \node at (2 , 3.25) {$\vdots$};

        \draw[step = 2] (-0.5,-0.5) grid (2.5,2.5);
        \draw[step = 2] (-0.5,3.5) grid (2.5,6.5);

        \draw[red] (-0.2,-0.6) -- (2.2,6.6);
        \draw[red] (0.46667,-0.6) -- (2.2,4.6);
        \draw[red] (-0.2,1.4) -- (1.53333,6.6);

        \draw[red] (1.13333, -0.6) -- (2.2, 2.6666);
        \draw[red] (-0.2,3.4) -- (0.866666,6.6);
       
        \node at (-1, 0) {$z$};
        \node at (-1.5, 6) {$z+n$};
    \end{tikzpicture}
    \caption{(a) Tensor product of Ising$(2;X)$ with Ising$(n+1,Z)$ for $n+1$ even. (b) Lattice deformation.}
    \label{fig:ising_hyper}
\end{figure*}

In this section, we generalize the construction of 3D fermionic TC in Sec.~\ref{sec:embp_3dftc}, and construct a family of codes indexed by positive integers $n \in \mathbb{Z}_+$, using the $\ee-\mm$ balanced product construction. The first member of this family is the 3D fermionic TC, the second is the $3$-fermion Walker-Wang model \cite{walker20123+,Burnell14,2023CMaPh.398..469H}, and so on. We will show, from the coupled-layer perspective, these models generalize the coupled-layer construction in \cite{wang2013boson}.

First we recall, the 2D TC is the tensor product of two repetition codes. 
\begin{align*}
    \text{TC} = \text{Ising}(2;X) \otimes \text{Ising}(2;Z),
\end{align*}
where Ising$(2;X)$ and Ising$(2;Z)$ denotes the 1D Ising chain in $X$-basis and $Z$-basis respectively. That is,
\begin{align*}
    &\text{Ising}(2; X) = \langle X_iX_{i+1} \rangle,\\
    &\text{Ising}(2; Z) = \langle Z_iZ_{i+1}\rangle .
\end{align*}
A straightforward generalization is to simply replace one of the factors by a 1D multi-spin Ising code
\begin{align*}
    &\text{Ising}(n+1; X) =\bigg\langle\, \prod_{j=0}^{n} X_{i+j}\,\bigg\rangle,\\
    &\text{Ising}(n+1; Z) = \bigg\langle\, \prod_{j=0}^{n} Z_{i+j}\,\bigg\rangle .
\end{align*}
For each $n\in \mathbb{Z}_+$, we define the following generalized toric code
\begin{align*}
    \text{TC}^{(n)} := \text{Ising}(2;X) \otimes \text{Ising}(n+1;Z),
\end{align*}
where for $n=1$ we recover TC$^{(1)} = $ TC.

It turns out that all TC$^{(n)}$ possess an $\ee-\mm$ symmetry. This is because there is an isomorphism between bits in checks in the $n$-spin Ising code $\text{Ising}(n+1; Z)$~\cite{Gorantla25}.  We can also see this explicitly by inspecting the stabilizers. As an example, take $n+1$ to be even. In this case, we can take the chain complexes
\begin{align*}
& \begin{tikzpicture}[scale=0.8, every node/.style={scale=1}]
    \node at (0,-0.05) {$\text{Ising}(2; X):$};
    \node [text=red!80!black] at (1.5, 0) {$E_1$};
    \node at (3.5, 0) {$V_1$};
    \draw[->] (1.8, 0) -- (3.2,0);
\end{tikzpicture}\\
& \begin{tikzpicture}[scale=0.8, every node/.style={scale=1}]
    \node at (0,-0.05) {$\text{Ising}(n+1; Z):$};
    \node at (2.0, 0) {$V_2$};
    \node [text=blue!80!black] at (4, 0) {$E_2$};
    \draw[->] (2.3, 0) -- (3.7,0);
\end{tikzpicture}
\end{align*}
Thus the qubits in CSS$^{(n)}$ can be placed on vertices and plaquettes. The case where $n+1$ is odd is identical except qubits are placed on vertices and horizontal edges. The $X$-stabilizers $\bm \xi(e_1, v_2)$ and $Z$-stabilizers $\bm \beta(v_1, e_2)$ are shown in Figure~\ref{fig:ising_hyper}(a). Clearly, there is an $\ee-\mm$ symmetry generated by the translation
\begin{equation*}
    \begin{tikzpicture}
        [scale = 0.8, baseline=0.7cm]
            \node [circle,fill, inner sep=1.5pt] at (0,0) {};
            \node [circle,fill, inner sep=1.5pt]at (1, 1) {};
            \node [circle,fill, inner sep=1.5pt]at (2, 2) {};

            \draw[step = 2] (-0.5,-0.5) grid (2.5,2.5);
    
            \draw[->] (0.1,0.1) -- (0.9,0.9);
            \draw[->] (1.1,1.1) -- (1.9,1.9);
        \end{tikzpicture}
\end{equation*}
followed by a Hadamard. This $\ee-\mm$ symmetry reduces to the usual $\ee-\mm$ duality in TC for $n=1$. We can now perform the $\ee-\mm$ balanced product using these codes as inputs. To be concrete, for simplicity we will discuss the case where CSS$_1$ is chosen to be the 2D TC, with the translational part of the $\ee-\mm$ symmetry the same as before, given by
\begin{equation*}
    \begin{tikzpicture}
        [scale = 0.8, baseline=0.7cm]
            \node [circle,fill, inner sep=1.5pt] at (1,0) {};
            \node [circle,fill, inner sep=1.5pt]at (0, 1) {};
            \node [circle,fill, inner sep=1.5pt]at (2, 1) {};
            \node [circle,fill, inner sep=1.5pt]at (1, 2) {};

            \draw[step = 2] (-0.5,-0.5) grid (2.5,2.5);
    
            \draw[->] (1.1,1.9) -- (1.9,1.1);
            \draw[->] (0.1,0.9) -- (0.9,0.1);
        \end{tikzpicture}
\end{equation*}
and we choose CSS$_2$ to be TC$^{(n)}$ with the $\ee-\mm$ symmetry discussed above. Therefore, the coupled-layer construction of TC $\otimes_{\ee-\mm}$ TC$^{(n)}$ is given by stacking layers of TC and condensing bound states of anyons across several layers, and we expect some topological phase as an outcome.

As seen in Sec.~\ref{sec:embp_3dftc_dr}, even though the outcome does not depend on the representatives of the quotient, the specific form of the stabilizers might not look intuitive if the representatives are not chosen carefully. In order to make the stabilizers similar to that of the fermionic TC, we first perform a couple of lattice deformations, similar to what we did for $n=1$. First we choose a new horizontal axis extending in the $(1,n)$ direction, given by the red lines in Figure~\ref{fig:ising_hyper}(b). On the new lattice, the qubits live on vertices and horizontal edges, which correspond to vertices and plaquettes of the original lattice. By further doubling the grid points along the horizontal axis, we can put all qubits on vertices, and the stabilizers on the new lattice become
\begin{equation*}
    \begin{tikzpicture}
    [scale = 0.6, baseline=2cm]
        \node at (2, 9.5) {$\underline{\,\bm \beta(v_1, e_2) \,}$};
        
        \node at (2, 0) {$Z$};
        \node at (4, 0) {$Z$};
        
        \node at (0, 8) {$Z$};
        \node at (2, 8) {$Z$};
        \node at (2, 6) {$Z$};
        \node at (2, 4) {$Z$};
        \node at (2, 2) {$Z$};

        \node at (0 , 5.25) {$\vdots$};
        \node at (2 , 5.25) {$\vdots$};
        \node at (4 , 5.25) {$\vdots$};
        
        \node at (2,-0.8) {$2l$};        
        \node at (-1, 0) {$z$};
        \node at (-1.5, 8) {$z+n$};
        
        \draw[step = 2] (-0.5,-0.5) grid (4.5,4.5);
        \draw[step = 2] (-0.5,5.5) grid (4.5,8.5);
    \end{tikzpicture}
    \hspace{30pt}
    \begin{tikzpicture}
    [scale = 0.6, baseline=2cm]
        \node at (2, 9.5) {$\underline{\,\bm \xi(e_1, v_2) \,}$};
        
        \node at (2, 0) {$X$};
        \node at (4, 0) {$X$};
        
        \node at (0, 8) {$X$};
        \node at (2, 8) {$X$};
        \node at (2, 6) {$X$};
        \node at (2, 4) {$X$};
        \node at (2, 2) {$X$};

        \node at (0 , 5.25) {$\vdots$};
        \node at (2 , 5.25) {$\vdots$};
        \node at (4 , 5.25) {$\vdots$};
        
        \node at (2,-0.8) {$2l+1$};        
        
        \draw[step = 2] (-0.5,-0.5) grid (4.5,4.5);
        \draw[step = 2] (-0.5,5.5) grid (4.5,8.5);
    \end{tikzpicture}
\end{equation*}
Note that $\bm \beta(v_1, e_2)$ center on even coordinates $2l$, while $\bm \xi(e_1, v_2)$ center on odd coordinates $2l+1$. Keeping track of the $\ee-\mm$ symmetry, we see the symmetry now is generated by a simple translation to the right, followed by a Hadamard.

Next we perform a second lattice deformation. For each $k\in \mathbb{Z}$, in the $n$ horizontal layers labeled by $z = kn+1, \,\cdots ,\, (k+1)n$, we shift the qubits to the right by $2k+1$ sites. Note for $k<0$, shifting to the right by a negative number of sites means shifting to the left. By doing this, we see the $Z$-stabilizers become
\begin{align*}
    & \bm \beta_{2l, kn-a} = Z_{2l,kn-a} Z_{2l, (k+1)n-a}\\
    &\hspace{65pt} \prod_{i=0}^a Z_{2l-1,kn-a+i}\prod_{i=a+1}^{n}Z_{2l+1, kn-a+i},
\end{align*} 
where we have relabeled the stabilizers $\bm \beta(v_1, e_2)$ by $\bm \beta_{2l, kn-a}$ for $l$, $k\in \mathbb{Z}$ and $0\leq a \leq n-1$.
Note that $\bm \beta_{2l, kn-a}$ written in this form is $n$-site translational invariant along the vertical axis, and two-site translational invariant along the horizontal axis. Pictorially, these stabilizers are shown in Figure~\ref{fig:FnZStab} below. Similarly, the $X$-stabilizers become  
\begin{align*}
    & \bm \xi_{2l+1, kn-a} = X_{2l+1,kn-a} X_{2l+1, (k+1)n-a}\\
    &\hspace{65pt}  \prod_{i=0}^a X_{2l,kn-a+i}\prod_{i=a+1}^{n}X_{2l+2, kn-a+i}
\end{align*}
with $l$, $k\in \mathbb{Z}$ and $0\leq a \leq n-1$. The $n=1$ case reduces to the rotated toric code in Sec.~\ref{sec:embp_3dftc_clc}. In this basis, the translational part of $\ee-\mm$ symmetry is still a shift to the right. Thus the stabilizer group is given by
\begin{align*}
    \text{TC}^{(n)} = 
    \langle \,
    \bm \beta_{2l, z} ,\,
    \bm \xi_{2l+1, z}
    \,|\,
    l,\, z \in \mathbb{Z}
    \, \rangle
\end{align*}

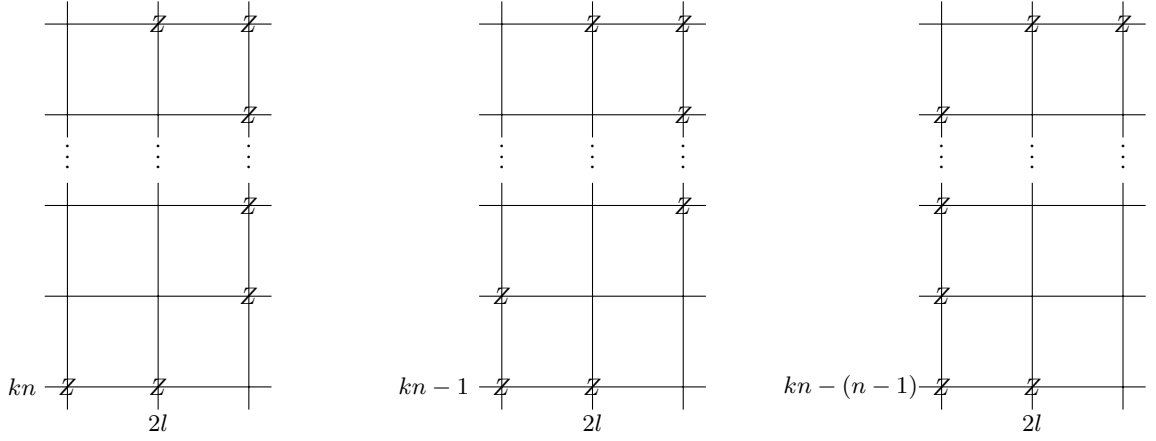
\begin{figure*}
\begin{minipage}{0.3\textwidth}
    \begin{tikzpicture}
    [scale = 0.6, baseline=2cm]
        \node at (0, 0) {$Z$};
        \node at (2, 0) {$Z$};
        
        \node at (2, 8) {$Z$};
        \node at (4, 8) {$Z$};
        \node at (4, 6) {$Z$};
        \node at (4, 4) {$Z$};
        \node at (4, 2) {$Z$};

        \node at (0 , 5.25) {$\vdots$};
        \node at (2 , 5.25) {$\vdots$};
        \node at (4 , 5.25) {$\vdots$};
        
        \node at (2,-0.8) {$2l$};        
        \node at (-1, 0) {$kn$};
        
        \draw[step = 2] (-0.5,-0.5) grid (4.5,4.5);
        \draw[step = 2] (-0.5,5.5) grid (4.5,8.5);
    \end{tikzpicture}
\end{minipage}
\begin{minipage}{0.3\textwidth}
    \begin{tikzpicture}
    [scale = 0.6, baseline=2cm]
        \node at (0, 0) {$Z$};
        \node at (2, 0) {$Z$};
        \node at (0, 2) {$Z$};
        
        \node at (2, 8) {$Z$};
        \node at (4, 8) {$Z$};
        \node at (4, 6) {$Z$};
        \node at (4, 4) {$Z$};

        \node at (0 , 5.25) {$\vdots$};
        \node at (2 , 5.25) {$\vdots$};
        \node at (4 , 5.25) {$\vdots$};
        
        \node at (2,-0.8) {$2l$};
        \node at (-1.5, 0) {$kn-1$};
        
        \draw[step = 2] (-0.5,-0.5) grid (4.5,4.5);
        \draw[step = 2] (-0.5,5.5) grid (4.5,8.5);
    \end{tikzpicture}
\end{minipage}
\begin{minipage}{0.3\textwidth}
    \begin{tikzpicture}
    [scale = 0.6, baseline=2cm]
        \node at (0, 0) {$Z$};
        \node at (2, 0) {$Z$};
        \node at (0, 2) {$Z$};
        \node at (0, 6) {$Z$};
        \node at (0, 4) {$Z$};
        
        \node at (2, 8) {$Z$};
        \node at (4, 8) {$Z$};

        \node at (0 , 5.25) {$\vdots$};
        \node at (2 , 5.25) {$\vdots$};
        \node at (4 , 5.25) {$\vdots$};
        
        \node at (2,-0.8) {$2l$};
        \node at (-2, 0) {$kn-(n-1)$};
        
        \draw[step = 2] (-0.5,-0.5) grid (4.5,4.5);
        \draw[step = 2] (-0.5,5.5) grid (4.5,8.5);
    \end{tikzpicture}
\end{minipage}
\begin{minipage}{0.95\textwidth}
    \caption{From left to right, $Z$-stabilizers $\bm \beta_{2l,kn}$, $\bm \beta_{2l,kn-1}$, and $\bm \beta_{2l,kn-(n-1)}$ are shown. The $X$-stabilizers $\bm \xi_{2l+1,z}$ are obtained by shifting $\bm \beta_{2l,z}$ to the right by one site, and replace $Z$ by $X$.}
    \label{fig:FnZStab}
\end{minipage}
\end{figure*}

Using this basis, we present the coupled-layer construction TC $\otimes_{\ee - \mm}$ TC$^{(n)}$. First we choose representatives of $\mathbb{Z}\backslash$TC$^{(n)}$,
\vspace{-0.5em}
\begin{itemize}
    \item Qubits: representatives have coordinates $(0,z)$.
    \item Stabilizers: representatives are $\bm \beta_{0,z}$. 
\end{itemize}
\vspace{-0.5em}
We group these qubits into sets of $n$ in the following way. For each $k\in \mathbb{Z}$, define $k[i]:= kn-(i-1)$, with $i = 1, \cdots, n$. We then let the $k$th set contain qubits $\{\,(0, k[i]) \,|\, i = 1, \cdots, n\,\}$. It follows that there are $n$ stabilizer representatives supported on the qubit set $k$ and $k+1$, which are $\{ \, \bm \beta_{0, k[i]} \,|\, i = 1, \cdots, n  \,\}$

The stacking is given as follows. For each $k\in \mathbb{Z}$, we introduce $n$ copies of TC, putting $n$ qubits on each horizontal edge in the 3D cubic lattice. For each pair of vertex and $\bm \beta_{0,z}$, introduce an $X$ ancilla. Since there are $n$ stabilizers support on the qubit set $k$ and $k+1$, we put $n$ qubits on each vertical edge. Therefore the total Hilbert space is spanned by $n$ copies of edges of the 3D cubic lattice. The stabilizer group is
\begin{align*}
    \mathcal{S}_0 = 
    \Bigg\langle\,
    \bm a_v^{(k[i])},\,
    \bm b_p^{(k[i])},\,
     \begin{tikzpicture}[baseline = 0.5cm, scale = 0.6]
        \node at (0,1) {$X_1$};
        \draw (0,0) -- (0,2);      
    \end{tikzpicture}
    ,\,\cdots,\,
    \begin{tikzpicture}[baseline = 0.5cm, scale = 0.6]
        \node at (0,1) {$X_n$};
        \draw (0,0) -- (0,2);      
    \end{tikzpicture}
    \,\Bigg\rangle.
\end{align*}

The code switching terms are given by, for $i = 1, \cdots n$,
\begin{align*}
    \bm \beta\left(
    \hspace{5pt}
    \begin{tikzpicture}[scale = 0.5, baseline =-0.1 cm]
            \draw (0,0) -- (2,0);
            \node[fill, circle, black, inner sep = 1pt] at (1, 0) {};     
    \end{tikzpicture}
    \hspace{5pt}
    ,
    \hspace{5pt}
    \beta_{0,k[i]}
    \right) 
    \hspace{5pt}
    = 
    \hspace{5pt}
    \begin{tikzpicture}[scale = 1, baseline = 1cm]          
            \coordinate (A) at (0,0,0);
            \coordinate (B) at (0,2,0);
            \coordinate (C) at (2,2,0);
            \coordinate (D) at (2,0,0);
%
            \coordinate (U) at (0,0,-2);
            \coordinate (O) at (2,2,2);
%
%
 %
            \draw (A) -- (B) --(C) -- (D) -- cycle;
            \draw [dotted] (A) -- (U);
            \draw (C) -- (O);
            \foreach \X/\Y in { A/B, B/C, C/D, D/A}
            \node at ($( \X )!0.5!( \Y )$) {$Z_i$}; 
            \node at (0, 0, -1.1) {$X_1\cdots X_i$};
            \node at (2, 2, 1.1) {$X_i\cdots X_n$};
    \end{tikzpicture}
\end{align*}
\begin{align*}
    \bm\beta\left(\hspace{5pt}
    \begin{tikzpicture}[scale = 0.8, baseline = -0.4cm]
            \draw (0,0,0) -- (0,0,2);
            \node [fill, circle, black, inner sep = 1pt] at (0,0, 1) {};            
    \end{tikzpicture}
    \hspace{5pt}    
    ,
    \hspace{5pt}
    \beta_{0, k[i]}
    \right)
    \hspace{5pt}
    = 
    \hspace{5pt}
    \begin{tikzpicture}[scale = 1, baseline = 0.5cm]          
            \coordinate (A) at (0,0,0);
            \coordinate (B) at (0,0,2);
            \coordinate (C) at (0,2,2);
            \coordinate (D) at (0,2,0);
%
            \coordinate (U) at (-2,0,0);
            \coordinate (O) at (2,2,2);
%
%
 %
            \draw (A) -- (B) --(C) -- (D) -- cycle;
            \draw [dotted] (A) -- (U);
            \draw (C) -- (O);
            \foreach \X/\Y in { A/B, C/D}
            \node at ($( \X )!0.5!( \Y )$) {$Z_i$}; 
            \node at (0,1.2,2) {$Z_i$};
            \node at (0,0.8,0) {$Z_i$};
            \node at (-1.1, 0, 0) {$X_1\cdots X_i$};
            \node at (1.1, 2, 2) {$X_i\cdots X_n$};
    \end{tikzpicture}
\end{align*}
and the remaining commuting terms in $\mathcal{S}_0$ are
\begin{align*}
    \bm \xi \left( a_v, k[i]
    \right)
    \hspace{5pt}
    = 
    \hspace{5pt}
    \begin{tikzpicture}[scale = 1, baseline = 0cm]
        \coordinate (O) at (0,0,0);
        \coordinate (A) at (-1.5,0,0);
        \coordinate (B) at (1.5,0,0);
        \coordinate (C) at (0,-1.5,0);
        \coordinate (D) at (0,1.5,0);
        \coordinate (E) at (0,0,1.5);
        \coordinate (F) at (0,0,-1.5);
 %
        \draw (A) -- (B);
        \draw (C) -- (D); 
        \draw (E) -- (F);
        \foreach \X/\Y in { O/E, O/F}
        \node at ($( \X )!0.5!( \Y )$) {$X_i$};
        \foreach \X/\Y in { O/C, O/D}
        \node at ($( \X )!0.5!( \Y )$) {$X_i$};
        \foreach \X/\Y in {O/A, O/B}
        \node at ($( \X )!0.5!( \Y )$) {$X_i$}; 
    \end{tikzpicture}
\end{align*}
\begin{align*}
    \bm \zeta\left(b_p, k[i]
    \right)
    \hspace{5pt}
    = 
    \hspace{5pt}
    \begin{tikzpicture}[scale = 1, baseline = -0.5cm]   
        \coordinate (A) at (0,0,0);
        \coordinate (B) at (2,0,0);
        \coordinate (C) at (2,0,2);
        \coordinate (D) at (0,0,2);
%
        \coordinate (O) at (2,2,2);
        \coordinate (U) at (0,-2,0);
%
%
        \draw  (A) -- (B) -- (C) -- (D) -- cycle;
        \draw [dotted] (A) -- (U);
        \draw (C) -- (O);
%
%
        \foreach \X/\Y in { A/B, B/C, C/D, D/A}
        \node[circle, inner sep=1.5pt] at ($( \X )!0.5!( \Y )$) {$Z_i$};
        \node[circle, inner sep=1.5pt] at (2,1.2,2) {$X_i \cdots X_n$};
        \node[circle, inner sep=1.5pt] at (0,-1.2,0) {$X_1 \cdots X_i$};
    \end{tikzpicture}
\end{align*}
As mentioned in the beginning of the section, $n = 1$ gives the 3D fermionic TC, and $n = 2$ gives the 3-fermion Walker-Wang model. Indeed, this coupled-layer construction can be viewed as anyon condensation shown in Figure~\ref{fig:ww_fermion}. For each $k\in \mathbb{Z}$, the bound state in $n+1$ neighboring layers $\mathsf{f}_k \mm_{k+1} \cdots \mm_{k+n-1} \mathsf{f}_{k+n}$ is condensed, where the subscripts label the TC layers \cite{wang2013boson}. After the condensation, the boundary is generated by $n$ fermions $\{f_1, \cdots, f_n\}$, which also have fermionic mutual statistics. The corresponding topological action of the bulk is \cite{chen2023higher}
\begin{align*}
    S_n = \frac{1}{2}\int \sum_{i=1}^n B_i \cup B_i + \sum_{i<j} B_i \cup B_j.
\end{align*}
where $B_i$ are dynamical 2-form gauge fields. We conclude that the resulting code is the Walker-Wang model for the anyon theory $\{f_1, \cdots, f_n\}$
\begin{align*}
 \text{TC} \otimes_{\ee - \mm} 
     \text{TC}^{(n)} =
    \text{WW}(\{f_1, \cdots, f_n\}).
\end{align*}
On a three torus, this code has 3 logicals if $n$ is odd, and zero logicals if $n$ is even.

\begin{figure*}
    \begin{minipage}{0.28\textwidth}
        \includegraphics[width=1\linewidth]{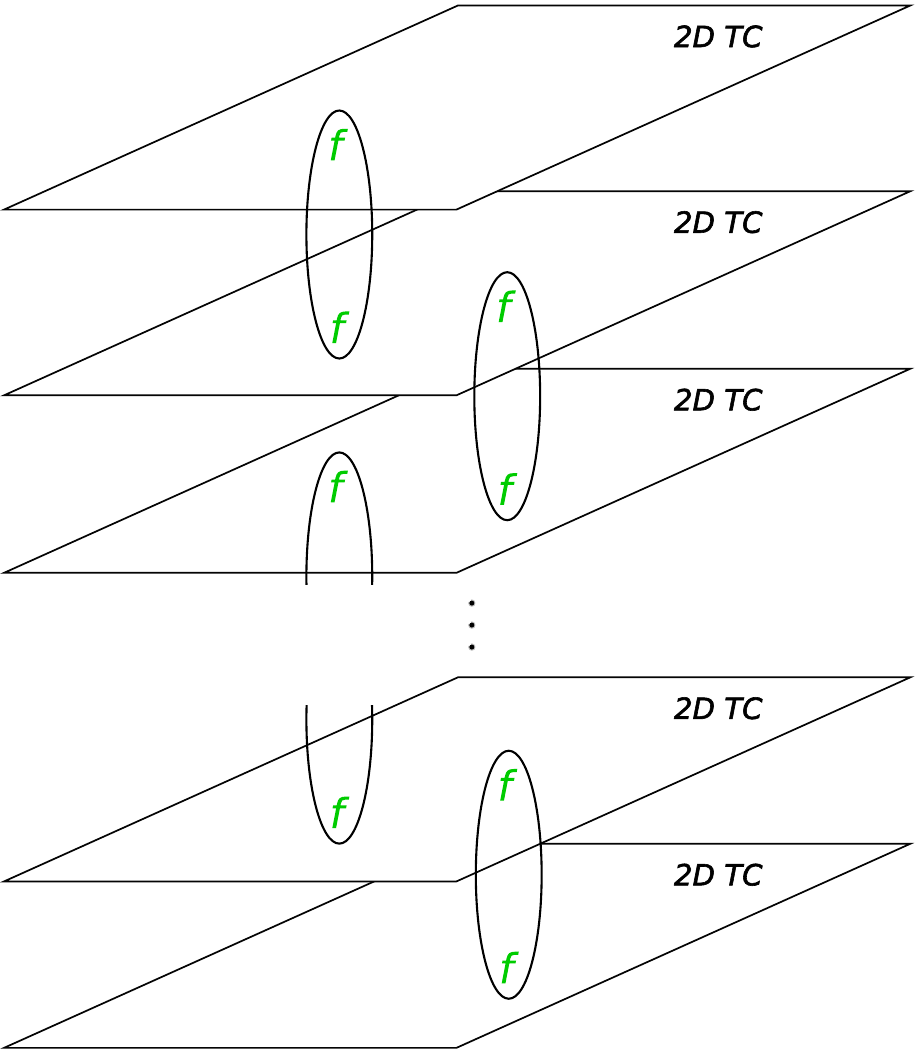}
    \end{minipage}
    \hspace{0.5cm}
    \begin{minipage}{0.28\textwidth}
        \includegraphics[width=1\linewidth]{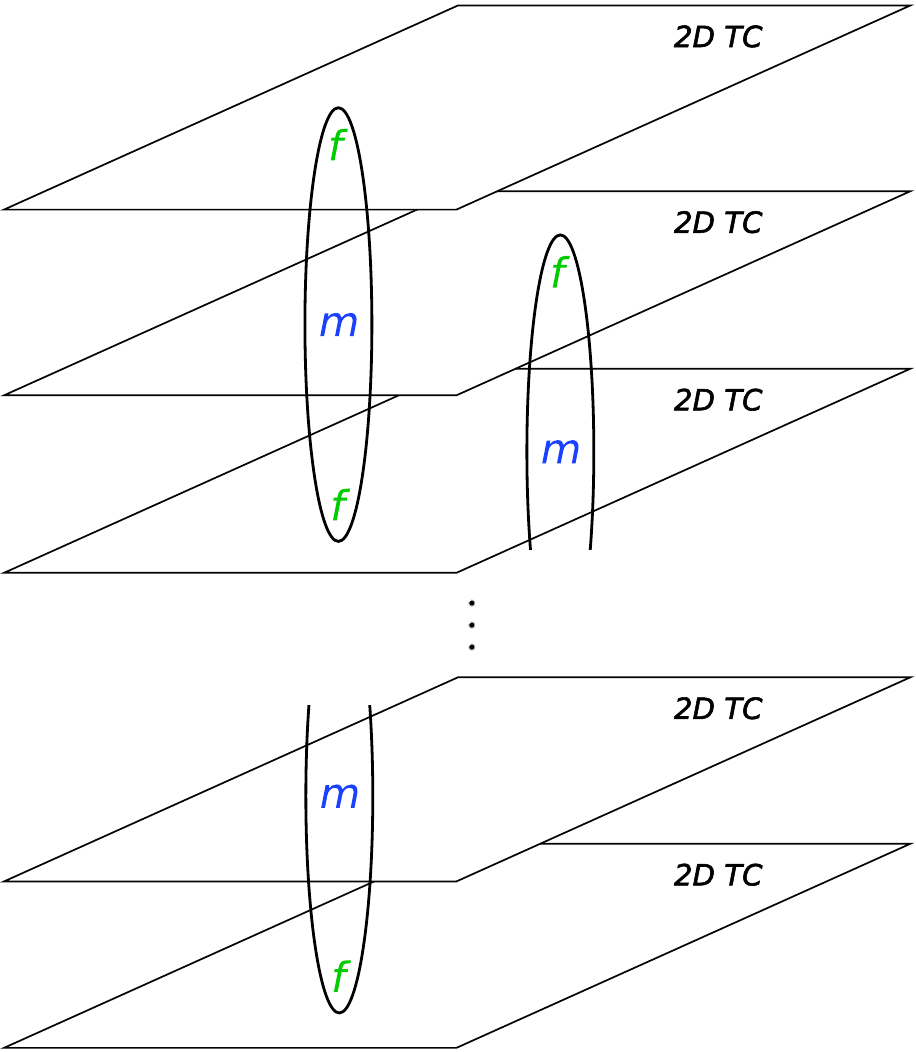}
    \end{minipage}
    \hspace{0.5cm}
    \begin{minipage}{0.28\textwidth}
        \includegraphics[scale=0.32]{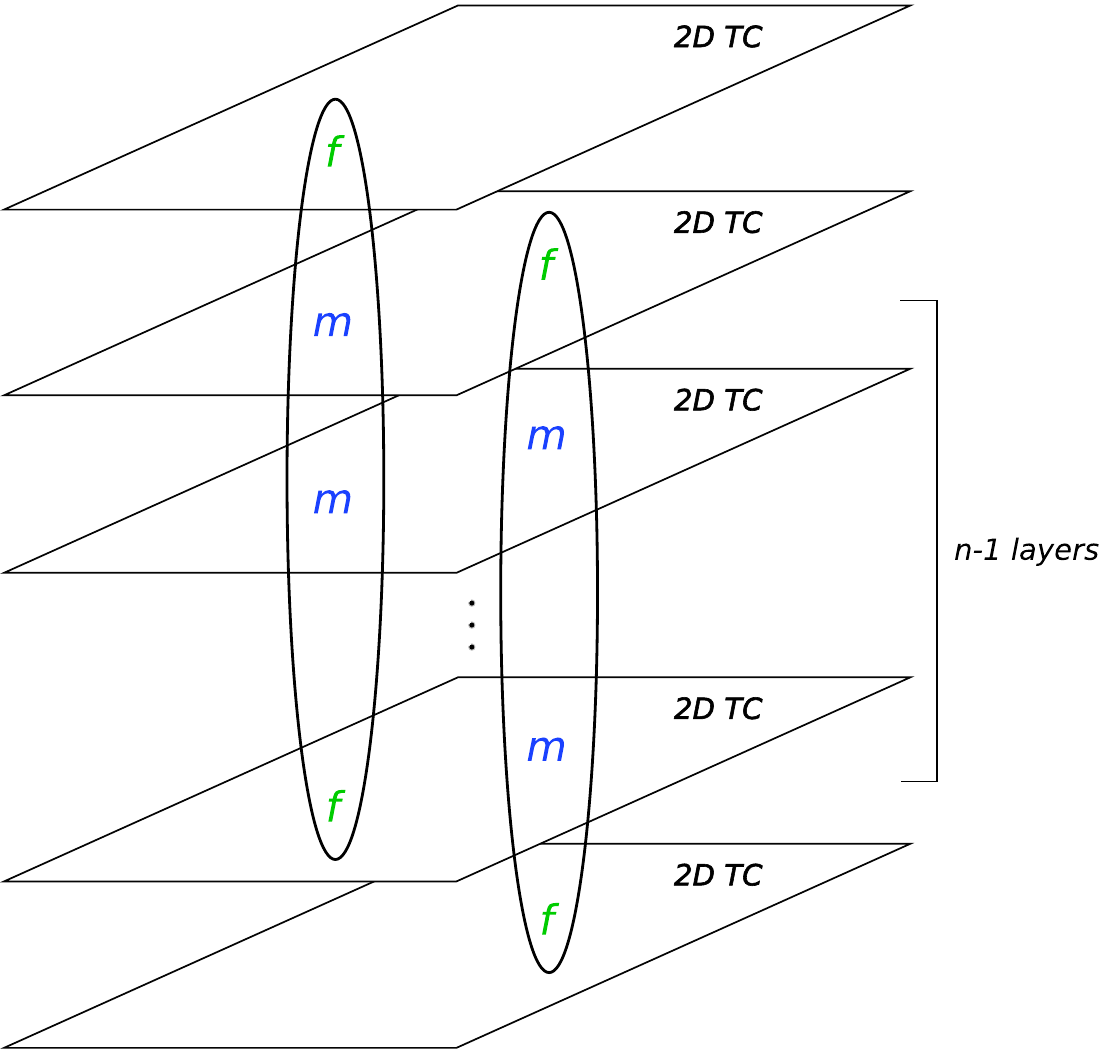}
    \end{minipage}
    \caption{The condensation schemes of the codes TC $\otimes_{\ee-\mm}$ TC$^{(n)}$. From left to right, it gives respectively $n=1$, fermionic TC; $n=2$, 3-fermion Walker-Wang model; and generic $n$.}
    \label{fig:ww_fermion}
\end{figure*}

\section{Balanced Coupled-Layer Codes}\label{sec:bclc}

In this section, we introduce a framework which combines ideas from the coupled-layer codes in Sec.~\ref{sec:clc} and the $\ee-\mm$ balanced product from Sec.~\ref{sec:embp}, which includes the traditional balanced product. The construction can be performed on non-CSS codes, hence cannot be written in terms of the three-term chain complex. As a result, we directly describe the construction from the coupled-layer perspective. 

Start with a pair of Pauli stabilizer codes, $\mathcal{H}$ with stabilizers $\{ \bm h_t \}$ defined on qubits $Q_1$; and $\mathcal{F}$ with stabilizers $\{\bm f_s\}$ defined on qubits $Q_2$, we moreover assume each $\bm f_s$ is of a pure Pauli type similar to Sec.~\ref{sec:clc} and denote the corresponding Pauli type $\tau(\bm f_s)$. We will let $\mathcal{H}$ take the role of CSS$_1$, and will be stacked; and $\mathcal{F}$ take the role of CSS$_2$, which is used as a recipe for coupling the layers together. We assume $\mathcal{H}$ and $\mathcal{F}$ are invariant under a right $G$-action and left $G$-action respectively. Similar to Sec.~\ref{sec:gqc}, we assume that the actions are defined by two components. First, the action freely permutes qubits and supports of stabilizers by $q_1\mapsto q_1\cdot g$ and $q_2 \mapsto g\cdot q_2$. Second, each $g\in G$ applies an onsite Clifford unitary $\bm h_t \mapsto \bar U^\dagger_g \bm h_t \bar U_g$ (resp. $\bm f_s \mapsto \bar U_g \bm f_s \bar U_g^\dagger$), where $\bar U_g$ is the product of single-qubit Clifford unitaries $U_g$ over all qubits, which satisfies $U_gU_h = U_{gh}$. The total action on stabilizers by $g$ is denoted $\mathcal{U}_g$. For instance, given operator $\bm O = \prod_{q_1} \bm O_{q_1}$ on $Q_1$, we have
\begin{align*}
    \mathcal{U}_g(\bm O)_{q_1\cdot g} = U_g^\dagger \bm O_{q_1} U_g.
\end{align*}
Note that we require the unitary part of $G$-action is the same for $\mathcal{H}$ and $\mathcal{F}$ (up to conjugation).

Both $\mathcal{H}$ and $\mathcal{F}$ can be quotient by $G$ as prescribed in Sec.~\ref{sec:gqc}. However, we want to use $\mathcal{F}$ to keep track of the code switching on stacks of $\mathcal{H}$, thus we need to quotient $\mathcal{F}$ and choose representatives of $G\backslash \mathcal{F}$. To this end, we choose qubit representatives $\{\tilde q_2\}$, and stabilizer representatives $\{\tilde{\bm f}\}$ of $G\backslash \mathcal{F}$. Given $q_2$ (resp. $\bm f$), we define $g_{q_2}$ (resp. $g_f$) to be the unique group element relating to its representative, that is $q_2 = g_{q_2} \cdot \tilde{q}_2$ (resp. $\bm f = \mathcal{U}_{g_f}(\tilde{\bm f})$). In addition, for each representative $\tilde{\bm f}$, we need to specify
\vspace{-0.5em}
\begin{itemize}
    \item Excitations: $\mathcal{E}^{(\tilde f)} = \{ e^{(\tilde f)} \}$, where each $\bm e^{(\tilde f)}$ is of Pauli type $\tau(\tilde{\bm f})$. Formally, the excitations can be written in terms of a linear map $\mathcal{E}^{(\tilde f)} \ra Q_1$ similar to the mapping cone. Moreover, we require the supports of $\bm e^{(\tilde f)}$ to be mutually disjoint and freely permuted by $G$.
    \vspace{-0.5em}
    \item Algebra-preserving map: choose map $\bm{\Gamma}^{(\tilde f)}$ defined on the excitation algebra $\mathcal{A}^{(\tilde f)} := \langle \bm h_t, \, \bm e^{(\tilde f)}\rangle$ which takes each generator to a Pauli operator on an auxiliary Hilbert space $Q^{(\tilde f)}$, while preserving commutation relation.
\end{itemize}
\vspace{-0.5em}
The condition on the support of $\bm e^{(\tilde f)}$ essentially says the excitations behave like single Pauli operators under the group action.

Next, we describe the coupled-layer construction. For each qubit representative $\tilde q_2$, introduce a copy of $\mathcal{H}$, with stabilizers $\{\,\bm h_t^{(\tilde q_2)} \,\}$. For each stabilizer representative $\tilde{\bm f}$, introduce the set of stabilizers $\{\, \bm{\Gamma}^{(\tilde f)} (\bm h_t) \,\}$ living on the auxiliary qubits $Q^{(\tilde f)}$. Thus the stabilizer group of the stacked system is
\begin{align*}
    \mathcal{S}_0 = \langle  \,\bm 
    h_t^{(\tilde q_2)}, \bm{\Gamma}^{(\tilde f)}(\bm h_t)  
    \,  \rangle
\end{align*}
defined on the total Hilbert space
\begin{align*}
    \left(Q_1\otimes (G\backslash Q_2)\right) \oplus \left( \bigoplus_{\tilde f} Q^{(\tilde f)}\right).
\end{align*}

Next we perform code switching. For each pair of $\tilde f$ and $e \in \mathcal{E}^{(\tilde f)}$, enforce
\begin{align}
    \bm A(e, \tilde f) := \bm{\Gamma}^{(\tilde f)}(\bm e) \prod_{q_2\in \tilde f}  \mathcal{U}_{g_{q_2}}(\bm e^{(\tilde q_2)}), 
    \label{eq:Atermbalanced}
\end{align}
where $\bm e^{(\tilde q_2)}$ is the excitation $\bm e$ in layer $\tilde q_2$. Note the ordering of terms in the product does not matter. Even if multiple $q_2 \in \tilde f$ belong to the same orbit so that multiple $\mathcal{U}_{g_{q_2}}(\bm e^{(\tilde q_2)})$ live on the same layer, their supports are disjoint by assumption, therefore they always commute. We show these terms pairwise commute in Appendix \ref{app:balancedcommute}, hence the code switching is well-defined.

After the code switching, the commuting terms in $\mathcal{S}_0$ remain, which are
\begin{align*}
    \bm B(h_t, \tilde q_2) := \bm h_t^{(\tilde q_2)} \prod_{f\ni \tilde q_2} \bm{\Gamma}^{(\tilde f)}(\mathcal{U}_{g_f}(\bm h_t)). 
\end{align*}
Again, the ordering of the product does not matter since this term is in the stabilizer group $\mathcal{S}_0$.

\vspace{0.1cm}

We can check $\bm A(e, \tilde f)$ and $\bm B(h_t, \tilde q_2)$ commute explicitly. The two terms have potential non-zero overlap only when there exists $q_2'\in \tilde f$ such that $[q_2'] = [\tilde q_2]$, which is equivalent to the existence of $f' \ni \tilde q_2$ such that $[f'] = [\tilde f]$. In that case the potential overlap is on layer $\tilde q_2$ and $\tilde f$. Collect the set of group elements
\begin{align*}
    C = \{g_{q_2'} 
    \,|\,
    q_2'\in \tilde f,\,
    \supp{\mathcal{U}_{g_{q_2'}}(\bm e)} \cap \supp{\bm h_t} \neq \emptyset\}.
\end{align*}
We write
\begin{align*}
    & \bm A(e, \tilde f) =  \bm{\Gamma}^{(\tilde f)}(\bm e) \prod_{g\in C}  \mathcal{U}_{g}(\bm e^{(\tilde q_2)})  \times [\cdots],\\
    & \bm B(h_t, \tilde q_2) = \bm h_t^{(\tilde q_2)} \prod_{g\in C} \bm{\Gamma}^{(\tilde f)}(\mathcal{U}_{g^{-1}}(\bm h_t))
    \times[\cdots].
\end{align*}
The $[\cdots]$ in $\bm A(e,\tilde f)$ have no overlap with any term in $\bm B(h_t, \tilde q_2)$, hence commutes. On the other hand, given $\bm{\Gamma}^{(\tilde f')}(\mathcal{U}_{g_{f'}}(\bm h_t))$ in $[\cdots]$ in $\bm B(h_t, \tilde q_2)$, if $\tilde f' \neq \tilde f$, then this term has no overlap with $\bm A(e, \tilde f)$; if $\tilde f' = \tilde f$, then we must have $g_{f'}^{-1} \notin C$, therefore
\begin{align*}
    \supp{\mathcal{U}_{g_{f'}^{-1}}(\bm e)} \cap \supp{\bm h_t} = \emptyset.
\end{align*}
Therefore
\begin{align*}
    \bm e \cdot \mathcal{U}_{g_{f'}}(\bm h_t) 
    &= \mathcal{U}_{g_{f'}} (\mathcal{U}_{g_{f'}^{-1}}(\bm e) \cdot \bm h_t)\\
    &=\mathcal{U}_{g_{f'}} (
     \bm h_t \cdot \mathcal{U}_{g_{f'}^{-1}}(\bm e))\\
     & = \mathcal{U}_{g_{f'}}(\bm h_t) \cdot \bm e.
\end{align*}
Since $\bm{\Gamma}^{(\tilde f)}$ preserves commutation relation, we see $\bm{\Gamma}^{(\tilde f)}(\mathcal{U}_{g_{f'}}(\bm h_t))$ commutes with $\bm{\Gamma}^{(\tilde f)}(\bm e)$, hence it commute with all the terms in $\bm A(e, \tilde f)$. This shows the $[\cdots]$ in both terms commute with everything, so we do not need to keep track of them in studying the commutation. Let sgn$(\bm A,\bm B) := \bm A\bm B\bm A\bm B$ to be the sign upon commuting a pair of operators, we have
\begin{align*}
    \hspace{-5pt}\text{sgn}(\bm A(e, \tilde f), \bm B(h_t, \tilde q_2))& \\
    &\hspace{-90pt}=  \text{sgn}\Big(\bm{\Gamma}^{(\tilde f)}(\bm e),\,
    \prod_{g\in C} \bm{\Gamma}^{(\tilde f)}(\mathcal{U}_{g^{-1}}(\bm h_t))\Big)
    \text{sgn}\Big(\prod_{g\in C}  \mathcal{U}_{g}(\bm e^{(\tilde q_2)}),\,
    \bm h_t^{(\tilde q_2)} \Big)\\
    &\hspace{-90pt} = \prod_{g\in C}
    \left[\text{sgn}
    \left(\bm{\Gamma}^{(\tilde f)}(\bm e),\,
    \bm{\Gamma}^{(\tilde f)}(\mathcal{U}_{g^{-1}}(\bm h_t))\right)
    \text{sgn}\left(\mathcal{U}_{g}(\bm e^{(\tilde q_2)}),\,
    \bm h_t^{(\tilde q_2)} \right)\right].
\end{align*}
Since for each $g\in C$, we have
\begin{align*}
    \text{sgn}\left(\mathcal{U}_{g}(\bm e),\,
    \bm h_t \right) &= \text{sgn}\left(\bm e,\,
    \mathcal{U}_{g^{-1}}( \bm h_t)\right)\\
    &= \text{sgn}
    \left(\bm{\Gamma}^{(\tilde f)}(\bm e),\,
    \bm{\Gamma}^{(\tilde f)}(\mathcal{U}_{g^{-1}}(\bm h_t))\right).
\end{align*}
Thus we see the sign always cancel in pairs, and we have
\begin{align*}
    \text{sgn}(\bm A(e, \tilde f), \bm B(h_t, \tilde q_2)) = 1.
\end{align*}
This shows that the two terms always commute, and we have
\begin{align*}
    \mathcal{S} \supseteq \langle \,
    \bm A(e, \tilde f), \,
    \bm B(h_t, \tilde q)
    \, \rangle.
\end{align*}
We take the above two terms as defining the stabilizer group of the balanced coupled-layer code. However, from the code switching perspective, there can be more terms in $\mathcal{S}_0$ that remain in $\mathcal{S}$, as in the case of the tensor product. For example, if certain product $\prod_{t\in T} \bm h_t$ commutes with all $\bm e^{(\tilde f)}$ for a given $\tilde f$, then clearly $\prod_{t\in T} \bm{\Gamma}^{(\tilde f)}(\bm h_t)$ remains in $\mathcal{S}$. In particular, if there is a nontrivial relation between the stabilizers $\prod_{t\in T} \bm h_t = 1$, then $\prod_{t\in T} \bm{\Gamma}^{(\tilde f)} (\bm h_t)$ is in $\mathcal{S}$ for all $\tilde f$. We do not attempt to write down all such terms in the most general form.

This construction generalizes balanced product in three major ways. First, we are not working on the chain complex level, so neither $\mathcal{H}$ nor $\mathcal{F}$ needs to be a CSS code. Second, we are allowed to condense more general excitations than a single $X$ or $Z$. As we have seen in Sec~\ref{sec:cssclc} and Sec~\ref{sec:clc}, condensing excitations beyond single $X$ or $Z$ can produce exotic phases of matter like fractons. Finally, the group actions are allowed to have a unitary component in addition to permutation. As we have seen in Sec.~\ref{sec:embp}, this allows us to produce non-CSS codes even if the two input codes are CSS.

To see this construction includes balanced product as a special case, we make the following reductions. Take two CSS codes, we regard CSS$_1: A_1\ra Q_1 \ra B_1$ as $H$ and CSS$_2: A_2\ra Q_2 \ra B_2$ as $F$. Let $G$ act only by permutation, that is, take $U_g = id$. Choose $\{\tilde a_2\}$, $\{\tilde b_2\}$, $\{\tilde q_2\}$ to be representatives of $G\backslash A_2$, $G\backslash B_2$ and $G\backslash Q_2$, respectively. As for the condensation data, we take $\{\bm e^{(\tilde a_2)}\} := \{X_{q_1}\}$ to be single Pauli $X$, and $\{\bm e^{(\tilde b_2)}\} := \{Z_{q_1}\}$ to be single Pauli $Z$. We take $\bm{\Gamma}^{(\tilde a_2)}$ to be the gauging map which gauges all $X$ symmetries in CSS$_1$, and $\bm{\Gamma}^{(\tilde b_2)}$ to be the gauging map which gauges all $Z$ symmetries. The code switching terms are
\begin{align*}
    \bm A(q_1,  \tilde a_2) 
    &= \bm{\Gamma}^{(\tilde a_2)}(X_{q_1}) \prod_{q_2\in \tilde a_2} \mathcal{U}_{g_{q_2}}(X_{q_1}^{(\tilde q_2)})\\ 
    &= \prod_{b_{1} \ni q_1} X_{b_1,\tilde a_2}\prod_{q_2\in \tilde a_2} X_{q_1\cdot g_{q_2}, \tilde q_2}, \\
    \bm A(Z_{q_1},  \tilde b_2) &= \bm{\Gamma}^{(\tilde b_2)}(Z_{q_1}) \prod_{q_2\in \tilde b_2} \mathcal{U}_{g_{q_2}}(Z_{q_1}^{(\tilde q_2)})\\
    &= \prod_{a_{1} \ni q_1} Z_{a_1,\tilde b_2}\prod_{q_2\in \tilde b_2} Z_{q_1\cdot g_{q_2}, \tilde q_2}, 
\end{align*}
which are just $\bm \alpha(q_1, \tilde a_2)$ and $\bm \beta(q_1, \tilde b_2)$ in balanced product. The remaining terms in $\mathcal{S}_0$ are
\begin{align*}
    \bm B(a_1, \tilde q_2) &= \bm{a_1}^{(\tilde q_2)} \prod_{a_2 \ni \tilde q_2}\bm{\Gamma}^{(\tilde a_2)}(\mathcal{U}_{g_{a_2}}(\bm{a_1})) \prod_{b_2 \ni \tilde q_2}\bm{\Gamma}^{(\tilde b_2)}(\mathcal{U}_{g_{b_2}}(\bm{a_1}))\\ 
    & = \bm{a_1}^{(\tilde q_2)}   \prod_{b_2 \ni \tilde q_2}X_{a_1\cdot g(b_2), \tilde b_2},\\
    \bm B(b_1, \tilde q_2) &= \bm{b_1}^{(\tilde q_2)} \prod_{a_2 \ni \tilde q_2}\bm{\Gamma}^{(\tilde a_2)}(\mathcal{U}_{g_{a_2}}(\bm{b_1}))
    \prod_{b_2 \ni \tilde q_2}\bm{\Gamma}^{(\tilde b_2)}(\mathcal{U}_{g_{b_2}}(\bm{b_1}))\\
    &= \bm{b_1}^{(\tilde q_2)}   \prod_{a_2 \ni \tilde q_2}Z_{b_1\cdot g(a_2), \tilde a_2}.
\end{align*}
These are exactly the terms $\bm \xi(a_1, \tilde q_2)$ and $\bm \zeta(b_1, \tilde q_2)$ in the balanced product. Hence, we have recovered the balanced product as a special case, when the action is free.

It is also not hard to see that this includes the $\ee-\mm$ balanced product as a special case. If we choose $\mathcal{H}$ to be CSS$_1$, and $\mathcal{F}$ is chosen to be CSS$_2$. The unitary part of the action is either trivial or Hadamard $\mathfrak{h}$. The $X$-excitations ($Z$-excitations) are chosen to be single $X$ ($Z$), and the algebra-preserving map $\bm{\Gamma}_x$ ($\bm{\Gamma}_z$) is gauging all $X$-stabilizers ($Z$-stabilizers), then we recover the $\ee-\mm$ balanced product defined in Sec~\ref{sec:embp}.

\section{Conclusion and Outlook}\label{sec:conclusion}

In this work, we studied a new class of stabilizer codes which are constructed from two constituent codes. Inspired by the coupled-layer construction in condensed matter physics, we presented several new classes of codes which go far beyond the tensor and balanced product codes.
\vspace{-0.4em}
\begin{enumerate}
    \item The CSS coupled-layer code introduces the idea of condensing general excitations, instead of the excitations created by single Pauli $X$ or $Z$ as in the case of tensor product. This is captured by a few related concepts, namely the mapping cone, which in turn is related to partial gauging/partial lattice surgery between the layers. Equivalently, the coupling between layers can also be formulated via a pair of excitation algebras and an algebra-preserving map. We showed how the X-cube code can be constructed within this framework. Using the excitation algebra and algebra-preserving map as fundamental objects, we were able to generalize the construction to non-CSS codes. This framework contains the XYZ product code, in particular, the Chamon model. 

    \item We introduce the concept of a generalized quotient where the group action acts not only by permutation, but also accompanied by an onsite unitary. As an example, we discussed the $\ee-\mm$ balanced product where the onsite unitary is Hadamard. We also relate this construction to an ordinary balancing of CSS codes via the symplectic doubling \cite{KovalevPryadko13,liu2024subsystem} in Appendix~\ref{sec:sd}. Notable examples constructed under this framework are the 3D fermionic TC and $3$-fermion Walker-Wang model. 

    \item  Combining the ideas from $1$ and $2$, we attempted to provide a unifying framework, which condenses general excitations, and at the same time is twisted by a group action, which is combination of permutation and onsite unitary. Although we do not have an illuminating example in low dimensions that goes beyond the previous constructions, it would be interesting to construct codes from our most general formalism.
\end{enumerate}
\vspace{-0.4em}

 We also note that there are several avenues for further generalization of our construction. First, the excitations being condensed are assumed to be created from Pauli operators with disjoint support, which makes them behave similarly to single-Pauli operators. Second, we only considered single-qubit Cliffords in defining the group action. It would be elegant to extend the construction to arbitrary excitations, as well as grouping the qubits into clusters and allowing the group to act within each cluster. One could also further ask whether there is a unifying mathematical structure that encompasses this construction \cite{hastings2021fiber, panteleev2024maximally}. For instance, in some sense the outcome can be viewed as a fiber bundle over the second code, whose fiber is the first code. Can such abstract points of view bring insights, and potentially alleviate the cumbersome calculations in Appendix~\ref{app:balancedcommute}?

From the physical perspective, it would be interesting to understand how to interpret the result of the generalized quotient in Sec~\ref{sec:gqc} as a physical operation. For example, we have seen that performing the $\ee-\mm$ quotient on the 4D TC is the 3D fermionic TC. It seems reasonable that performing a compactification of the 4D toric code with an $\ee-\mm$ duality defect inserted should give rise to the fermionic toric code. Showing that this is indeed the case, and also similarly for the larger family we constructed would be interesting future work.

Another possible direction is to investigate the error correction properties of the coupled-layer codes presented, in particular, the scaling of logicals and code distance. Currently, our work is inspired by and mostly focuses on quantum phases of matter. We showed that various fracton phases and fermionic phases are reproduced under our framework. However, we have not commented on whether our construction can produce codes with good parameters. Indeed, it is observed in Sec.~\ref{sec:cssclc_4dqc} that naive choices can lead to codes with constant distances. Hence, it is interesting to study what assumptions on the input codes, as well as the excitation algebra and algebra-preserving maps, can lead to codes with practical parameters. Moreover, we hope that the coupled-layer construction can offer a new perspective on error correction or decoding properties given those of constituent codes. Finally, we note that our construction can be straightforwardly generalized to qudits, although generalizing to non-stabilizer codes (e.g. based on non-Abelian groups) remains an open question.

\begin{acknowledgments}
We thank Katie Chang, Jeongwan Haah, Po-Shen Hsin, Chris Pattison, and Vedika Khemani for useful discussions. This work was
supported by the U.S. National Science Foundation (NSF) under Award No. PHY 2310614. T.-C.W. and S.Z. also acknowledge support from Stony Brook University's Center for Distributed Quantum Processing.
\end{acknowledgments}


\bibliography{references}

\onecolumngrid
\clearpage


\appendix

\section{Logicals in the CSS Coupled-Layer Code}
\label{sec:CSSlogical}

In this appendix, we count the number of logicals in the CSS coupled-layer code (whose corresponding chain complex is given in Fig.~\ref{fig:CoupLayerCSS}) under some modest assumptions. Specifically, we assume that

\begin{enumerate}
    \item \label{a3} $\delta_1'$ and $d_1'$ are surjective;
    \item \label{a4} $d_1'( \Ker d_1) = \Ker \epsilon_x'^T$ and $\delta_1'( \Ker \delta_1) = \Ker \epsilon_z'^T$.
\end{enumerate}
The second assumption says that every logical in the classical code with the parity check matrix $\epsilon_x'^T$ ($\epsilon_z'^T$) has a preimage in $B_1$ ($A_1$) which is a $Z$ ($X$) metacheck in CSS$_1$. In particular, this is satisfied when $\bm{\Gamma}$ corresponds to the gauging map as proved in \cite{Tantivasadakarn20}. On the other hand, the first assumption implies the following lemma, which will be used in later calculations
\begin{lemma}
If assumption \ref{a3} is satisfied, then $\Ker\epsilon_x \subseteq \Ker\epsilon_x'$ and $\Ker\epsilon_z \subseteq \Ker\epsilon_z'$    
\end{lemma}
\begin{proof}
    Take $\xi \in \Ker \epsilon_x$, then the operator $\epsilon_x(\xi)$ is the identity, which commutes with all $Z$ stabilizers of CSS$_1$. Therefore $\epsilon_x'(\xi)$ must commute with all operators $d_1'(B_1)$. But $d_1'$ is surjective, so $\epsilon_x'(\xi)$ can only be the identity, and $\xi \in \Ker \epsilon_x'$. Similarly, we can take $\chi \in \Ker\epsilon_z$ and argue that $\chi$ must belong to $\Ker\epsilon_z'$.
\end{proof}

\subsection{Classical Case} 

First we consider an easier case, where CSS$_1$ is a classical code in $Z$ basis, and CSS$_2$ is a classical code in $X$ basis. The stabilizer group is $\mathcal S = \langle \,\bm \alpha(e_x,a_2), \bm\zeta(b_1, q_2) \,\rangle$, and the corresponding chain complex is shown below 
\begin{equation*}
\begin{tikzpicture}[scale=0.9]
    \def\L{-4}
    \def\C{0}
    \def\R{4}
    \def\shadedwidth{60}
    \def\shadedheight{160}
    \def\H{1}
    \def\HH{2}
    \def\offs{2}
   
    \node[font=\Large, text=black] at (\C, -\H) {$\oplus$};

    \node[text=red!80!black] at ($(\L, -\H) +(0, 0)$) {$\mathcal{E}_{x}\otimes A_{2}$};

    \node (Q1Q2) at (\C,       0.0) {$Q_1 \otimes Q_{2}$};
    \node (B1A2) at (\C, -\HH) {$Q_{x}\otimes A_{2}$};
    
    \node[text=blue!80!black] at ($(\R, -\H) + (0, 0)$) {$B_{1}\otimes Q_{2}$};

    \draw[->] (-3, -0.8) -- (Q1Q2) node[midway, above] {$\epsilon_x \otimes \delta_2$};
    \draw[->] (-3, -1.2) -- (B1A2) node[midway, below] {$\epsilon_x' \otimes id$};
    \draw[->] (Q1Q2) -- (3, -0.8) node[midway, above] {$d_1^T \otimes id$};
    \draw[->] (B1A2) -- (3, -1.2) node[midway, below] {$d_1'^T \otimes \delta_2$};

\node[text=red!80!black] at ($(\L, -\H) + ( \offs, 0)$) {$\alpha(e_x, a_2)$};
  \node[text=blue!80!black] at ($(\R, -\H) - (\offs, 0)$) {$\zeta(b_1, q_2)$};
\end{tikzpicture}
\end{equation*}
Note this is a sub-diagram of Figure~\ref{fig:CoupLayerCSS}. To count the number of logical qubits, it suffices to find a maximal commuting set of logical operators. For each layer $a_2$, and each $v\in \Ker \epsilon_x'^T$, the operator $\mathcal{Z}(v,a_2)=  \prod_{q_x\in v} Z_{q_x, a_2}$ is a logical of $\mathcal S$. However, not all of these are independent. By assumption \ref{a4}, for each $v\in \Ker \epsilon_x'^T$, there is a $\Tilde{v} \in \Ker d_1$ such that $d_1'(\Tilde{v}) = v$. It follows that
\begin{align*}
    \prod_{b_1\in \Tilde{v}} \bm \zeta(b_1, q_2) = \prod_{a_2\ni q_2} \mathcal{Z}(v, a_2).
\end{align*}
The left side is a product of stabilizers, so the right side is not independent in the code space, and we should quotient out this subspace. The independent contributions are thus
\begin{align}\label{eq:z_in_a2}
    \Ker \epsilon_x'^T \otimes A_2/\Imaa \delta_2^T.
\end{align}
Similarly, for each layer $q_2$ and $v\in \Ker d_1^T$, one can write an $X$ operator $\mathcal{X}(v, q_2) = \prod_{q_1\in v} X_{q_1,q_2}$. They are also not independent. For each $\Tilde{v} \in \Ker \epsilon_x'$, $\epsilon_x(\Tilde{v})\in \Ker d_1^T$ because $\bm{\Gamma}_x$ preserves the algebra. We thus have
\begin{align*}
    \prod_{e_x\in \Tilde{v}} \bm \alpha(e_x,a_2) = \prod_{q_2\in a_2}\mathcal{X}(\epsilon_x(\Tilde{v}), q_2). 
\end{align*}
The independent contributions are thus
\begin{align}\label{eq:x_in_q2}
    (\Ker d_1^T\otimes Q_2 )\Big/( \epsilon_x(\Ker \epsilon_x')\otimes \Imaa \delta_2).
\end{align}
To show (\ref{eq:z_in_a2}) $\oplus$  (\ref{eq:x_in_q2}) is a maximal set of logical operators, we count the redundancies between stabilizers and show that
\begin{align}\label{eq:check_log}
    \rm  logicals = qubits - stabilizers + redundancies.
\end{align}
First, we find the redundancies between the $X$ stabilizers. Given $v \in \Ker \epsilon_x'$, we have
\begin{align*}
    \prod_{e_x\in v} \bm \alpha(e_x,a_2) = \prod_{q_2\in a_2}\left(\prod_{e_x\in v} \bm{e_x}^{(q_2)}\right).
\end{align*}
If $v \in \Ker \epsilon_x$, then the right-hand side is $1$. Otherwise, take $\theta\in \Ker \delta_2$, and
\begin{align*}
    \prod_{a_2\in \theta}\prod_{e_x\in v} \bm \alpha(e_x,a_2) = 1.
\end{align*}
These two cases are given by
\begin{align}\label{eq:red_in_alp_cond}
    \Ker \epsilon_x \otimes A_2
    \hspace{10pt}
    \text{and}
    \hspace{10pt}
    \Ker\epsilon_x'/\Ker\epsilon_x \otimes \Ker\delta_2.
\end{align}
Similarly, we can find the redundancies between the $Z$ stabilizers
\begin{align}\label{eq:red_in_beta}
    (\Ker d_1\cap \Ker d_1')\otimes Q_2
    \hspace{10pt}
    \text{and}
    \hspace{10pt}
    \frac{\Ker d_1}{\Ker d_1\cap \Ker d_1'} \otimes \Ker \delta_2^T.
\end{align}
The total redundancy between the stabilizers is (\ref{eq:red_in_alp_cond}) $\oplus$  (\ref{eq:red_in_beta}).

The number of qubits and stabilizers is easy to count. The qubits are $Q_1\otimes Q_2 \oplus Q_x\otimes A_2$. Then the stabilizers are $\mathcal{E}_x\otimes A_2 \oplus B_1\otimes Q_2$. Now we can verify (\ref{eq:check_log}). For notational simplicity, below, when we write the vector space, we actually mean its dimension.  
\begin{align*}
    &
    \Ker \epsilon_x'^T ( A_2 - \Imaa \delta_2^T) + \Ker d_1^T Q_2 -  (\epsilon_x(\Ker \epsilon_x')) \Imaa \delta_2\\
    \overset{?}{=}&Q_1Q_2 + Q_xA_2 - \mathcal{E}_xA_2 - B_1Q_2 + 
    \Ker \epsilon_x A_2
    +(\Ker\epsilon_x'-\Ker\epsilon_x ) \Ker\delta_2 
    +(\Ker d_1\cap \Ker d_1') Q_2
    + (\Ker d_1-\Ker d_1\cap \Ker d_1') \Ker \delta_2^T\\
    =& Q_1Q_2 + Q_xA_2 - \mathcal{E}_xA_2 - B_1Q_2 + 
    \Ker \epsilon_x \Imaa \delta_2
    +\Ker\epsilon_x'\Ker\delta_2 
    +(\Ker d_1\cap \Ker d_1') \Imaa \delta_2
    + \Ker d_1\Ker \delta_2^T.
\end{align*}
The equality labeled by $\overset{?}{=}$ is what we need to verify. From second to third line we used $A_2 - \Ker\delta_2 = \Imaa \delta_2$ and $Q_2-\Ker\delta_2^T = \Imaa \delta_2$. Now use $\Ker \epsilon_x'^T = Q_x - \mathcal{E}_x + \Ker \epsilon_x'$ and $\Ker d_1^T = Q_1 - B_1 + \Ker d_1$, we reduce the unproven equality to
\begin{align*}
    &\Ker \epsilon_x' A_2 - \Ker\epsilon_x'^T\Imaa \delta_2^T + \Ker d_1 Q_2 -  (\epsilon_x(\Ker \epsilon_x')) \Imaa \delta_2\\
    \overset{?}{=}&
    \Ker \epsilon_x \Imaa \delta_2
    +\Ker\epsilon_x'\Ker\delta_2 
    +(\Ker d_1\cap \Ker d_1') \Imaa \delta_2
    + \Ker d_1\Ker \delta_2^T.
\end{align*}
Replace $\Ker\epsilon_x'^T$ by $d_1'(\Ker d_1)$ in the first line by assumption \ref{a4}. As well as $\epsilon_x(\Ker \epsilon_x') + \Ker \epsilon_x = \Ker\epsilon_x'$ and $d_1'(\Ker d_1) + (\Ker d_1\cap \Ker d_1') = \Ker d_1$, we reduce the unproven equality to
\begin{align*}
    \Ker \epsilon_x' A_2  + \Ker d_1 Q_2 
    \overset{?}{=}
    \Ker \epsilon_x' \Imaa \delta_2
    +\Ker\epsilon_x'\Ker\delta_2 
    +\Ker d_1 \Imaa \delta_2+ \Ker d_1\Ker \delta_2^T.
\end{align*}
This equality indeed holds because $\Imaa \delta_2 + \Ker\delta_2 = A_2$ and $\Imaa \delta_2 + \Ker\delta_2^T = Q_2$, and (\ref{eq:check_log}) is verified.

In order to give an example, we count the logicals of the X-cube model. One can check the condensation data for X-cube satisfies assumptions \ref{a3} and \ref{a4}. Assuming periodic boundary condition, $\Ker \epsilon_x'^T = L_x+L_y-1$, $A_2 = L_z$, $\Imaa \delta_2^T = L_z-1$, hence (\ref{eq:z_in_a2}) gives $L_x+L_y-1$. $\Ker d_1^T = L_x+L_y$, $Q_2 = L_z$, $\epsilon_x(\Ker\epsilon_x') = L_x + L_y - 2$, $\Imaa\delta_2 = L_z-1$, hence (\ref{eq:x_in_q2}) gives $L_x + L_y + 2(L_z-1)$. In total, the dimension of (\ref{eq:z_in_a2}) $\oplus$  (\ref{eq:x_in_q2}) gives $2(L_x+L_y+L_z) - 3$.

\subsection{Quantum Case}

Next, we consider the case where both CSS$_1$ and CSS$_2$ are quantum codes. The stabilizer group is $\mathcal S = \langle \,\bm\alpha(e_x,a_2) , \bm \beta(e_z, b_2) , \bm\xi(a_1,q_2), \bm\zeta(b_1, q_2) \,\rangle$, and the chain complex is given in Figure~\ref{fig:CoupLayerCSS}. First we find a maximal set of commuting logicals. We make this set roughly symmetric in the $X$ and $Z$ logicals. The $Z$ operators in (\ref{eq:z_in_a2}) are still logicals. Similarly, there is a set of $X$-logicals  in layers of $b_2$
\begin{align}\label{eq:x_in_b2}
    \Ker \epsilon_z'^T \otimes B_2/\Imaa d_2^T.
\end{align}
The $X$ operators in (\ref{eq:x_in_q2}) are no longer logicals since they do not necessarily commute with $Z$-stabilizers $\bm \beta(e_z,b_2)$. Instead, given $v\in \Ker d_1^T$ and $w \in \Ker d_2^T$, the operator $\mathcal{X}(v,w) = \prod_{q_1\in v} \prod_{q_2\in w} X_{q_1,q_2}$ can be verified to be a logical. Indeed, $\mathcal{X}(v,w)$ commutes with $\bm\zeta(b_1,q_2)$ because $\prod_{q_1\in v}X_{q_1,q_2}$ always commutes with $\bm{b_1}^{(q_2)}$, and it commutes with $\bm \beta(e_z,b_2)$ because $\prod_{q_2\in w} X_{q_1,q_2}$ always commutes with $\prod_{q_2\in b_2}\bm{e_z}^{(q_2)}$. Again, a subspace of these operators is in fact a product of stabilizers. These are
\begin{align*}
    &\tilde v \in \Ker \epsilon_x',
    \hspace{5pt}
    a_2 \in A_2
    :
    \hspace{35pt}
    \prod_{e_x\in \tilde v}\bm \alpha(e_x,a_2) = \mathcal{X}(\epsilon_x(\tilde v ),\delta_2(a_2)),\\
    &a_1 \in A_1,
    \hspace{5pt}
    \tilde w \in \Ker d_2^T
    :
    \hspace{35pt}
    \prod_{q_2\in \tilde w}\bm\xi(a_1,q_2) = \mathcal{X}(\delta_1(a_1),\tilde w).
\end{align*}
Therefore, the subspace of $\Ker d_1^T \otimes \Ker d_2^T$ which intersects the space of stabilizers is
\begin{align*}
    \epsilon_x(\Ker\epsilon_x') \otimes \Imaa \delta_2 + \Imaa \delta_1 \otimes \Ker d_2^T,
\end{align*}
where the sum is vector space sum, not a direct sum, since these two spaces might have non-zero overlap. This subspace needs to be quotient out, and the set of $X$-logicals are
\begin{align}\label{eq:x_in_Q_quant}
    \frac{\Ker d_1^T \otimes \Ker d_2^T}{\epsilon_x(\Ker\epsilon_x') \otimes \Imaa \delta_2 + \Imaa \delta_1 \otimes \Ker d_2^T}.
\end{align}
However, this is not a maximal set of logicals in layers $q_2$. In addition to the $X$-logicals in (\ref{eq:x_in_Q_quant}), there can also be a similar set of $Z$-logicals. We may try to write down the same operators for $Z$-logicals, take $v\in \Ker \delta_1^T$ and $w \in \Ker \delta_2^T$, we have $\mathcal{Z}(v,w) = \prod_{q_1\in v}\prod_{q_2\in w}Z_{q_1,q_2}$. However, such a $Z$ operator might not commute with all $X$-logicals in (\ref{eq:x_in_Q_quant}). For instance, assume $v$ and $w$ represent nontrivial $Z$ logicals in CSS$_1$ and CSS$_2$ respectively, and let $\tilde v$ and $\tilde w$ be their conjugate $X$-logicals, so that sgn$\big(\prod_{q_1\in v}Z_{q_1}, \prod_{q_1\in\tilde v}X_{q_1}\big) = -1$ and sgn$(\prod_{q_2\in w}Z_{q_2}, \prod_{q_2\in\tilde w}X_{q_2}) = -1$, then sgn$(\mathcal{Z}(v,w),\mathcal{X}(\tilde v,\tilde w)) = -1$. Hence to guarantee commutativity, we take $w \in \Imaa d_2$, then one can check $\mathcal{Z}(v, w)$ commutes with all $X$-stabilizers as well as $X$-logicals in (\ref{eq:x_in_Q_quant}). There is also a subspace of $\Ker\delta_1^T \otimes \Imaa d_2$ which is a product of stabilizers, these are
\begin{align*}
    &\tilde v \in \Ker \epsilon_z',
    \hspace{5pt}
    b_2\in B_2:
    \hspace{35pt}
    \prod_{e_z \in \tilde v} \bm \beta(e_z, b_2) = \mathcal{Z}(\epsilon_{z}(\tilde v),d_2(b_2)),\\
    & b_1\in B_1,
    \hspace{5pt}
    b_2 \in B_2:
    \hspace{35pt}
    \prod_{q_2\in b_2}\bm\zeta(b_1,q_2) = \mathcal{Z}(d_1(b_1),d_2(b_2)).
\end{align*}
Quotienting out this subspace, the set of $Z$-logicals are
\begin{align}\label{eq:z_in_Q_quant}
    \frac{\Ker\delta_1^T}{\epsilon_z(\Ker\epsilon_z') + \Imaa d_1} \otimes \Imaa d_2.
\end{align}
We claim that a maximal set of logical operators in $H$ is (\ref{eq:z_in_a2}) $\oplus$ (\ref{eq:x_in_b2}) $\oplus$  (\ref{eq:x_in_Q_quant}) $\oplus$  (\ref{eq:z_in_Q_quant}), which needs to be checked against (\ref{eq:check_log}).

For example, consider the usual tensor product. In this case $\epsilon_z' = \delta_1^T$, $\epsilon_x' = d_1^T$, so $\epsilon_z(\Ker\epsilon_z') = \Ker\delta_1^T$ and $\epsilon_x(\Ker\epsilon_x') = \Ker d_1^T$. (\ref{eq:z_in_Q_quant}) is trivial. The dimensions of (\ref{eq:z_in_a2}), (\ref{eq:x_in_b2}) and (\ref{eq:x_in_Q_quant}) are, respectively
\begin{align*}
    &(\ref{eq:z_in_a2}):\hspace{15pt}
    \Ker d_1 \Ker \delta_2 = H^{1}(\text{CSS}_1)H^{-1}(\text{CSS}_2),\\
    &(\ref{eq:x_in_b2}):\hspace{15pt}
    \Ker \delta_1 \Ker d_2 = H^{-1}(\text{CSS}_1)H^{1}(\text{CSS}_2),\\
    &(\ref{eq:x_in_Q_quant}):\hspace{15pt}
    \Ker d_1^T  \Ker d_2^T-(\Ker d_1^T  \Imaa \delta_2 + \Imaa \delta_1 \Ker d_2^T - \Imaa \delta_1 \Imaa \delta_2) = H^{0}(\text{CSS}_1)H^{0}(\text{CSS}_2).
\end{align*}
Indeed this agrees with the dimension of the tensor product cohomology $H^0(\text{CSS}_1\otimes \text{CSS}_2)$.

In order to check (\ref{eq:check_log}), we also need to count the redundancies between the stabilizers. (\ref{eq:red_in_alp_cond}) and (\ref{eq:red_in_beta}) are still present. Recall that these are redundancies between $\bm\alpha(e_x,a_2)$ and $\bm\zeta(b_1,q_2)$, respectively. Exactly the same redundancies hold in replacing $\bm \alpha(e_x,a_2)$ with $\bm \beta(e_z,b_2)$, and $\bm\zeta(b_1,q_2)$ with $\bm\xi(a_1,q_2)$. We simply exchange $\delta\leftrightarrow d$ and $\epsilon_x \leftrightarrow \epsilon_z$,
\begin{align}
    &\bm\beta(e_z,b_2): \hspace{25pt}
    \Ker \epsilon_z \otimes B_2
    \hspace{10pt}
    \text{and}
    \hspace{10pt}
    \Ker\epsilon_z'/\Ker\epsilon_z \otimes \Ker d_2,
    \label{eq:red_in_beta_cond}\\
    &\bm\xi(a_1,q_2): 
    \hspace{25pt}
    \Ker \delta_1\cap \Ker \delta_1'\otimes Q_2
    \hspace{10pt}
    \text{and}
    \hspace{10pt}
    \frac{\Ker \delta_1}{\Ker \delta_1\cap \Ker \delta_1'} \otimes \Ker d_2^T.
    \label{eq:red_in_alp}
\end{align}
In addition, there can also be redundancies involving both $\bm\alpha(e_x,a_2)$ and $\bm\xi(a_1,q_2)$. Pick $v\in \Ker\epsilon_x'$ such that $\epsilon_x(v) \in \Imaa \delta_1$. Let $\tilde v \in A_1$, $\delta_1(\tilde v) = \epsilon_x(v)$. Then for each $a_2\in A_2$,
\begin{align*}
    \left(\prod_{e_x\in v} \bm \alpha(e_x, a_2) \right)
    \left(\prod_{a_1\in\tilde v}\prod_{q_2\in a_2}\bm\xi(a_1,q_2)\right) = 1.
\end{align*}
Hence, the space of redundancy is
\begin{align}\label{eq:red_in_alpalp}
    (\epsilon_x(\Ker\epsilon_x')\cap \Imaa \delta_1) \otimes \Imaa \delta_2.
\end{align}
In exactly the same way, the redundancies involving both $\bm\beta(e_z,b_2)$ and $\bm\zeta(b_1,q_2)$ are obtained by exchanging $\delta\leftrightarrow d$ and $\epsilon_x \leftrightarrow \epsilon_z$,
\begin{align}\label{eq:red_in_betabeta}
    (\epsilon_z(\Ker\epsilon_z')\cap \Imaa d_1) \otimes \Imaa d_2.
\end{align}
We claim that the total space of redundancies is (\ref{eq:red_in_alp_cond}) $\oplus$  (\ref{eq:red_in_beta}) $\oplus$ 
(\ref{eq:red_in_beta_cond}) $\oplus$ (\ref{eq:red_in_alp}) $\oplus$ 
(\ref{eq:red_in_alpalp}) $\oplus$  (\ref{eq:red_in_betabeta}).

Now we check (\ref{eq:check_log}). To avoid clustering of notation, we again write each vector-space symbol for its dimension. The number of qubits is $Q_1Q_2+Q_xA_2+Q_zB_2$. The number of stabilizers is $\mathcal{E}_xA_2 + B_1Q_2 + \mathcal{E}_zB_2 + A_1Q_2$. We have already shown in the classical case that
\begin{align}\label{eq:check_log_x_cond}
    Q_1Q_2+Q_xA_2 - \mathcal{E}_xA_2 - B_1Q_2 + (\ref{eq:red_in_alp_cond}) + (\ref{eq:red_in_beta}) = (\ref{eq:z_in_a2}) + (\ref{eq:x_in_q2}).
\end{align}
Exchanging the role of $X$ and $Z$, we also have
\begin{align}\label{eq:check_log_z_cond}
    Q_1Q_2+Q_zB_2 - \mathcal{E}_zB_2 - A_1Q_2 + (\ref{eq:red_in_beta_cond}) + (\ref{eq:red_in_alp}) = (\ref{eq:x_in_b2}) + \Ker \delta_1^T Q_2 -  (\epsilon_z(\Ker \epsilon_z')) \Imaa d_2,
\end{align}
where the last two terms on the right-hand side are the analogue of (\ref{eq:x_in_q2}). Subtract these two equalities from the (\ref{eq:check_log}), then (\ref{eq:check_log}) $-$ (\ref{eq:check_log_x_cond}) $-$ (\ref{eq:check_log_z_cond}) gives
\begin{align*}
    -Q_1Q_2 + (\ref{eq:red_in_alpalp}) +  (\ref{eq:red_in_betabeta}) 
    \overset{?}{=}    (\ref{eq:x_in_Q_quant}) +  (\ref{eq:z_in_Q_quant}) - (\ref{eq:x_in_q2}) - (\Ker \delta_1^T Q_2 -  (\epsilon_z(\Ker \epsilon_z')) \Imaa d_2).
\end{align*}
Expand the right-hand side,
\begin{align*}
    (\ref{eq:x_in_Q_quant}) &= 
    \Ker d_1^T \Ker d_2^T - (\epsilon_x(\Ker\epsilon_x'))  \Imaa \delta_2 - \Imaa \delta_1  \Ker d_2^T + (\ref{eq:red_in_alpalp})\\
    &= (\ref{eq:x_in_q2}) - \Ker d_1^T \Imaa d_2 - \Imaa \delta_1  \Ker d_2^T + (\ref{eq:red_in_alpalp}),\\
    (\ref{eq:z_in_Q_quant}) &= 
    \Ker\delta_1^T \Imaa d_2 -(\epsilon_z(\Ker\epsilon_z')) \Imaa d_2 - \Imaa d_1  \Imaa d_2 + (\ref{eq:red_in_betabeta}) \\
    &= (\Ker\delta_1^T  Q_2 -(\epsilon_z(\Ker\epsilon_z')) \Imaa d_2) - \Ker\delta_1^T \Ker d_2^T - \Imaa d_1  \Imaa d_2 + (\ref{eq:red_in_betabeta}).
\end{align*}
Hence, we are only left with
\begin{align*}
    Q_1Q_2
    \overset{?}{=}  
     \Ker d_1^T \Imaa d_2 + \Imaa \delta_1  \Ker d_2^T + \Ker\delta_1^T \Ker d_2^T + \Imaa d_1  \Imaa d_2.
\end{align*}
Indeed, this holds because
\begin{align*}
    \text{right-hand side} &= (\Ker d_1^T  + \Imaa d_1 ) \Imaa d_2  + (\Imaa \delta_1   + \Ker\delta_1^T )\Ker d_2^T  \\
    &= Q_1 \Imaa d_2 + Q_1 \Ker d_2^T\\
    &= Q_1Q_2.
\end{align*}
This concludes the counting of logicals in the CSS version. Note that in choosing a maximal set of commuting logicals, we take half of these to be $X$-logicals and another half to be $Z$-logicals, in a symmetric manner. However, it is not hard to find a maximal set of pure $X$-logicals, or pure $Z$-logicals.

\section{Understanding $\ee-\mm$ from Symplectic Doubling}
\label{sec:sd}

In this appendix, we explore the interaction of the $\ee-\mm$ quotient code studied in Sec.~\ref{sec:gqc} and the $\ee-\mm$ balanced product in Sec.~\ref{sec:embp} with an operation called symplectic doubling \cite{KovalevPryadko13,liu2024subsystem}. Symplectic doubling is a mapping between stabilizer codes, which maps an arbitrary Pauli stabilizer code to a CSS code with double the number of qubits. This potentially allows us to apply the chain complexes formalism to non-CSS codes by factoring through its symplectic double. Since $\ee-\mm$ related operations in general turn CSS codes into non-CSS codes, it is insightful to apply the symplectic doubling on both sides, and ask how do these operations manifest on the chain complex level. Indeed, we will show if a stabilizer code admits an $\ee-\mm$ symmetry, then the doubled code has an induced symmetry which acts by simply permuting qubits without the original unitary. It follows that the $\ee-\mm$ quotient and balanced product are closely related to the quotient and balanced product of this induced action.

First we review the symplectic doubling. Given a stabilizer code $\mathcal S = \langle \bm s \rangle$ with qubit vector space $Q$. Up to a phase, the stabilizers $\bm s$ are given by 
\begin{align*}
    \bm s = \prod_{q\in s_X} X_q \prod_{q\in s_Z} Z_q,
\end{align*}
where $s_X$ ($s_Z$) is the support of $\bm s$ on which $X$ ($Z$) acts. Then $\mathcal S$ can be mapped to a CSS code with double the number of qubits as follows. Let $q^{(0)}$ and $q^{(1)}$ denote the two copies of qubit $q$ in $H$, with the superscipt a $\mathbb{Z}_2$ label, the stabilizers of the doubled CSS code are given by
\begin{align*}
    &\omega_X(\bm s) = \prod_{q\in s_X} X_{q^{(0)}} \prod_{q\in s_Z} X_{q^{(1)}},\\
    &\omega_Z(\bm s) = \prod_{q\in s_X} Z_{q^{(1)}} \prod_{q\in s_Z} Z_{q^{(0)}}.
\end{align*}
Then the stabilizer code is mapped to a CSS code, whose stabilizer group is
\begin{align*}
    \omega(\mathcal{S}) := 
    \langle\,
    \omega_X(\bm s),\,
    \omega_Z(\bm s)
    \,|\,
    \bm s\in \mathcal{S}
    \,\rangle.
\end{align*}
In the rest of this appendix, we will show the interplay between $\omega$ and $\ee-\mm$ quotient, as well as $\ee-\mm$ balanced product.

\subsection{$\ee-\mm$ Quotient from Symplectic Doubling}
\label{sec:gqc_sdq}

In this appendix, we show how symplectic doubling is related to the $\ee-\mm$ quotient code in Sec.~\ref{sec:gqc}. Specifically, we will show that the symplectic double code naturally inherits a symmetry related to the $\ee$-$\mm$ quotient. In particular, there exists the following commuting diagram

\begin{align}\label{eqn:quo_doub_comm_diag}
\begin{adjustbox}{valign=c}
    \begin{tikzpicture}[scale=0.8, every node/.style={scale=1}]
    \def\L{-3.5}
    \def\C{0}
    \def\R{5.5}
    \def\shadedwidth{60}
    \def\shadedheight{160}
    \def\H{3.5}
    \def\shift{0.3}
    %
    \node (1) at (0 , 0) {$\omega(\mathcal{S})$};
    \node  (2) at (6.5, 0) {$\omega(\mathcal{S})/G = \omega(\mathcal{S}/G)$};
%
    \node (3) at (0, \H) {$\mathcal{S}$};
    \node  (4) at (\R, \H) {$\mathcal{S}/G$};
%
%
%
    \draw[->] (1) -- (2); 
    \draw[->] (3) -- (4);
%
    \draw[->] (3) -- (1); 
    \draw[->] (4) -- (5.5,0.4); 
    \node at ($(1)!0.5!(5.5, 0) - (0,\shift)$) {  quotient};
    \node at ($(3)!0.5!(4) + (0,\shift)$) {$\ee-\mm$ quotient};
    \node at ($(3)!0.5!(1) - (\shift, 0)$) {$\omega$};
    \node at ($(4)!0.5!(5.5, 0) + (\shift, 0)$) {$\omega$};
    \node at ($(1)!0.5!(4)$) {$\circlearrowleft$};
\end{tikzpicture}
\end{adjustbox}
\end{align}
where $\omega$ corresponds to the symplectic doubling.

To begin with, consider a stabilizer code $\mathcal{S} = \langle \, \bm s\,\rangle$ with an $\ee-\mm$ symmetry given by group $G$, then $\omega(\mathcal{S})$ inherits an symmetry of $G$ by pure permutation as follows. Given $g\in G$, define the homomorphism $\gamma: G \ra \mathbb{Z}_2$ by
\begin{align*}
    \gamma(g) = 
    \begin{cases}
        0, \, U_g = id,\\
        1, \, U_g = \mathfrak{h}.
    \end{cases}
\end{align*}
In particular, we have $U_g = \mathfrak{h}^{\gamma(g)}$. Consider the action
\begin{align*}
    q^{(a)} \cdot g := (q\cdot g)^{(a + \gamma(g))},
\end{align*}
where $q\cdot g$ is the original permutation on $Q$, and $a + \gamma(g)$ is addition in $\mathbb{Z}_2$. To see this is a symmetry, consider
\begin{align*}
    \omega_X(\bm s) \cdot g = \prod_{q\in s_X} X_{(q\cdot g)^{(0+\gamma(g))}} \prod_{q\in s_Z} X_{(q\cdot g)^{(1 + \gamma(g))}}.
\end{align*}
There are two cases
\begin{itemize}
    \item $\gamma(g) = 0$, then $(s\cdot g)_X = s_X\cdot g$ and $(s\cdot g)_Z = s_Z\cdot g$,
    \begin{align*}
        \omega_X(\bm s) \cdot g = \prod_{q\in (s\cdot g)_X} X_{q^{(0)}} \prod_{q\in (s\cdot g)_Z} X_{q^{(1)}}.
    \end{align*}
    \item $\gamma(g) = 1$, then $(s\cdot g)_X = s_Z\cdot g$ and $(s\cdot g)_Z = s_X\cdot g$,
    \begin{align*}
        \omega_X(\bm s) \cdot g = \prod_{q\in (s\cdot g)_Z} X_{q^{(1)}} \prod_{q\in (s\cdot g)_X} X_{q^{(0)}}.
    \end{align*}
\end{itemize}
In either case, we see that
\begin{align}\label{eq:doubling_action}
    \omega_X(\bm s) \cdot g = \omega_X(\bm s \cdot g),
\end{align}
and similarly for $\omega_Z(\bm s)$.

Given the above induced $G$ action on the CSS code $\omega(\mathcal{S})$ by permutation, we can consider the quotient of $\omega(\mathcal{S})$ at the chain complex level. It is natural to believe this quotient is related to the original $\ee-\mm$ balanced product. We make this relation precise below.

First, consider the quotient of the induced action on $\omega(\mathcal{S})$. The qubit representatives can be chosen to be $\tilde q^{(0)}$ and $\tilde q^{(1)}$, where $\tilde q$ is a representative of $Q/G$. Similarly, due to Eq.~\eqref{eq:doubling_action}, the stabilizer representatives can be chosen to be $\omega_X(\tilde{\bm s})$ and $\omega_Z(\tilde{\bm s})$, where $\tilde{\bm s}$ is a representative of $\mathcal{S}/G$. With the notation in Sec.~\ref{sec:gqc}, the stabilizers of the quotient CSS code are
\begin{align*}
    &\omega_X(\tilde{\bm s})_G = \prod_{q\in \tilde s_X} X_{\tilde q^{(0 + \gamma(g_q))}} \prod_{q\in \tilde s_Z} X_{\tilde q^{(1 + \gamma(g_q))}},\\
    &\omega_Z(\tilde{\bm s})_G = \prod_{ q \in \tilde s_X} Z_{\tilde q^{(1+\gamma(g_q))}} \prod_{q\in \tilde s_Z} Z_{\tilde q^{(0+\gamma(g_q))}}.
\end{align*}
Hence, we have
\begin{align*}
    \omega(\mathcal{S})/G 
    \,=\,
    \langle\, \omega_X(\tilde{\bm s})_G,\, \omega_Z(\tilde{\bm s})_G
    \,|\,
    \tilde{\bm s}\in \mathcal{S}/G
    \,\rangle.
\end{align*}

On the other hand, recall the stabilizers of the $\ee-\mm$ quotient of $H$ is given by Eq.~\eqref{eq:em_quotient}. Specializing to this case, we have
\begin{align*}
    & \mathcal{S}/G = \langle\,
    \tilde{\bm s}_G
    \,|\,
    \tilde{\bm s}\in \mathcal{S}/G
    \,\rangle,\\
    &\tilde{\bm s}_G = \prod_{q\in \tilde s_X} \mathfrak{h}^{\gamma(g_q)} X_{\tilde q} \mathfrak{h}^{\gamma(g_q)}
    \prod_{q\in \tilde s_Z} \mathfrak{h}^{\gamma(g_q)} Z_{\tilde q} \mathfrak{h}^{\gamma(g_q)}.
\end{align*}
We can consider the doubling of $\mathcal{S}/G$, and we see that it is equal to the quotient of $\omega(\mathcal{S})$,
\begin{align*}
    &\omega_X(\tilde{\bm s}_G) = \prod_{q\in \tilde s_X} X_{\tilde q^{(0 + \gamma(g_q))}} \prod_{q\in \tilde s_Z} X_{\tilde q^{(1 + \gamma(g_q))}} = \omega_X(\tilde{\bm s})_G,\\
    &\omega_Z(\tilde{\bm s}_G) = \prod_{ q \in \tilde s_X} Z_{\tilde q^{(1+\gamma(g_q))}} \prod_{q\in \tilde s_Z} Z_{\tilde q^{(0+\gamma(g_q))}} = \omega_Z(\tilde{\bm s})_G.
\end{align*}
That is, 
\begin{align*}
    \omega(\mathcal{S}/G) = \omega(\mathcal{S})/G.
\end{align*}

To summarize, we see that given a stabilizer code $\mathcal{S}$ with an $\ee-\mm$ symmetry, the doubling translates the $\ee-\mm$ action to a pure permutation. In addition, The CSS code obtained by doubling the quotient of $\mathcal{S}$ by the $\ee-\mm$ symmetry is equivalent to quotient the doubling of $\mathcal{S}$, as in Eq.~\eqref{eqn:quo_doub_comm_diag}.

We again use the TC as an example to illustrate the above commutative diagram. First going along the top path in Eq.~\eqref{eqn:quo_doub_comm_diag}, we have seen this gives the repetition code in $Y$ basis
\begin{align*}
    \text{TC}/\mathbb{Z} = \langle Y_i Y_{i+1} \rangle
\end{align*}
The symplectic doubling gives two qubits per site, with stabilizers
\begin{align*}
    \omega(\text{TC} / \mathbb{Z}) =  \langle X_{i}^{(0)}X_{i}^{(1)} X_{i+1}^{(0)} X_{i+1}^{(1)} ,
 Z_{i}^{(0)}Z_{i}^{(1)} Z_{i+1}^{(0)} Z_{i+1}^{(1)}\rangle.
\end{align*}
Next going along the bottom path in Eq.~\eqref{eqn:quo_doub_comm_diag}, the doubling of TC puts two qubits per edge with stabilizers
\begin{align*}
    & \omega(\text{TC}) = 
    \langle\,
     \bm a_v,\, \bm b_v,\, 
     \bm a_p,\, \bm b_p
     \,\rangle\\
    & \bm a_v = \prod_{e \in v} X_{e}^{(0)},
    \hspace{35pt}
    \bm a_p = \prod_{e \in p} X_{e}^{(1)},\\
    & \bm b_v = \prod_{e \in v} Z_{e}^{(1)},
    \hspace{35pt}
    \bm b_p = \prod_{e \in p} Z_{e}^{(0)}.
\end{align*}
The induced $\mathbb{Z}$ action is a shift on lattice, followed by a shift in the superscript $\mathbb{Z}_2$ label
\begin{align*}
    \begin{tikzpicture}
        [scale = 0.8, baseline=0.7cm]
            \node [circle,fill, inner sep=1.5pt] at (1,0) {};
            \node [circle,fill, inner sep=1.5pt]at (0, 1) {};
            \node [circle,fill, inner sep=1.5pt]at (2, 1) {};
            \node [circle,fill, inner sep=1.5pt]at (1, 2) {};
            \node at (1, -0.4) {$(1)\,/\,(0)$};
            \node at (-1, 1) {$(0)\,/\,(1)$};
            \node at (1, 2.4) {$(0)\,/\,(1)$};
            \node at (2.9, 1) {$(1)\,/\,(0)$};
            \draw[step = 2] (-0.5,-0.5) grid (2.5,2.5);
            \draw[->] (1.1,1.9) -- (1.9,1.1);
            \draw[->] (0.1,0.9) -- (0.9,0.1);
        \end{tikzpicture}
\end{align*}
This action clearly switches between $\bm a_v  \leftrightarrow \bm a_p$ and $\bm b_v  \leftrightarrow \bm b_p$, and quotient by this recovers $\omega(\text{TC}/\mathbb{Z})$, hence we indeed have
\begin{align*}
    \omega(\text{TC}/\mathbb{Z}) = \omega(\text{TC})/\mathbb{Z}.
\end{align*}

\subsection{$\ee-\mm$ Balancing from Symplectic Doubling}

In this appendix, we discuss how the symplectic doubling is related to the $\ee-\mm$ balanced product construction in Sec.~\ref{sec:embp}. We have just seen that giving a stabilizer code with $\ee-\mm$ symmetry, the doubled code has an induced group action given by simply permuting the qubits without the original Hadamards. As a result, given a pair of CSS codes with $\ee-\mm$ symmetry, we can first perfrom symplectic doubling on both of them, and apply the traditional balanced product \cite{breuckmann2021balanced} using the induced actions. It is natural to expect a relation between this balanced product code and the $\ee-\mm$ balanced product we have defined. We will make this relation precise below.

Given a CSS code $A\ra Q\ra B$ , its symplectic double is
\begin{align}
    \omega(\text{CSS}) = \text{CSS} \oplus \text{CSS}^\vee
\end{align}
That is, it consists of the original CSS code on qubits $q^{(0)}$ with $X$-stabilizers $\omega_X(\bm a)$ and $Z$-stabilizers $\omega_Z(\bm b)$, as well as the Hadamarded code CSS$^\vee$ on qubits $q^{(1)}$ with $Z$-stabilizers $\omega_Z(\bm a)$ and $X$-stabilizers $\omega_X(\bm b)$. Now consider the tensor product of two symplectic doubles. This is
\begin{align}
    \label{eqn:double_prod}
    \omega(\text{CSS}_1) \otimes \omega(\text{CSS}_2)
    &=
    (\text{CSS}_1 \otimes \text{CSS}_2) 
    \oplus 
    (\text{CSS}_1^\vee \otimes \text{CSS}_2^\vee) \nonumber\\
    &\hspace{-1cm} 
    \oplus (\text{CSS}_1^\vee \otimes \text{CSS}_2) 
    \oplus 
    (\text{CSS}_1 \otimes \text{CSS}_2^\vee).
\end{align}
Note the first line corresponds to the double of CSS$_1 \otimes$ CSS$_2$ with the induced action. We can make this correspondence explicit by the following chain isomorphism
\begin{align*}
    \phi: (\text{CSS}_1 \otimes \text{CSS}_2) 
    \oplus 
    (\text{CSS}_1^\vee \otimes \text{CSS}_2^\vee) 
    \rightarrow
    \omega(\text{CSS}_1 \otimes \text{CSS}_2).
\end{align*}
Recall that we place qubits, $X$-stabilizers, $Z$-stabilizers on degree $0$, $-1$ and $1$ respectively, the linear map on qubits is thus
\begin{align*}
    \phi^0 : \,
    \begin{cases}
    q_1^{(i)} \otimes q_2^{(i)} 
    \,\mapsto \,
    (q_1\otimes q_2)^{(i)},
    \hspace{15pt}
    i = 0,1\\
    \omega_1(a_1) \otimes \omega_2(b_2) 
    \,\mapsto \,
    (a_1\otimes b_2)^{(0)}\\
    \omega_2(b_1) \otimes \omega_1(a_2) 
    \,\mapsto \,
    (b_1\otimes a_2)^{(0)}\\
    \omega_1(b_1) \otimes \omega_2(a_2) 
    \,\mapsto \,
    (b_1\otimes a_2)^{(1)}\\
    \omega_2(a_1) \otimes \omega_1(b_2) 
    \,\mapsto \,
    (a_1\otimes b_2)^{(1)}
    \end{cases}
\end{align*}
on $X$-stabilizers
\begin{align*}
    \phi^{-1}:
    \,\begin{cases}
    q_1^{(0)} \otimes \omega_1(a_2) \mapsto \omega_1(q_1\otimes a_2)\\
    q_1^{(1)} \otimes \omega_1(b_2) \mapsto \omega_1(q_1\otimes b_2)\\
    \omega_1(a_1) \otimes q_2^{(0)} \mapsto \omega_1(a_1\otimes q_2)\\
    \omega_1(b_1) \otimes q_2^{(1)} \mapsto \omega_1(b_1\otimes q_2)
    \end{cases}
\end{align*}
on $Z$-stabilizers
\begin{align*}
    \phi^1:
    \,\begin{cases}
    q_1^{(0)} \otimes \omega_2(b_2) \mapsto \omega_2(q_1\otimes b_2)\\
    q_1^{(1)} \otimes \omega_2(a_2) \mapsto \omega_2(q_1\otimes a_2)\\
    \omega_2(b_1) \otimes q_2^{(0)} \mapsto \omega_2(b_1\otimes q_2)\\
    \omega_2(a_1) \otimes q_2^{(1)} \mapsto \omega_2(a_1\otimes q_2)
    \end{cases}
\end{align*}
 Each $\phi$ clearly extends to a linear isomorphism. There are two things remain to check, first $\phi$ is indeed a chain map; second $\phi$ commutes with the group action. To see $\phi$ commutes with the boundary maps, take the $X$-stabilizer $q_1^{(0)} \otimes \omega_1(a_2)$ as an example, in additive notation,
\begin{align*}
    q_1^{(0)} \otimes \omega_1(a_2) &\mapsto \sum_{b_1\ni q_1} \omega_2(b_1) \otimes \omega_1(a_2) + \sum_{q_2\in a_2} q_1^{(0)} \otimes q_2^{(0)}\\
    & \overset{\phi^0}{\mapsto}
    \sum_{b_1\ni q_1} (b_1\otimes a_2)^{(0)} + \sum_{q_2\in a_2} (q_1\otimes q_2)^{(0)}\\
    & = \omega_1(q_1\otimes a_2)\\
    & = \phi^{-1}(q_1^{(0)} \otimes \omega_1(a_2)).
\end{align*}
Similarly for the other stabilizers. Finally, we need to show $\phi$ commutes with the group action. Take the $X$-stabilizer $q_1^{(0)} \otimes \omega_1(a_2)$ as an example again, first act by $g$, then apply $\phi^{-1}$ gives
\begin{align*}
    (q_1^{(0)} \otimes \omega_1(a_2)) \cdot g &= 
    q_1^{(0)} \cdot g \otimes g^{-1} \cdot \omega_1(a_2)\\
    &= q_1^{(0)} \cdot g \otimes \omega_1(g^{-1} \cdot  a_2)\\
    & \overset{\phi^{-1}}{\mapsto} \omega_1(q_1 \cdot g \otimes g^{-1} \cdot  a_2).
\end{align*}
On the other hand, first apply $\phi^{-1}$, then act by $g$ gives
\begin{align*}
    \phi^{-1}(q_1^{(0)} \otimes \omega_1(a_2)) \cdot g &= \omega_1(q_1 \otimes a_2) \cdot g\\
    &= \omega_1((q_1 \otimes a_2)\cdot g) \\
    &= \omega_1(q_1 \cdot g \otimes g^{-1} \cdot a_2).
\end{align*}
Hence, we see that $\phi^{-1}$ indeed preserves $G$ action, and similarly for $\phi^{0}$ and $\phi^{1}$. As a result, we see that
\begin{align*}
    \hspace{-10pt}
    (\text{CSS}_1 \otimes \text{CSS}_2) 
    \oplus 
    (\text{CSS}_1^\vee \otimes \text{CSS}_2^\vee)\Big /G = \omega(\text{CSS}_1 \otimes_{\ee-\mm} \text{CSS}_2). 
\end{align*}
Using the same argument, we can show that $(\text{CSS}_1^\vee \otimes \text{CSS}_2) \oplus (\text{CSS}_1 \otimes \text{CSS}_2^\vee)$ is the doubling of CSS$_1^\vee \otimes$ CSS$_2$, the quotient of which gives $\omega(\text{CSS}_1^\vee \otimes_{\ee-\mm} \text{CSS}_2) $. Therefore, we conclude that
\begin{align*}
    \omega(\text{CSS}_1) \otimes_G \omega(\text{CSS}_2) = &
    \omega(\text{CSS}_1 \otimes_{\ee-\mm} \text{CSS}_2)\oplus
    \omega(\text{CSS}_1^\vee \otimes_{\ee-\mm} \text{CSS}_2). 
\end{align*}

We can go one step further. If there exists a $g\in G$ such that $\gamma(g) = 1$ (that is $U_g = \mathfrak{h}$), then the second line of Eq.~\eqref{eqn:double_prod} is isomorphic to the first line. To see this, pick such a $g$ with $\gamma(g) = 1$, then $x \mapsto x\cdot g$ for $x \in A_1$, $B_1$ or $Q_1$ defines an isomorphism between CSS$_1$ and CSS$_1^\vee$. It follows that the following mapping 
\begin{align*}
    & x_1 \otimes x_2 \mapsto x_1 \cdot g \otimes x_2\,,
    \hspace{25pt}
    x_i \in A_i, \, B_i,\, Q_i,
\end{align*}
defines an isomorphism CSS$_1 \otimes $ CSS$_2$ $\rightarrow$ CSS$_1^\vee \otimes $ CSS$_2$. If we further assume $g \in Z(G)$, then this mapping commutes with the $G$ action, and extends to a $G$-module isomorphism. Since the quotient commutes with the symplectic doubling, we have that
\begin{align*}
    \omega(\text{CSS}_1) \otimes_G \omega(\text{CSS}_2) \cong
    \omega^2(\text{CSS}_1 \otimes_{\ee-\mm} \text{CSS}_2).
\end{align*}
In particular, this holds when $G$ is abelian with some $g$ such that $\gamma(g) = 1$.

\section{Details for Balanced Coupled-Layer Codes}
\label{app:balancedcommute}
In this appendix, we show that the condensation terms Eq.~\eqref{eq:Atermbalanced} mutually commute, and therefore give a well-defined code switching.

First, we show that $\{\textit{\textbf{A}}(e,\tilde f) \,|\, e\in \mathcal{E}^{(\tilde f)}\}$ are mutually commuting. For simplicity, let us first take the excitations to be single Pauli operators, as in the case of tensor product. That is, we set $\{\mathcal{E}^{(\tilde f)}\} = \{P_{q_1}\,|\, q_1\in Q_1 \}$, where $P$ is the Pauli operator of type ty$(\tilde{ \bm f})$. Then the code switching terms are $\bm A(q_1, \tilde f)$.

\begin{lemma}
    The supports of $\bm A(q_1, \tilde f)$ and $\bm A(q_1', \tilde f)$ overlap in layers $\tilde q_2$ if and only if there exists a unique $g\in G$ such that $q_1' = q_1\cdot g$, and $\supp{\tilde {\bm f}} \cap \supp{\mathcal{U}_g(\tilde { \bm f})} \neq \emptyset$.
\end{lemma}
\begin{proof}
    If $q_1 = q_1'$, then the statement is trivial. Let us assume $q_1 \neq q_1'$. 
    
    \noindent $(\Longrightarrow)$ direction: if $\supp{\bm A(q_1, \tilde f)}\cap \supp{\bm A(q_1', \tilde f)} \neq \emptyset$, then there exists $q_2$, $q_2'\in \tilde f$, such that $[q_2] = [q_2'] = [\tilde q_2]$, and $q_1\cdot g_{q_2} = q_1'\cdot g_{q_2'}$. In this way, the overlap is between
    \begin{align*}
        U^\dagger_{g_{q_2}}P^{(\tilde q_2)}_{q_1\cdot g_{q_2}} U_{g_{q_2}} \in \bm A(q_1, \tilde f)
        \hspace{15pt}
        \text{and}
        \hspace{15pt}
        U^\dagger_{g_{q_2'}}
        P^{(\tilde q_2)}_{q_1'\cdot g_{q_2'}} U_{g_{q_2'}} \in \bm A(q_1', \tilde f).
    \end{align*}
   
    Since $[q_2] = [q_2']$, there exists $g\in G$ such that $q_2 = g\cdot q_2'$, therefore $\supp{\tilde{ \bm f}} \cap \supp{\mathcal{U}_g( \tilde{ \bm f})} \neq \emptyset$. Moreover, we have $g_{q_2} = gg_{q_2'}$, therefore $(q_1\cdot g)\cdot g_{q_2'} = q_1'\cdot g_{q_2'}$. Acting by $g_{q_2'}^{-1}$, we have $q_1' = q_1\cdot g$.

    \noindent $(\Longleftarrow)$ direction: if $\supp{\tilde {\bm f}} \cap \supp{\mathcal{U}_g( \tilde {\bm f})} \neq \emptyset$, then there exists $q_2$, $q_2' \in \tilde f$ such that $q_2 = g\cdot q_2'$, so that $g_{q_2} = gg_{q_2'}$. In which case, $\bm A(q_1, \tilde f)$ and $\bm A(q_1', \tilde f)$ overlap at
    \begin{align*}
        U^\dagger_{gg_{q_2'}}P^{(\tilde q_2)}_{q_1\cdot gg_{q_2'}} U_{gg_{q_2'}} \in \bm A(q_1, \tilde f)
        \hspace{15pt}
        \text{and}
        \hspace{15pt}
        U^\dagger_{g_{q_2'}}
        P^{(\tilde q_2)}_{q_1'\cdot g_{q_2'}} U_{g_{q_2'}} \in \bm A(q_1', \tilde f).
    \end{align*}
    Finally, the group element satisfying $q_1' = q_1 \cdot g$ is unique because the group action is free.
\end{proof}

Consider now the commutation relation between $\bm A(q_1, \tilde f)$ and $\bm A(q_1', \tilde f)$. If their supports have no overlap in layers $\tilde q_2$, then they commute trivially. Otherwise, by the lemma above, there exists a unique $g$, such that $q_1' = q_1 \cdot g$ and $\supp{\tilde {\bm f}} \cap \supp{\mathcal{U}_g(\tilde {\bm f})} \neq \emptyset$. Find all the pairs $\{(q_{(n)}, \,g\cdot q_{(n)}) \,|\, n = 1, \cdots, N\} \subset \tilde f\times \tilde f$. Since $\tilde {\bm f}$ commutes with $\mathcal{U}_g(\tilde {\bm f})$, they must commute on the overlap, that is, 
\begin{align*}
    \left.\tilde {\bm f}\,\cdot \mathcal{U}_g (\tilde {\bm f})\right|_{\supp{\tilde {\bm f}} \cap \supp{\mathcal{U}_g (\tilde {\bm f})}} &= \prod_{n=1}^N P_{g\cdot q_{(n)}}
    \prod_{n=1}^N U_g P_{g\cdot q_{(n)}}U_g^\dagger
    \\
    &= 
    \prod_{n=1}^N P_{g\cdot q_{(n)}}
    \left(U_g P_{g\cdot q_{(n)}}U_g^\dagger\right) 
    \\
    &=
    \left( \prod_{n=1}^N \phi_g \right)
    \prod_{n=1}^N 
    \left(
    U_g P_{g\cdot q_{(n)}}U_g^\dagger
    \right)
    P_{g\cdot q_{(n)}}
    \\
    &=
    \prod_{n=1}^N 
    \left(U_g P_{g\cdot q_{(n)}}U_g^\dagger\right)
    P_{g\cdot q_{(n)}}\\
    &= \left.\,\mathcal{U}_g (\tilde {\bm f}) \cdot \tilde{\bm  f}\right|_{\supp{\tilde {\bm f}} \cap \supp{\mathcal{U}_g (\tilde {\bm f})}},
\end{align*}
where $P_{g\cdot q_{(n)}}$ is the Pauli $P$ acting on qubit $q_{(n)}$, and $\phi_g$ is the phase by commuting $P$ and $U_g P U_g^\dagger$. In terms of this set, the overlap and commutation relation between $\bm A(q_1, \tilde f)$ and $\bm A(q_1', \tilde f)$ is
\begin{align*}
    \left.
    \bm A(q_1, \tilde f) \bm A(q_1', \tilde f)
    \right|_{\supp{\bm A(q_1, \tilde f)}\cap \supp{\bm A(q_1', \tilde f)}}
    =&
    \prod_{n=1}^N 
    \left(
    U_{gg_{q_{(n)}}}^\dagger
    P^{(\tilde q_{(n)})}_{q_1\cdot gg_{q_{(n)}}} U_{gg_{q_{(n)}}} 
    \right)
    \prod_{n=1}^N 
    \left(
    U_{g_{q_{(n)}}}^\dagger P^{(\tilde q_{(n)})}_{q_1'\cdot g_{q_{(n)}}} U_{g_{q_{(n)}}} 
    \right)\\
    =&
    \prod_{n=1}^N 
    \left(
    U_{gg_{q_{(n)}}}^\dagger
    P^{(\tilde q_{(n)})}_{q_1\cdot gg_{q_{(n)}}} U_{gg_{q_{(n)}}} 
    \right)
    \left(
    U_{g_{q_{(n)}}}^\dagger P^{(\tilde q_{(n)})}_{q_1'\cdot g_{q_{(n)}}} U_{g_{q_{(n)}}} 
    \right)\\
    =&
    \prod_{n=1}^N 
    U_{gg_{q_{(n)}}}^\dagger
    \left[
    P^{(\tilde q_{(n)})}_{q_1\cdot gg_{q_{(n)}}} 
    \left( 
    U_g 
    P^{(\tilde q_{(n)})}_{q_1'\cdot g_{q_{(n)}}}
    U_g^\dagger
    \right)
    \right]
    U_{gg_{q_{(n)}}}\\
    =& \left( \prod_{n}^N \phi_g \right)
    \prod_{n=1}^N 
    U_{gg_{q_{(n)}}}^\dagger
    \left[ 
    \left( 
    U_g 
    P^{(\tilde q_{(n)})}_{q_1'\cdot g_{q_{(n)}}}
    U_g^\dagger
    \right)
    P^{(\tilde q_{(n)})}_{q_1\cdot gg_{q_{(n)}}}
    \right]
    U_{gg_{q_{(n)}}}\\
    =&
    \prod_{n=1}^N 
    \left(
    U_{g_{q_{(n)}}}^\dagger P^{(\tilde q_{(n)})}_{q_1'\cdot g_{q_{(n)}}} U_{g_{q_{(n)}}}
    \right)
    \left(
    U_{gg_{q_{(n)}}}^\dagger
    P^{(\tilde q_{(n)})}_{q_1\cdot gg_{q_{(n)}}} U_{gg_{q_{(n)}}} 
    \right)\\
    =&
    \left.
    \bm A(q_1', \tilde f) \bm A(q_1, \tilde f)
    \right|_{\supp{\bm A(q_1, \tilde f)}\cap \supp{\bm A(q_1', \tilde f)}}.
\end{align*}
This shows $\{\bm A(q_1, \tilde f) \,|\, q_1\in Q_1\}$ is a commuting set if we choose the excitations to be single Pauli operators.

For a general set of excitations $\mathcal{E}^{(\tilde f)}$, the assumption that $\{\supp{\bm e} \,|\, e\in \mathcal{E}^{(\tilde f)}\}$ are mutually disjoint and permuted freely by $G$ essentially says they behave in the same way as single Pauli operators under the group action, so that we have a similar lemma,

\begin{lemma}
    The supports of $\bm A(e, \tilde f)$ and $\bm A(e', \tilde f)$ overlap in layers $\tilde q_2$ if and only if there exists a unique $g\in G$ such that $\supp{\bm e'} = \supp{\bm e}\cdot g$, and $\supp{\tilde {\bm f}} \cap \supp{\mathcal{U}_g(\tilde {\bm f})} \neq \emptyset$.
\end{lemma}
\begin{proof}
    In the proof above, replace $P_{q_1}$ and $P_{q_1'}$ by $\bm e$ and $\bm e'$ respectively and note that different excitations $\bm e$ and $\bm e'$ have disjoint support.
\end{proof}

Then essentially the same calculation for $\bm A(e, \tilde f)$ and $\bm A(e', \tilde f)$ goes through. If their supports have no overlap, then they commute trivially. Otherwise, there exists a unique $g$, such that $\supp{\bm e'} = \supp{\bm e} \cdot g$ and $\supp{\tilde {\bm f}} \cap \supp{\mathcal{U}_g( \tilde {\bm f})} \neq \emptyset$. Find all the pairs $\{(q_{(n)}, \,g\cdot q_{(n)}) \,|\, n = 1, \cdots, N\} \subset \tilde f\times \tilde f$. 
\begin{align*}
    \left.
    \bm A(e, \tilde f)
    \bm A(e', \tilde f)
    \right|_{\supp{\bm A(e, \tilde f)}\cap \supp{\bm A(e', \tilde f)}}
    =& \prod_{n=1}^N  \mathcal{U}_{gg_{q_{(n)}}}
    (\bm e^{(\tilde q_{(n)})})
    \prod_{n=1}^N \mathcal{U}_{g_{q_{(n)}}}(
    \bm e'^{(\tilde q_{(n)})})\\
    =&\prod_{n=1}^N  \mathcal{U}_{gg_{q_{(n)}}}(\bm e^{(\tilde q_{(n)})})
    \,\mathcal{U}_{g_{q_{(n)}}}(
    \bm e'^{(\tilde q_{(n)})})
    \\
    =&
    \prod_{n=1}^N 
    \left(
    \prod_{q_1\in e} U_{gg_{q_{(n)}}}^\dagger
    P^{(\tilde q_{(n)})}_{q_1\cdot gg_{q_{(n)}}} U_{gg_{q_{(n)}}} 
    \right)
    \left(
    \prod_{q_1'\in e'}U_{g_{q_{(n)}}}^\dagger P^{(\tilde q_{(n)})}_{q_1'\cdot g_{q_{(n)}}} U_{g_{q_{(n)}}}
    \right)\\
    =&
    \prod_{n=1}^N 
    U_{gg_{q_{(n)}}}^\dagger
    \left[
    \prod_{q_1\in e'}
    P^{(\tilde q_{(n)})}_{q_1\cdot g_{q_{(n)}}} 
    \left( 
    U_g 
    P^{(\tilde q_{(n)})}_{q_1\cdot g_{q_{(n)}}}
    U_g^\dagger
    \right)
    \right]
    U_{gg_{q_{(n)}}}\\
    =& \left( \prod_{n}^N \phi_g \right)^{|\supp{\bm e'}|}
    \prod_{n=1}^N 
    U_{gg_{q_{(n)}}}^\dagger
    \left[ \prod_{q_1\in e'}
    \left( 
    U_g 
    P^{(\tilde q_{(n)})}_{q_1\cdot g_{q_{(n)}}}
    U_g^\dagger
    \right)
    P^{(\tilde q_{(n)})}_{q_1\cdot g_{q_{(n)}}} 
    \right]
    U_{gg_{q_{(n)}}}\\
    =&
    \prod_{n=1}^N 
    \left(
    \prod_{q_1'\in e'}U_{g_{q_{(n)}}}^\dagger P^{(\tilde q_{(n)})}_{q_1'\cdot g_{q_{(n)}}} U_{g_{q_{(n)}}}
    \right)
    \left(
    \prod_{q_1\in e} U_{gg_{q_{(n)}}}^\dagger
    P^{(\tilde q_{(n)})}_{q_1\cdot gg_{q_{(n)}}} U_{gg_{q_{(n)}}} 
    \right)\\
    =&\prod_{n=1}^N 
    \mathcal{U}_{g_{(n)}}(\bm e'^{(\tilde q_{(n)})})
    \,
    \mathcal{U}_{gg_{(n)}}( \bm e^{(\tilde q_{(n)})})\\
    =&
    \left.
    \bm A(e', \tilde f)
    \bm A(e, \tilde f)
    \right|_{\supp{\bm A(e, \tilde f)}\cap \supp{\bm A(e', \tilde f)}}
\end{align*}
Where from first line to the second line, we expanded the operator 
\begin{align*}
    \mathcal{U}_g(\bm e) = \prod_{q_1\in e}U_g^\dagger P_{q_1 \cdot g} U_g
\end{align*}
From the second line to the third line, since $\supp{\bm e}
$ and $\supp{\bm e'}$ are in bijection via $\supp{\bm e'} = \supp{\bm e} \cdot g$, we have collapsed the two products into one which ranges over $q_1\in e'$. In the fourth line, $|\supp{\bm e'}|$ is the size of the support of $\bm e'$.

Now we consider the commutation relation between $\bm A(e, \tilde f )$ for different $\tilde f$. Again, we first consider the case where the excitations are single Pauli operators. This calculation is essentially identical to that of the last section. Let $P$ and $P'$ be the Pauli operator of type ty($\tilde {\bm f}$) and ty($\tilde {\bm f}'$) respectively. The excitations we choose are $\mathcal{E}^{(\tilde f)} = \{P_{q_1} \,|\, q_1\in Q_1\}$ and $\mathcal{E}^{(\tilde f')} = \{P'_{q_1} \,|\, q_1\in Q_1\}$.

\begin{lemma}
    The supports of $\bm A(q_1,\tilde f)$ and $\bm A(q_1', \tilde f')$ overlap if and only if there exists a unique $g\in G$ such that $q_1' = q_1\cdot g$, and $\supp{\tilde {\bm f}} \cap \supp{\mathcal{U}_g( \tilde {\bm f}')} \neq \emptyset$.
\end{lemma}
\begin{proof}
    The proof is the same as before. 

    \noindent ($\Longrightarrow$) direction: if $\bm A(q_1,\tilde f)$ and $\bm A(q_1', \tilde f')$ have overlap, then there must exist $q_2\in \tilde f$ and $q_2'\in \tilde f'$, such that $[q_2] = [q_2'] = [\tilde q_2]$ (so that $q_2 = g\cdot q_2'$ for some $g\in G$), and $q_1\cdot g_{q_2} = q_1'\cdot g_{q_2'}$. In which case the overlap is between
    \begin{align*}
        U_{g_{q_2}}^\dagger  P^{(\tilde q_2)}_{q_1 \cdot g_{q_2}}
        U_{g_{q_2}} \in \bm A(q_1, \tilde f)
        \hspace{15pt}
        \text{and}
        \hspace{15pt}
        U_{g_{q'_2}}^\dagger  P'^{(\tilde q_2)}_{q'_1 \cdot g_{q'_2}}
        U_{g_{q'_2}} \in \bm A(q'_1, \tilde f').
    \end{align*}
    Since $g\cdot q_2' \in \supp{\mathcal{U}_g( \tilde {\bm f'})} \cap \supp{\tilde {\bm f}}$, we see the intersection is not empty. Since $g_{q_2} = gg_{q_2'}$, we see $q_1\cdot g = q_1'$.

    \noindent ($\Longleftarrow$) direction: Pick $q_2 \in \supp{\tilde {\bm f}} \cap \supp{\mathcal{U}_g(\tilde {\bm f'})}$, so that there exists $q_2' \in \tilde f'$ and $q_2 = g\cdot q_2'$. This gives $g_{q_2} = g g_{q'_2}$. We see there is an overlap between $\bm A(q_1,\tilde f)$ and $\bm A(q_1', \tilde f')$ on
    \begin{align*}
        U_{gg_{q'_2}}^\dagger  P^{(\tilde q_2)}_{q_1 \cdot gg_{q'_2}}
        U_{gg_{q'_2}} \in \bm A(q_1, \tilde f)
        \hspace{15pt}
        \text{and}
        \hspace{15pt}
        U_{g_{q'_2}}^\dagger  {P'}^{(\tilde q_2)}_{q'_1 \cdot g_{q'_2}}
        U_{g_{q'_2}} \in \bm A(q'_1, \tilde f'),
    \end{align*} where
    $g$ is again unique because the action is free.
\end{proof}

\noindent Now we consider the commutation relation between $\bm A(q_1, \tilde f)$ and $\bm A(q'_1, \tilde f')$. If they have no overlap, then they commute trivially. Otherwise, there exists $g\in G$ such that $q_1' = q_1\cdot g$, and $\supp{\tilde {\bm f}} \cap \supp{\mathcal{U}_g( \tilde {\bm f'})} \neq \emptyset$. Find all pairs $\{(q_{(n)}, \,g\cdot q_{(n)}) \,|\, n = 1, \cdots, N\} \subset \tilde f' \times \tilde f$. Since $\tilde {\bm f}$ commutes with $\mathcal{U}_g( \tilde {\bm f'})$, they must commute on their overlap, that is,
\begin{align*}
    \left.\tilde {\bm f}\,
    \cdot
    \mathcal{U}_g(\tilde{\bm f}')\right|_{\supp{\tilde {\bm f}} \cap \supp{\mathcal{U}_g(\tilde{\bm f}')}}&=\prod_{n=1}^N P_{g\cdot q_{(n)}}
    \prod_{n=1}^N  U_g P'_{g\cdot q_{(n)}}U_g^\dagger\\
    &=\prod_{n=1}^N P_{g\cdot q_{(n)}}
    \left(U_g P'_{g\cdot q_{(n)}}U_g^\dagger\right)\\
    &= \left( \prod_{n=1}^N \varphi_g \right)
    \prod_{n=1}^N 
    \left(
    U_g P'_{g\cdot q_{(n)}}U_g^\dagger
    \right)
    P_{g\cdot q_{(n)}}\\
    &=
    \prod_{n=1}^N 
    \left(U_g P'_{g\cdot q_{(n)}}U_g^\dagger\right)
    P_{g\cdot q_{(n)}}\\
    &= \left. \mathcal{U}_g(\tilde{\bm f}') \, \cdot \tilde {\bm f} \right|_{\supp{\tilde {\bm f}} \cap \supp{\mathcal{U}_g(\tilde{\bm f}')}},
\end{align*}
where $\varphi_g$ is the phase by commuting $P$ and $U_g P' U_g^\dagger$. Using this, the commutation between $\bm A(q_1, \tilde f)$ and $\bm A(q'_1, \tilde f')$ on the intersection of their support is
\begin{align*}
    \left.
    \bm A(q_1, \tilde f)
    \bm A(q_1', \tilde f')
    \right|_{\supp{\bm A(q_1, \tilde f)}\cap \supp{\bm A(q_1', \tilde f')}}
    =&\prod_{n=1}^N 
    \left(
    U_{gg_{q_{(n)}}}^\dagger
    P^{(\tilde q_{(n)})}_{q_1\cdot gg_{q_{(n)}}} U_{gg_{q_{(n)}}} 
    \right)
    \prod_{n=1}^N
    \left(
    U_{g_{q_{(n)}}}^\dagger {P'}^{(\tilde q_{(n)})}_{q_1'\cdot g_{q_{(n)}}} U_{g_{q_{(n)}}}
    \right)\\
    =& \prod_{n=1}^N 
    \left(
    U_{gg_{q_{(n)}}}^\dagger
    P^{(\tilde q_{(n)})}_{q_1\cdot gg_{q_{(n)}}} U_{gg_{q_{(n)}}} 
    \right)
    \left(
    U_{g_{q_{(n)}}}^\dagger {P'}^{(\tilde q_{(n)})}_{q_1'\cdot g_{q_{(n)}}} U_{g_{q_{(n)}}}
    \right)\\
    =&
    \prod_{n=1}^N 
    U_{gg_{q_{(n)}}}^\dagger
    \left[
    P^{(\tilde q_{(n)})}_{q_1\cdot gg_{q_{(n)}}} 
    \left( 
    U_g 
    {P'}^{(\tilde q_{(n)})}_{q_1'\cdot g_{q_{(n)}}}
    U_g^\dagger
    \right)
    \right]
    U_{gg_{q_{(n)}}}\\
    =& \left( \prod_{n}^N \varphi_g \right)
    \prod_{n=1}^N 
    U_{gg_{q_{(n)}}}^\dagger
    \left[ 
    \left( 
    U_g 
    {P'}^{(\tilde q_{(n)})}_{q_1'\cdot g_{q_{(n)}}}
    U_g^\dagger
    \right)
    P^{(\tilde q_{(n)})}_{q_1\cdot gg_{q_{(n)}}}
    \right]
    U_{gg_{q_{(n)}}}\\
    =&
    \prod_{n=1}^N 
    \left(
    U_{g_{q_{(n)}}}^\dagger {P'}^{(\tilde q_{(n)})}_{q_1'\cdot g_{q_{(n)}}} U_{g_{q_{(n)}}}
    \right)
    \left(
    U_{gg_{q_{(n)}}}^\dagger
    P^{(\tilde q_{(n)})}_{q_1\cdot gg_{q_{(n)}}} U_{gg_{q_{(n)}}} 
    \right)\\
    =&
    \left.
    \bm A(q_1', \tilde f')
    \bm A(q_1, \tilde f)
    \right|_{\supp{\bm A(q_1, \tilde f)}\cap \supp{\bm A(q_1', \tilde f')}}.
\end{align*}
This shows when the excitations are chosen to be single Pauli operators, $\bm A(q_1, \tilde f)$ commutes with $\bm A(q'_1, \tilde f')$.

Next we consider general excitations $\bm A(e, \tilde f)$ and $\bm A(e', \tilde f')$. This case is slightly different from previous ones, so we state the lemma slightly differently,
\begin{lemma}
    The pairs of terms in $\bm A(e, \tilde f)$ and $\bm A(e', \tilde f')$ that have non-empty overlapping supports are in one-to-one correspondence with pairs $\mathcal{K} :=\left\{(q_2, \, g\cdot q_2) \in \tilde f' \times \tilde f \,|\, \supp{\bm e'} \cap \supp{\mathcal{U}_g(\bm e)} \neq \emptyset \right\}$. Moreover, the overlapping regions between each pairs are mutually disjoint.
\end{lemma}
\begin{proof}
    If for $q_2 \in \tilde f$ and $q'_2 \in \tilde f'$, such that $\mathcal{U}_{g_{q_2}}(\bm e^{(\tilde q_2)})$ and $\mathcal{U}_{g_{q'_2}}(\bm e'^{(\tilde q'_2)})$ have overlapping supports, then we must have $[q_2] = [q'_2]$, and $\supp{\mathcal{U}_{g_{q_2}}(\bm e)} \cap \supp{\mathcal{U}_{g_{q'_2}}(\bm e') } \neq \emptyset$. Thus there exists a unique $g\in G$ (since the group action is free), such that $q_2 = g\cdot q'_2$, and $g_{q_2} = gg_{q'_2}$. Since the group action on qubits is free, $\supp{\mathcal{U}_{g g_{q'_2}}(\bm e)} \cap \supp{\mathcal{U}_{g_{q'_2}}(\bm e')} \neq \emptyset$ implies $\supp{\mathcal{U}_g(\bm e) } \cap \supp{\bm e'} \neq \emptyset$. Therefore we have $(q'_2, \, g\cdot q'_2) \in \mathcal{K}$.

    Conversely, given $(q_2, \, g\cdot q_2) \in \mathcal{K}$, the term $\mathcal{U}_{gg_{q_2} }(\bm e^{(\tilde q_2)})   \in \bm A(e, \tilde f)$ and $\mathcal{U}_{g_{q_2}}(\bm e'^{(\tilde q_2)}) \in \bm A(e', \tilde f')$ have non-empty overlap.

    These two maps are inverses of each other, hence they are bijections.

    Finally, in a fixed layer $\tilde q_2$, the overlap $\supp{\mathcal{U}_{g_{q_2}}(\bm e)} \cap \supp{\mathcal{U}_{g_{q'_2}}(\bm e') }$ and $\supp{\mathcal{U}_{g_{q''_2}}(\bm e)  } \cap \supp{\mathcal{U}_{g_{q'''_2}}(\bm e')}$ are disjoint because the supports of $\bm e \in \mathcal{E}^{(\tilde f)}$ or $\mathcal{E}^{(\tilde f')}$ are mutually disjoint, and the group acts freely on the supports.
\end{proof}

\noindent Note the difference is that there can be multiple $g$ such that the supports of $\bm e'$ and $\mathcal{U}_g(\bm e)$ have a non-empty overlap. Now let us assume there are only two group elements participate in $\mathcal{K}$, namely $g$ and $h$. The case where more than two group elements participate is completely analogous. Observe for $(q_2, \, g\cdot q_2)\in \mathcal{K}$ or $(q_2, \, h\cdot q_2)\in \mathcal{K}$, we have respectively
\begin{align*}
    &\supp{\mathcal{U}_{gg_{q_2}}(\bm e)} \cap \supp{\mathcal{U}_{g_{q_2}}(\bm e')} = \left(\supp{\mathcal{U}_g(\bm e)} \cap \supp{\bm e'} \right)\cdot g_{q_2},\\
    &\supp{\mathcal{U}_{hg_{q_2}}(\bm e) } \cap \supp{\mathcal{U}_{g_{q_2}}(\bm e')} = \left(\supp{\mathcal{U}_h(\bm e)} \cap \supp{\bm e'} \right)\cdot g_{q_2}.
\end{align*}
All of which are mutually disjoint by the lemma. Let $\mathcal{K}_g := \{q_2 \in \tilde f' \,|\, (q_2, \, g\cdot q_2) \in \mathcal{K}\}$, and similarly $\mathcal{K}_h := \{q_2 \in \tilde f' \,|\, (q_2, \, h\cdot q_2) \in \mathcal{K}\}$. They are not necessarily disjoint! Then we have 
\begin{align*}
    \left.
    \bm A(e, \tilde f)
    \right|_{\supp{\bm A(e, \tilde f)}\cap \supp{\bm A(e', \tilde f')}}
    =&
    \prod_{q_2\in \mathcal{K}_g} 
    \prod_{q_1\in \supp{\mathcal{U}_g(\bm e)} \cap \supp{\bm e'} } 
    \left(
    U_{gg_{q_2}}^\dagger
    P^{(\tilde q_2)}_{q_1\cdot g_{q_2}} U_{gg_{q_2}} 
    \right)\\
    &\hspace{35pt}
    \prod_{q_2\in \mathcal{K}_h} 
    \prod_{q_1\in \supp{\mathcal{U}_h(\bm e)} \cap \supp{\bm e'} } 
    \left(
    U_{hg_{q_2}}^\dagger
    P^{(\tilde q_2)}_{q_1\cdot g_{q_2}} U_{hg_{q_2}} 
    \right)\\
    \left.
    \bm A(e', \tilde f')
    \right|_{\supp{\bm A(e, \tilde f)}\cap \supp{\bm A(e', \tilde f')}}
    =&
    \prod_{q_2\in \mathcal{K}_g} 
    \prod_{q_1\in \supp{\mathcal{U}_g(\bm e)} \cap \supp{\bm e'} }
    \left(U_{g_{q_2}}^\dagger {P'}^{(\tilde q_2)}_{q_1\cdot g_{q_2}} U_{g_{q_2}}
    \right)\\
    &\hspace{35pt}
    \prod_{q_2\in \mathcal{K}_h} 
    \prod_{q_1\in \supp{\mathcal{U}_h(\bm e)} \cap \supp{\bm e'} } 
    \left(U_{g_{q_2}}^\dagger {P'}^{(\tilde q_2)}_{q_1\cdot g_{q_2}} U_{g_{q_2}}
    \right).
\end{align*}
Since the overlaps are pairwise disjoint, we can write their product as
\begin{align*}
    \left.
    \bm A(e, \tilde f)
    \bm A(e', \tilde f')
    \right|_{\supp{\bm A(e, \tilde f)}
    \cap 
    \supp{\bm A(e', \tilde f')}}
    =&
    \prod_{q_2\in \mathcal{K}_g} 
    \prod_{q_1\in \supp{\mathcal{U}_g(\bm e)} \cap \supp{\bm e'} } 
    \left(
    U_{gg_{q_2}}^\dagger
    P^{(\tilde q_2)}_{q_1\cdot g_{q_2}} U_{gg_{q_2}} 
    \right)
    \left(U_{g_{q_2}}^\dagger {P'}^{(\tilde q_2)}_{q_1\cdot g_{q_2}} U_{g_{q_2}}
    \right)\\
    & \hspace{35pt}
    \prod_{q_2\in \mathcal{K}_h} 
    \prod_{q_1\in \supp{\mathcal{U}_h(\bm e)} \cap \supp{\bm e'} } 
    \left(
    U_{hg_{q_2}}^\dagger
    P^{(\tilde q_2)}_{q_1\cdot g_{q_2}} U_{hg_{q_2}} 
    \right)
    \left(U_{g_{q_2}}^\dagger {P'}^{(\tilde q_2)}_{q_1\cdot g_{q_2}} U_{g_{q_2}}
    \right).
\end{align*}
The first line we have the overlap between $\mathcal{U}_{gg_{q_2}}(\bm e)$ and $\mathcal{U}_{g_{q_2}}(\bm e')$ for $q_2\in \mathcal{K}_g$. More specifically,
\begin{align*} 
    \left(\prod_{q_1\in \supp{\mathcal{U}_g(\bm e)} \cap \supp{\bm e'} } U_{gg_{q_2}}^\dagger
    P^{(\tilde q_2)}_{q_1\cdot g_{q_2}} U_{gg_{q_2}}\right)
    \in \bm A(e, \tilde f),
    \hspace{35pt}
    \left(\prod_{q_1\in \supp{\mathcal{U}_g(\bm e)} \cap \supp{\bm e'}}U_{g_{q_2}}^\dagger P^{(\tilde q_2)}_{q_1\cdot g_{q_2}} U_{g_{q_2}}\right)
    \in \bm A(e', \tilde f'),
\end{align*}
and similarly, the second line is the overlap between $\mathcal{U}_{hg_{q_2}}(\bm e)$ and $\mathcal{U}_{g_{q_2}}(\bm e')$ for $q_2\in \mathcal{K}_h$. Now we proceed similarly as before,
\begin{align*}
    \left.
    \bm A(e, \tilde f)
    \bm A(e', \tilde f')
    \right|_{\supp{\bm A(e, \tilde f)}\cap \supp{\bm A(e', \tilde f')}}
    =&\prod_{q_2\in \mathcal{K}_g} 
    U_{gg_{q_2}}^\dagger
    \left[
    \prod_{q_1\in \supp{\mathcal{U}_g(\bm e)} \cap \supp{\bm e'} }
    P^{(\tilde q_2)}_{q_1\cdot g_{q_2}} 
    \left( 
    U_g 
    {P'}^{(\tilde q_2)}_{q_1\cdot g_{q_2}}
    U_g^\dagger
    \right)
    \right]
    U_{gg_{q_2}}\\
    &\hspace{35pt}
    \prod_{q_2\in \mathcal{K}_h} 
    U_{hg_{q_2}}^\dagger
    \left[
    \prod_{q_1\in \supp{\mathcal{U}_h(\bm e)} \cap \supp{\bm e'} }
    P^{(\tilde q_2)}_{q_1\cdot g_{q_2}} 
    \left( 
    U_h 
    {P'}^{(\tilde q_2)}_{q_1\cdot g_{q_2}}
    U_h^\dagger
    \right)
    \right]
    U_{hg_{q_2}}\\
    =& \left( \prod_{q_2\in\mathcal{K}_g} \varphi_g \right)^{|\supp{\mathcal{U}_g(\bm e)} \cap \supp{\bm e'}|}
    \left( \prod_{q_2\in\mathcal{K}_h} \varphi_h \right)^{|\supp{\mathcal{U}_h(\bm e)} \cap \supp{\bm e'}|}\\
    & \hspace{35pt}
    \prod_{q_2\in \mathcal{K}_g} 
    U_{gg_{q_2}}^\dagger
    \left[
    \prod_{q_1\in \supp{\mathcal{U}_g(\bm e) } \cap \supp{\bm e'} } 
    \left( 
    U_g 
    {P'}^{(\tilde q_2)}_{q_1\cdot g_{q_2}}
    U_g^\dagger
    \right)
    P^{(\tilde q_2)}_{q_1\cdot g_{q_2}}
    \right]
    U_{gg_{q_2}}\\
    &\hspace{35pt}
    \prod_{q_2\in \mathcal{K}_h} 
    U_{hg_{q_2}}^\dagger
    \left[
    \prod_{q_1\in \supp{\mathcal{U}_h(\bm e)} \cap \supp{\bm e'} }
    \left( 
    U_h 
    {P'}^{(\tilde q_2)}_{q_1\cdot g_{q_2}}
    U_h^\dagger
    \right)
    P^{(\tilde q_2)}_{q_1\cdot g_{q_2}} 
    \right]
    U_{hg_{q_2}}.
\end{align*}
Now we use that fact that $\tilde {\bm f}$ commutes with $\mathcal{U}_g(\tilde {\bm f'})$. Once we fix $g$ such that $\supp{\bm e'} \cap \supp{\mathcal{U}_g(\bm e)} \neq \emptyset$, $\mathcal{K}_g$ is precisely the overlap between $\tilde {\bm f}'$ and $\mathcal{U}_{g^{-1}}(\tilde{\bm f})$, and we have
\begin{align*}
    \prod_{q_2\in \mathcal{K}_g} P_{g\cdot q_2}
    \left(U_g P'_{g\cdot q_2}U_g^\dagger\right) 
    =
    \left( \prod_{q_2\in \mathcal{K}_g} \varphi_g \right)
    \prod_{q_2\in \mathcal{K}_g} 
    \left(
    U_g P'_{g\cdot q_2}U_g^\dagger
    \right)
    P_{g\cdot q_2}
    =
    \prod_{q_2\in \mathcal{K}_g} 
    \left(U_g P'_{g\cdot q_2}U_g^\dagger\right)
    P_{g\cdot q_2},
\end{align*}
and similarly that $\tilde {\bm f}$ commutes with $\mathcal{U}_h(\tilde {\bm f'})$, we have
\begin{align*}
    \prod_{q_2\in \mathcal{K}_g} \varphi_g = \prod_{q_2\in \mathcal{K}_h} \varphi_h = 1.
\end{align*}
Therefore we arrive at,
\begin{align*}
    \left.
    \bm A(e, \tilde f)
    \bm A(e', \tilde f')
    \right|_{\supp{\bm A(e, \tilde f)}
    \cap
    \supp{\bm A(e', \tilde f')}}
    =&
    \prod_{q_2\in \mathcal{K}_g} 
    U_{gg_{q_2}}^\dagger
    \left[
    \prod_{q_1\in \supp{\mathcal{U}_g(\bm e)} \cap \supp{\bm e'} } 
    \left( 
    U_g 
    {P'}^{(\tilde q_2)}_{q_1\cdot g_{q_2}}
    U_g^\dagger
    \right)
    P^{(\tilde q_2)}_{q_1\cdot g_{q_2}}
    \right]
    U_{gg_{q_2}}\\
    &\hspace{35pt}
    \prod_{q_2\in \mathcal{K}_h} 
    U_{hg_{q_2}}^\dagger
    \left[
    \prod_{q_1\in \supp{\mathcal{U}_h(\bm e)} \cap \supp{\bm e'} }
    \left( 
    U_h 
    {P'}^{(\tilde q_2)}_{q_1\cdot g_{q_2}}
    U_h^\dagger
    \right)
    P^{(\tilde q_2)}_{q_1\cdot g_{q_2}} 
    \right]
    U_{hg_{q_2}}\\
    =&
    \prod_{q_2\in \mathcal{K}_g} 
    \prod_{q_1\in \supp{\mathcal{U}_g(\bm e) } \cap \supp{\bm e'} } 
    \left(U_{g_{q_2}}^\dagger {P'}^{(\tilde q_2)}_{q_1\cdot g_{q_2}} U_{g_{q_2}}
    \right)
    \left(
    U_{gg_{q_2}}^\dagger
    P^{(\tilde q_2)}_{q_1\cdot g_{q_2}} U_{gg_{q_2}} 
    \right)\\
    & \hspace{35pt}
    \prod_{q_2\in \mathcal{K}_h} 
    \prod_{q_1\in \supp{\mathcal{U}_h(\bm e)} \cap \supp{\bm e'} } 
    \left(U_{g_{q_2}}^\dagger {P'}^{(\tilde q_2)}_{q_1\cdot g_{q_2}} U_{g_{q_2}}
    \right)
    \left(
    U_{hg_{q_2}}^\dagger
    P^{(\tilde q_2)}_{q_1\cdot g_{q_2}} U_{hg_{q_2}} 
    \right)\\
    =& \left.
    \bm A(e', \tilde f')
    \bm A(e, \tilde f)
    \right|_{\supp{\bm A(e, \tilde f)}
    \cap
    \supp{\bm A(e', \tilde f')}}.
\end{align*}
As mentioned before, generalizing to more than two group elements in $G$ participating in $\mathcal{K}$ is straightforward. We have shown that $\bm A(e, \tilde f)$ and $\bm A(e', \tilde f')$ commute.

\end{document}